\documentclass[11pt]{article}

\usepackage[usenames,dvipsnames]{xcolor}
\definecolor{Gred}{RGB}{219, 50, 54}
\definecolor{ToCgreen}{RGB}{0, 128, 0}

\usepackage[margin=1in]{geometry}
\usepackage[T1]{fontenc}

\usepackage[scale=0.97]{XCharter} 

\usepackage[libertine,bigdelims,vvarbb,scaled=1.05]{newtxmath} 

\usepackage{babel}
\usepackage[spacing=true,kerning=true,babel=true,tracking=true]{microtype}

\DeclareMathAlphabet{\pazocal}{OMS}{zplm}{m}{n} 
\renewcommand{\mathcal}[1]{\pazocal{#1}}

\usepackage{makecell}
\usepackage{dsfont}

\usepackage{multirow}
\usepackage{amsmath,amsthm}
\usepackage{booktabs}   % for \toprule, \midrule, \bottomrule
\usepackage{bm}
\usepackage{bbm}
\usepackage{textgreek}
\usepackage{mathtools}
\usepackage{enumitem}
\usepackage[numbers,comma,sort&compress]{natbib}
\usepackage{authblk}
\usepackage{graphicx}
\usepackage[font=small]{caption}
\usepackage[labelformat=simple]{subcaption}

\usepackage{float}
\usepackage[linesnumbered,ruled,vlined]{algorithm2e}
\SetKwInput{KwInput}{Input}
\SetKwInput{KwOutput}{Output}
\SetKwInOut{Promise}{Promise}
\SetKwInput{Goal}{Goal}
\SetKwProg{Fn}{function}{}{}
\SetKwFor{RepTimes}{repeat}{times}{}
\SetKwFunction{Ver}{Verify}
\SetKwFunction{Prep}{PrepareState}
\usepackage{algorithmic}

\usepackage{physics}
\usepackage{footnote}
\usepackage{xcolor}
\usepackage{amsfonts}
\usepackage{mathrsfs}
\usepackage{bbm}
\usepackage{braket}
\usepackage[colorlinks]{hyperref}
\usepackage{cleveref}
\hypersetup{
      colorlinks=true,
  citecolor=ToCgreen,
  linkcolor=Sepia,
  filecolor=Gred,
  urlcolor=Gred
  }
\usepackage{placeins}
\usepackage{comment}
\usepackage{textgreek}

\newtheorem{theorem}{Theorem}[section]
\newtheorem{fact}[theorem]{Fact} 
\newtheorem{definition}[theorem]{Definition}
\newtheorem{lemma}[theorem]{Lemma}

\newtheorem{proposition}[theorem]{Proposition}
\newtheorem{conjecture}{Conjecture}
\newtheorem{example}[theorem]{Example}
\newtheorem{corollary}[theorem]{Corollary}
\newtheorem{remark}[theorem]{Remark}
\newtheorem{problem}[theorem]{Problem}
 
\newtheorem{condition}[theorem]{Condition} 
\newtheorem{observation}[theorem]{Observation}

\newcommand{\algo}[1]{\hyperref[algo:#1]{Algorithm~\ref*{algo:#1}}}

\newcommand{\vect}[1]{\ensuremath{\bm{#1}}}
\newcommand{\x}{\ensuremath{\bm{x}}}

\newcommand{\C}{\mathbb{C}}
\newcommand{\N}{\mathbb{N}}
\newcommand{\Z}{\mathbb{Z}}
\newcommand{\R}{\mathbb{R}}

\newcommand{\E}{\mathbb{E}}
\renewcommand{\Pr}{\mathrm{Pr}}

\newcommand{\cN}{\mathcal{N}}

\newcommand{\cT}{\mathcal{T}}

\newcommand{\cP}{\mathcal{P}}
\newcommand{\cD}{\mathcal{D}}
\newcommand{\cI}{\mathcal{I}}
\newcommand{\cS}{\mathcal{S}}

\newcommand{\bx}{{\bm{x}}}
\newcommand{\Id}{I}
\renewcommand{\C}{\mathbb{C}}
\newcommand{\bone}{\mathds{1}}
\renewcommand{\i}{\mathrm{i}}
\renewcommand{\exp}{\mathrm{exp}}
\newcommand{\bW}{\mathbf{W}}
\newcommand{\supp}{\mathrm{supp}}
\newcommand{\bi}{\mathbf{i}}
\newcommand{\Deg}{\mathsf{D}}

\newcommand{\polylog}{{\rm polylog}}

\DeclareMathOperator{\poly}{poly}
\DeclareMathOperator{\Var}{Var}

\renewcommand{\tr}{\mathrm{tr}}
\renewcommand{\Tr}{\mathrm{tr}}
\newcommand{\SWAP}{\textnormal{SWAP}}

\renewcommand{\x}{\vect{x}}

\newcommand{\0}{\mathbf{0}}
\renewcommand{\emptyset}{\varnothing}
\def\Tr{\operatorname{tr}}\def\:{\hbox{\bf:}}

\newcommand{\tO}{\tilde{O}}
\newcommand{\tOmega}{\tilde{\Omega}}
\newcommand{\tTheta}{\tilde{\Theta}}
\newcommand{\Mod}[1]{{\ (\mathrm{mod}\ #1)}}
\newcommand{\muh}{\mu_{\text{Haar}}}
\newcommand{\hgue}{H_{\text{GUE}}}
\newcommand{\rhomm}{\rho_{\text{mm}}}

\renewcommand{\epsilon}{\varepsilon}
\newcommand{\wh}{\widehat}
\newcommand{\Adv}[1]{\chi_{\le #1}}
\allowdisplaybreaks

\makeatletter
\renewcommand{\paragraph}{%
  \@startsection{paragraph}{4}%
  {\z@}{1.5ex \@plus 1ex \@minus .2ex}{-1em}%
  {\normalfont\normalsize\bfseries}%
}
\makeatother

\begin{document}

%%%%%%%%%%%%%%%%%%%%%%%%%%%%%%%%%%%%%%%%%%%%%%%%%%%%%%%%%%%%%%%%%%%%%%%%%%%%%%%%%%%%%%%%%%%%%%%%%%%%%%%%%%%%%%

\title{Information-Computation Gaps in Quantum Learning via Low-Degree Likelihood}
\author{
Sitan Chen
\thanks{SEAS, Harvard University. Email: \href{mailto:sitan@seas.harvard.edu}{sitan@seas.harvard.edu}.}
\qquad
Weiyuan Gong
\thanks{SEAS, Harvard University. Email: \href{mailto:wgong@g.harvard.edu}{wgong@g.harvard.edu}.}
\qquad
Jonas Haferkamp
\thanks{Saarland University. Email: \href{mailto: jhaferkamp42@gmail.com}{jhaferkamp42@gmail.com}.}
\qquad
Yihui Quek
\thanks{Department of Mathematics, Massachusetts Institute of Technology. Email: \href{mailto:yquek@mit.edu}{yquek@mit.edu}.}
}

\date{\today}
\maketitle
\pagenumbering{gobble}

\begin{abstract}
    In a variety of physically relevant settings for learning from quantum data, there is an established recipe for measuring polynomially many copies of that data such that the resulting measurement readouts contain enough information to reconstruct the underlying system. Yet designing protocols that can \emph{computationally efficiently} extract that information remains largely an art, and there are important cases where we believe this to be impossible, that is, where there is an \emph{information-computation gap}. While there is a large array of tools in the classical literature for giving evidence for average-case hardness of statistical inference problems, the corresponding tools in the quantum literature are far more limited. 
    
    One such framework in the classical literature, the \emph{low-degree method}, makes predictions about hardness of inference problems based on the failure of estimators given by low-degree polynomials. In this work, we extend this framework to the quantum setting and show a number of new information-computation gaps for quantum learning. 
    
    \begin{itemize}[leftmargin=*]
        \item We establish a general connection between state designs and low-degree hardness. We use this to obtain the first information-computation gaps for learning Gibbs states of random, sparse, non-local Hamiltonians. We also use it to prove hardness for learning random shallow quantum circuit states in a challenging model where states can be measured in \emph{adaptively} chosen bases. To our knowledge, the ability to model adaptivity within the low-degree framework was open even in classical settings. In addition, we also obtain a low-degree hardness result for quantum error mitigation against strategies with single-qubit measurements.
        \item We define a new quantum generalization of the planted biclique problem and identify the threshold at which this problem becomes computationally hard for protocols that perform local measurements. Interestingly, the complexity landscape for this problem shifts when going from local measurements to more entangled single-copy measurements.
        \item We show average-case hardness for the ``standard'' variant of Learning Stabilizers with Noise~\cite{poremba2024learning} and for agnostically learning product states~\cite{bakshi2024learninga}.
    \end{itemize}
\end{abstract}

\clearpage

\tableofcontents
\clearpage

%%%%%%%%%%%%%%%%%%%%%%%%%%%%%%%%%%%%%%%%%%%%%%%%%%%%%%%%%%%%%%%%%%%%%%%%%%%%%%%%%%%%%%%%%%%%%%%%%%%%%%%%%%%%%%

\newpage
\pagenumbering{arabic}
\section{Introduction}

Given access to quantum data, how efficiently can we learn a classical description of the underlying system that generates that data? Motivated by the exciting prospect of using quantum computers to learn about the physical world, there has been a recent surge of interest in proving upper and lower bounds for this problem. The gold standard is a learning protocol with the following three properties:
\begin{enumerate}[leftmargin=*,topsep=0.5em,itemsep=0pt,label=(\Alph*)]
    \item {\bf Low statistical complexity}: Only a polynomial amount of quantum data (e.g., copies) is needed
    \item {\bf Efficient quantum processing}: The measurements that are performed on the data are bounded in complexity, e.g., only a few copies are measured at a time, and each measurement can be implemented by a shallow quantum circuit.
    \item {\bf Efficient classical post-processing}: A $\mathrm{poly}(n)$-time classical algorithm takes in the readouts, i.e., the measurement outcomes, and outputs a classical description of the underlying system.\footnote{Ref.~\cite{king2024triply} put forth a similar three-part taxonomy, which they refer to as ``triple efficiency,'' specifically in the context of shadow tomography.}
\end{enumerate}
The statistical complexity for a wide range of quantum learning problems is by now very well-understood~\cite{odonnell2015quantum,odonnell2016efficient,haah2016sample,chen2022tight}. There has also been substantial progress on quantifying how the efficiency of quantum processing impacts statistical complexity, for instance in the form of tradeoffs between the number of copies needed and various quantum resources like number of ancilla, locality of the measurements, etc~\cite{king2024triply,chen2024optimal,chen2022exponential,aharonov2022quantum,huang2021information}. {\em Yet to date, we understand surprisingly little about how these factors interact with axis (C). }

\paragraph{Information-computation gaps.} On the positive side, the community has identified a number of learning problems that admit protocols which are efficient along all three axes, like tomography of Gibbs states~\cite{haah2022optimal,bakshi2024learning}, shallow circuit states~\cite{huang2024learning}, stabilizer states~\cite{aaronson2008identifying,montanaro2017learning}, and matrix product states~\cite{cramer2010efficient}. 

In these settings, it is straightforward to achieve axes (A) and (B). Indeed, physically relevant states are typically described by polynomially many degrees of freedom, so with a polynomial amount of quantum data there is in principle enough quantum information to learn the underlying system. And in many cases of interest this quantum information can be processed with inexpensive measurements~\cite{huang2020predicting} (see Remark~\ref{remark:info_theoretic}) into classical data that is information-theoretically sufficient for learning. As such, the core difficulty behind achieving the aforementioned results lies in axis (C).

On the negative side, there are situations where even though this classical data is information-theoretically sufficient, no polynomial-time classical algorithm can harness it to learn the underlying system. In the classical literature on statistical inference, this is called an \emph{information-computation gap}. 

Pseudorandom quantum objects like pseudorandom states~\cite{ji2018pseudorandom} and unitaries~\cite{ma2024construct} already provide one such example. Given $\mathrm{poly}(n)$ copies of a state which is promised to be either pseudorandom or Haar-random, if one measures each copy in the computational basis, the resulting classical readouts contain enough information to tell apart these two cases, yet no polynomial-time algorithm can do so. Apart from reducing from such cryptographic contexts however, there are few alternative approaches for establishing \--- or even predicting \--- the onset of information-computation gaps. In this work we ask:
\begin{center} 
    {\em Is there a general framework for identifying information-computation gaps in quantum learning?}
\end{center}

\paragraph{The low-degree conjecture.} We first observe that this state of affairs is in stark contrast to the one in classical learning theory, where numerous frameworks exist for giving evidence for information-computation gaps. Most of these frameworks operate by selecting a family of estimators (e.g., statistical query algorithms, belief propagation, Markov Chain Monte Carlo) that contain the most powerful known algorithms for a large class of problems, and showing that no algorithm from this family can solve the problem at hand. One framework that has proven especially successful at accurately predicting IC gaps is the \emph{low-degree method}~\cite{barak2019nearly,hopkins2017efficient,hopkins2017power,hopkins2018statistical}. Roughly, the family of estimators considered here is the set of all low-degree polynomials in the data (i.e., samples from some probability distribution), and the given learning problem is deemed to be computationally hard if no such low-degree polynomial estimator succeeds (see Section~\ref{sec:basic_low_degree} for a detailed overview). Polynomial estimators are expressive enough to capture a number of algorithmic recipes like spectral methods, message passing, and moment methods and are also closely related to sum-of-squares algorithms for inference. With few exceptions, this framework has obtained the ``right'' (according to widely believed conjectures) predictions for a slew of problems in statistical inference, notable examples of which include planted clique~\cite{barak2019nearly}, community detection~\cite{hopkins2017efficient}, tensor PCA~\cite{hopkins2017power}, sparse PCA~\cite{ding2024subexponential}, and combinatorial group testing~\cite{coja2022statistical}.

In Section~\ref{sec:lowdegree_examples}, we observe that a wide swath of known algorithms in quantum learning can also be realized as low-degree polynomial estimators, where the classical data that these estimators take as input is given by readouts from simple measurements of the quantum data. Motivated by this, in this work we extend the low-degree method to the quantum setting and establish new information-computation gaps in quantum learning.

\subsection{Our contributions}

We begin by informally describing the quantum low-degree formalism that we develop in this work, deferring a formal treatment to Section~\ref{sec:general}. Throughout, we will focus on \emph{hypothesis testing} problems. Here, one is given polynomially many copies of a state $\rho$ and would like to distinguish between the \emph{null hypothesis} that $\rho$ is a simple prescribed state (e.g., the maximally mixed state) and the \emph{alternative hypothesis} that $\rho$ was sampled from some ensemble $\mathcal{E}$ over ``interesting'' states from some class.

By definition, any strategy for measuring the copies of $\rho$ induces a classical distribution over readouts. It is then well-defined to consider low-degree polynomial estimators that take as input the bits in all of the readouts and output a prediction for whether $\rho$ came from the null or the alternative hypothesis. Given a class $\mathcal{C}$ of measurement strategies, we say that a hypothesis testing problem is \emph{low-degree hard for $\mathcal{C}$} if for \emph{all} strategies in that class, the resulting classical distributions over readouts under the null hypothesis and under the alternative hypothesis cannot be distinguished using any low-degree polynomial estimator.

This quantification over measurement strategies is the key novelty in passing from classical to quantum and poses a number of technical hurdles that we will overcome in the sequel.

\subsubsection{General conditions for low-degree hardness}

Our first set of results lays out general conditions under which a hypothesis testing problem is low-degree hard for various classes of measurement strategies.

\paragraph{From weak designs to strong indistinguishability.} We show that if the ensemble $\mathcal{E}$ is an approximate state $2$-design, that is, if its second-order moments approximately match those of the Haar measure, then hypothesis testing is low-degree hard.

\begin{theorem}[Informal, see \Cref{coro:general_ancilla}]\label{thm:general_informal}
    Distinguishing whether $\rho = \Id/2^n$ or whether $\rho$ was sampled from an ensemble $\mathcal{E}$, given $m = \mathrm{poly}(n)$ copies, is degree-$k$ hard for any non-adaptive, single-copy measurement strategy where each measurement is implemented using at most $O(n)$ ancilla, provided $\mathcal{E}$ forms a $2^{-\tilde{\Omega}(k\log n)}$-approximate state $2$-design.
\end{theorem}

\noindent We also obtain similar results for other strategies like local (i.e. single-\emph{qubit}) measurements (see \Cref{coro:general_single_qubit}) and projective measurements implemented by shallow geometrically local circuits (see \Cref{coro:general_bounded_depth}).

One interesting feature of Theorem~\ref{thm:general_informal} is that the condition on the ensemble $\mathcal{E}$ for degree-$k$ hardness is merely a condition on its \emph{degree-$2$ moments}. The takeaway is that, starting from the relatively weak condition that \emph{two} copies of $\rho$ drawn from $\mathcal{E}$ are \emph{statistically} indistinguishable from two copies of a Haar-random state, we can upgrade to show that under single-copy measurement strategies, \emph{polynomially} many copies of $\rho$ drawn from $\mathcal{E}$ are \emph{computationally} indistinguishable from polynomially many copies of a Haar-random state. This ``degree magnification'' effect appears to be an inherently quantum phenomenon.

\paragraph{Adaptive measurements.} One distinctive aspect of the quantum setting is that the measurements can also be chosen adaptively based on previous measurement outcomes. Modeling adaptivity in the low-degree framework however turns out to be surprisingly subtle. As we explain in Section~\ref{sec:general_adaptivity_full}, the issue is that arbitrary adaptivity can easily hide expensive classical computation that is not captured by the eventual degree of the classical post-processing, and natural attempts at constraining the complexity of the adaptive choices either fail to circumvent this or result in models where proving low-degree hardness is at least as hard as proving circuit lower bounds.

Nevertheless, we are able to identify two nontrivial settings in which we provide general conditions under which low-degree hardness holds even under adaptive measurements:

\begin{theorem}[Informal, see \Cref{coro:general_adaptivity_within_block} and \Cref{coro:general_adaptivity_among_block}]\label{thm:adaptivity_informal}
    Distinguishing whether $\rho = \Id/2^n$ or whether $\rho$ was sampled from an ensemble $\mathcal{E}$, given $m = \mathrm{poly}(n)$ copies, is degree-$k$ hard for single-copy projective measurements in bases that are adaptively chosen under either of the following conditions:
    \begin{itemize}[leftmargin=*,topsep=0.5em,itemsep=0pt]
        \item \underline{Full adaptivity \emph{within} blocks of $m_0=\mathrm{polylog}(n)$ copies each}: provided $\mathcal{E}$ is a $2^{-\mathrm{poly}(k,\log n)}$-approximate state $O(km_0)$-design.
        \item \underline{Full adaptivity \emph{across} $\mathrm{polylog}(n)$ blocks}: provided $\mathcal{E}$ is a $2^{-\mathrm{poly}(k,\log n)}$-approximate state $O(k)$-design.
    \end{itemize}
\end{theorem}

\noindent The reader may wonder why this is not immediate given that, for $k$-designs, degree-$k$ polynomials in $\rho$ look Haar-random by definition. However, the key point is that expensive computations possibly hidden in even a single adaptive choice can involve high-degree polynomials. 
In this context, Theorem~\ref{thm:adaptivity_informal} surfaces a subtle trade-off between the design-degree (and the level of approximation) and the adaptivity allowed in the framework.

We note that recent work in the low-degree literature has also identified several interesting \emph{classical} settings where there is an option to make adaptive choices, e.g., sublinear-time planted clique~\cite{mardia2024low} and geometric graph inference~\cite{bangachev2024fourier}. These works only prove hardness when the choices are non-adaptive, and capturing adaptivity in the low-degree framework here remains an open question, for which we believe our results shed light both technically and from a modeling perspective.

Our results on adaptivity rely on the development of a stronger notion of \emph{copy-wise} degree which allows for polynomial estimators whose monomials, while each depending on a small number of copies, can have higher degree over the qubits of those copies. This mirrors a notion studied in the classical low-degree literature~\cite{brennan2020statistical}. In Section~\ref{sec:SQ}, we also leverage connections, first shown in that work, to an alternative computational model called the statistical query model to show that these two models are equivalent in the realistic setting where there is stochastic noise in the measurement readouts.

\paragraph{Applications.} As nearly immediate applications of this general theory, we are able to deduce new computational hardness results for learning well-studied classes of states like random quantum circuit states, Gibbs states, and states arising in the context of quantum error mitigation. 

\begin{corollary}[Informal, see \Cref{coro:circuit_single_copy}, \Cref{coro:circuit_adaptivity_within_block}, and \Cref{coro:circuit_adaptivity_among_block}]
    For the above measurement strategies, it is degree-$\mathrm{polylog}(n)$ hard to distinguish between the maximally mixed state and a state prepared by a random geometrically local quantum circuit of $\mathrm{polylog}(n)$-depth, given $\mathrm{poly}(n)$ copies.
\end{corollary}

\noindent This extends prior work of~\cite{nietner2023average} which showed that the \emph{output distributions} of these circuits, when measured in the computational basis, are indistinguishable from uniformly random. That result, while not in the low-degree framework, can be interpreted as saying that random polylog-depth circuit states cannot be computationally distinguished from maximally mixed by measuring polynomially many copies in the computational basis. In contrast, our result can be interpreted as showing that this distinguishing problem is hard for a much more general family of single-copy measurement strategies, even ones which have some level of adaptivity. In Section~\ref{sec:qsq} we discuss the connection with a \emph{quantum statistical query} lower bound for this problem that was proved in~\cite{arunachalam2023role}. Taken together, these results can be viewed as modest steps towards the long-standing conjecture (see for instance Conjecture 1.1 of \cite{fefferman2024hardness} or Conjecture 1 of \cite{bostanci2024efficient}) that states prepared by superlog-depth random quantum circuits are pseudorandom.

Next we state our hardness result for learning Gibbs states:

\begin{corollary}[Informal, see Corollary~\ref{coro:hgue_gibbs_sparse}]\label{coro:gibbs_informal}
    For the measurement strategies in Theorem~\ref{thm:general_informal}, it is degree-$\Omega(n)$ hard to distinguish between the maximally mixed state and the Gibbs state of a Hamiltonian given by a random signed sum of $O(n^3)$ random Pauli strings.
\end{corollary}

\noindent To our knowledge, this is the first evidence that learning \emph{sparse} (but not necessarily local) Hamiltonians from Gibbs state access is computationally hard. In contrast, recent works have shown that such Hamiltonians can be learned given access to their dynamics~\cite{ma2024learning,hu2025ansatz}. Note that~\cite{gu2024simulatingquantumchaos} has proven a no-go theorem for the task of distinguishing between two ensembles of Hamiltonians which are both sparse in the \emph{computational} basis (and thus not sparse in the Pauli basis) given access to their \emph{time evolutions}.
% However, this result does not imply ours, as the output of their learner is a circuit approximation. 

We also show evidence that quantum error mitigation is computationally hard for \emph{geometrically local} noisy quantum circuits in any architecture with constant number of layers of single-qubit depolarizing noise at extremely low noise rate $O(1/n^{1-\delta'})$ for any $\delta'>0$. In contrast, quantum error mitigation is only known to be statistically hard for geometrically local noisy quantum circuits with $\omega(\log n)$ layers of noise at constant noise rate~\cite{quek2024exponentially}. The low noise rate of $O(1/n^{1-\delta'})$ is far from making the quantum error mitigation task statistically hard: there is a statistically efficient algorithm using multi-copy joint measurement as the trace distance between the two ensembles is large. However, whether there exists a statistically efficient protocol using single-qubit measurements for quantum error mitigation within the same parameter regime is unknown, leaving the existence of a strict information-computation gap an open question.

\begin{corollary}[Informal, see Theorem~\ref{thm:EM}]\label{thm:EM_informal}
    For any non-adaptive measurement strategy using $m=\poly(n)$ single-qubit measurements, there is a $D=\omega(\log n)$ such that it is degree-$D$ hard to distinguish between the maximally mixed state and geometrically local noisy quantum circuits on any architecture of depth $\omega(\log n)$ at extremely low local depolarization noise rate $O(1/n^{1-\delta'})$ for any $\delta'>0$. Furthermore, the low-degree hardness result still holds even if the noise is only applied in constantly many layers instead of every layer in the noisy quantum circuit.
\end{corollary}

\subsubsection{Beyond designs: fine-grained thresholds for quantum planted biclique}

One selling point of the classical low-degree method that the results above do not yet capture is its ability to pinpoint \emph{fine-grained} thresholds for the signal-to-noise ratio at which a problem transitions from being computationally tractable to computationally intractable but information-theoretically easy. More precisely, in addition to problems for which computational hardness holds for any polynomial sample complexity, the low-degree method also provides accurate predictions for problems for which computational hardness only manifests with a \emph{specific} polynomial scaling for the sample complexity. While identifying phase transitions is integral to understanding physical systems in other quantum information contexts \cite{Skinner_19_Measurement, Ippoliti21, Ippoliti_2024}, the field currently lacks frameworks for establishing such fine-grained statements when it comes to the onset of hardness for learning.

In the second part of this work, we take a step towards filling this gap by identifying the first average-case quantum learning problem, parametrized by a natural notion of signal-to-noise ratio, for which one can establish such fine-grained bounds. The problem we propose, which we call \emph{Quantum Planted Biclique}, can be thought of as a quantum generalization of the well-studied problem of \emph{planted biclique}~\cite{feldman2017statistical}.

The goal is to distinguish, given $n$ copies of an unknown $n$-qudit state $\sigma\in(\mathbb{C}^{d\times d})^{\otimes n}$, whether $\sigma$ is the maximally maxed state or was generated as follows: 
\begin{itemize}[leftmargin=*,topsep=0.5em,itemsep=0pt]
    \item Nature samples a Haar-random $d$-dimensional state $\rho$ and an unknown subsystem $S\subseteq[n]$ of expected size $\lambda$ at the outset.
    \item Each copy of $\sigma$ is constructed by ``planting'' $\rho$ into the qudits indexed by $S$ in an otherwise maximally mixed state, and then corrupting the state with global depolarizing noise of rate $\lambda/n$.
\end{itemize}
We defer the formal definition of this problem to Section~\ref{sec:conn_planted_clique}. 

For intuition, consider the case of $d = 2$. If $\rho$ were $\ket{1}$, rather than Haar-random, and every copy was measured in the computational basis, then the readouts correspond to a random bipartite $n\times n$ graph where a complete bipartite graph is planted in some subgraph of expected size $\lambda\times \lambda$. In contrast, if $\sigma$ were maximally mixed, the readouts would correspond to a random bipartite $n\times n$ graph. Distinguishing between these two cases is exactly the problem of \emph{(classical) planted biclique}. For that problem, distinguishing is information-theoretically possible but strongly believed to be computationally intractable when $\log(n) \ll \lambda \ll \sqrt{n}$~\cite{feldman2017statistical,brennan2020statistical}.

When $\rho$ is instead Haar-random, the resulting classical distinguishing problem that arises after measuring in the computational basis is very similar, but instead of a complete bipartite graph, what is planted is a bipartite graph whose edge probabilities are a constant bounded away from $1/2$. This is the problem of \emph{planted dense subgraph}, which exhibits the same information-computation gap as planted biclique.

We ask how the complexity landscape of quantum planted biclique changes as $d$ grows with $n$. It is straightforward to show that when $\lambda = \omega(n^{1/2}d^{1/4})$, there is a low-degree protocol based on local measurements in the computational basis which succeeds at distinguishing the two cases (see \Cref{lem:local_upper_pds}). In this work, we show a matching lower bound under our low-degree framework:

\begin{theorem}[Informal, see Theorem~\ref{thm:pds_lbd}]\label{thm:pds_informal}
    If $\lambda \le \tilde{o}(n^{1/2}d^{1/4})$, quantum planted biclique is degree-$n^{o(1)}$ hard for non-adaptive local measurements.
\end{theorem}

\noindent When $\lambda = \omega(\min(n^{1/2}d^{1/4}, d))$, one can verify that there is a computationally \emph{inefficient} protocol using the same measurements which succeeds (see \Cref{lem:local_upper_pds}). This means that for any $d = o(n^{2/3})$, there is an information-computation gap for this problem.

\begin{remark}
    When $d = \omega(n^{2/3})$, we conjecture that quantum planted biclique is not just low-degree hard, but information-theoretically impossible with local measurements. Interestingly, we observe that in this regime, there is a low-degree algorithm based on \emph{non-local} single-copy measurements which works provided $\lambda = \omega(n^{2/3})$ (see Lemma~\ref{thm:single_copy_pds}). We leave as an intriguing open question resolving the low-degree complexity of quantum planted biclique under non-local single-copy measurement strategies. For instance, is there such an algorithm that works even below $\lambda \asymp n^{2/3}$? This would yield a fine-grained \emph{computational} separation between local and non-local single-copy measurement strategies.
\end{remark}

\noindent Given the central role played by the planted clique conjecture within the classical literature on complexity of statistical inference, in particular the extensive web of reductions from this problem to other inference problems~\cite{berthet2013complexity,brennan2018reducibility,brennan2020reducibility}, our definition of the quantum version of this problem may also be of independent interest.

\subsubsection{Quantum hardness for hybrid and classical problems.}

In the last part of this work, we explore whether the low-degree framework can also be used to establish quantum hardness for problems where a portion or all of the input is \emph{classical}.

Our first result along these lines concerns the \emph{Learning Stabilizers with Noise (LSN)} problem~\cite{poremba2024learning}, recently proposed as a hybrid classical-quantum analogue of Learning Parity with Noise (LPN). Whereas the latter can be viewed as the task of decoding a random classical linear code in the presence of bit flip noise, LSN is that of decoding a random \emph{quantum stabilizer code} in the presence of local depolarizing noise. Importantly, in this problem, one is given a \emph{classical} description of a random stabilizer subgroup $S$, in addition to a noisy quantum encoded state given by passing an unknown state in the computational basis through an encoding circuit for the stabilizer code associated to $S$.

Importantly, one has to be careful how to formulate the low-degree framework in this hybrid context. LPN is notably low-degree hard even in the absence of noise because Gaussian elimination, while polynomial-time, is not a low-degree algorithm, and LSN inherits this behavior as a generalization of LPN. In particular, if one does not constrain what kind of measurement is performed on the encoded state, there are simple Gaussian elimination-based ways to ``cheat'' and select a measurement which renders the subsequent classical post-processing tractable under low degree. We circumvent this issue by modeling the choice of measurement as low-degree in the classical input, establishing the following average-case hardness result:

\begin{theorem}[Informal, see \Cref{thm:other_lsn_avg}]
    Learning Stabilizers with Noise is $\mathrm{polylog}(n)$-degree hard in the average case for projective measurements.
\end{theorem}

\noindent Interestingly, unlike the results of~\cite{poremba2024learning}, ours applies to the ``standard'' average-case formulation of the problem, where the underlying stabilizer code is uniformly random. One drawback of our result however is that it applies even when there is no noise. This is an inherent limitation of the low-degree method, as LPN and other problems with significant algebraic structure are low-degree hard despite being solvable in the absence of noise using brittle algebraic methods like Gaussian elimination~\cite{holmgren2020counterexamples}. Nevertheless, given that these algorithms are ``high-degree,'' the low-degree method provides evidence for the hardness of these problems in the presence of noise. 

Finally, we also consider purely classical problems. Recently, there has been interest in establishing quantum advantage for classical statistical inference problems~\cite{schmidhuber2025quartic,hastings2020classical}. In this work, we observe that a natural family of approaches based on reducing to a quantum state learning problem by applying a low-degree polynomial encoding of the classical input into copies of a quantum state, applying non-adaptive single-copy measurements, and running low-degree classical post-processing on the readouts will fail provided the original problem is low-degree hard in the classical sense. 

By applying this reasoning to tensor PCA~\cite{montanari2014statistical}, we obtain as a corollary the following average-case hardness result, which complements a recent result of~\cite{bakshi2024learninga} that showed that agnostically learning product states is $\mathsf{NP}$-hard in the worst case.

\begin{theorem}[Informal, see Theorem~\ref{thm:agnostic_product}]\label{cor:product_informal}
    There is a simple ensemble over $4n$-qubit pure states for which, given $\mathrm{poly}(n)$ copies of a state $\rho$ drawn from this ensemble, for any $D=o(n)$ it is degree-$D$ hard to approximately learn the best product state approximation to $\rho$, even though such an approximation exists with fidelity at least $1/\mathrm{poly}(n)$.
\end{theorem}

\subsection*{Roadmap}

In Section~\ref{sec:preliminary}, we review basic notions from quantum information and from the classical low-degree method. In Section~\ref{sec:related}, we discuss relevant prior work and situate our results within the broader literature on the complexity of quantum and classical learning. In Section~\ref{sec:overview}, we provide a technical overview of the main proof ingredients behind our contributions. 

In Section~\ref{sec:general}, we present our first set of results connecting state $2$-designs to low-degree hardness for non-adaptive single-copy measurement protocols. In Section~\ref{sec:adapt}, we prove our results on adaptive measurements. In Section~\ref{sec:applications}, we apply the theory from the preceding sections to derive new information-computation gaps for quantum error mitigation, learning Gibbs states, and random quantum circuits. In Section~\ref{sec:conn_planted_clique}, we present the quantum planted biclique problem and prove upper and lower bounds for this problem. In Section~\ref{sec:hybrid_classical}, we prove our results on low-degree hardness for problems with hybrid classical-quantum or fully classical input, namely for Learning Stabilizers with Noise and for agnostically learning product states. In the Appendix, we provide additional proof details.

\section{Preliminaries}\label{sec:preliminary}
Here, we collect the basic concepts and technical tools used in quantum information and classical low-degree likelihood ratio throughout this paper. The readers can skip these background results according to their knowledge in these fields.

\paragraph{Notation.} We use $\norm{A}$ and $\norm{A}_1$ to represent the operator norm and trace norm of a matrix $A$, and use $\norm{v}_p$ to represent the $L_p$ norm of a vector $v$. We use $\Id_n$ to denote the $2^n\times 2^n$ identity matrix, and omit the subscript $n$ when it is clear from context. We use $\tO$, $\tOmega$, and $\tTheta$ to hide poly-logarithmic dependences in big-O notation. Given two distributions $p,q$ on a discrete domain $X$, we denote by $d_{\mathrm{TV}}(p,q)\coloneqq\frac12\sum_{x\in X}\abs{p(x)-q(x)}$ the total variation distance between them. We will abbreviate the set $\{1,2,...,c\}$ as
$[c]$. 

Given a distribution $D$ and an integer $m$, we denote by $D^{\otimes m}$ the joint
distribution over $m$ independent samples from $D$. Given $f,g\in L^2(D)$, we use $\langle \cdot, \cdot\rangle_{D}$ and $\|\cdot\|_{D}$ to denote the inner product and induced norm for $L^2(D)$ and drop the subscript $D$ when it is clear from context. We will occasionally conflate probability distributions with their laws and density functions.

\subsection{Quantum information basics}\label{sec:basic_quantum}

\paragraph{Quantum states.} 
The \emph{density matrix} of a general $n$-qubit mixed state is given by a psd matrix $\rho\in\C^{2^n\times 2^n}$ with $\tr(\rho)=1$. We also use $d=2^n$ to denote its dimension. $\rho$ is a \emph{pure state} if it can be represented as a rank-1 matrix $\ketbra{\psi}$ for some unit vector $\ket{\psi}\in\mathbb{C}^{2^n}$, in which case $\tr(\rho^2)=1$ and we sometimes refer to the state as $\ket{\psi}$. Given an $n$-qubit quantum state $\rho$ and a subset $S\subseteq[n]$, the quantum state on the subsystem supported by qubits in $S$ (also known as the \emph{reduced density matrix}) is denoted by the \emph{partial trace} $\tr_{[n]\backslash S}(\rho)$. Given two states $\rho,\sigma\in\C^{2^n\times 2^n}$, their \emph{trace distance} is $d_{\tr}(\rho,\sigma)=\frac12\norm{\rho-\sigma}_1$.

\paragraph{Quantum measurements.}
Quantum measurements can in general be represented as \emph{positive operator-valued measures} (POVMs). An $n$-qubit POVM is defined as a set of psd operators $\{F_s\}_s$ with $\sum_s F_s = \Id$, where each $F_s$ is called a \emph{POVM element} corresponding to the \emph{measurement outcome} $s$. The probability of obtaining outcome $s$ when we measure a given state $\rho$ using $\{F_s\}_s$ is given by $\tr(F_s\rho)$.

In this work, we often consider the important special case of \emph{projector-valued measures} (PVMs). These are POVMs whose elements are projection matrices. The set of all PVMs captures all measurements that can be performed on $n$-qubit states without introducing additional ancillary qubits. By Naimark's dilation theorem~\cite{naimark1940second}, an arbitrary $2^m$-outcome POVM can be realized by a PVM on an enlarged quantum system by introducing $m$ ancillary qubits.

Given a PVM $\{\Pi_s\}_s$, we can also associate the projectors $\{\Pi_s\}_s$ with an Hermitian observable $M=\sum_s \Pi_s$, so that $\Pi_s$ is the projector to the eigenspace of $M$ with eigenvalue $m$. We diagonalize $M=U^\dagger D U$ by a unitary matrix $U$ (not unique). Without loss of generality, we can assume any PVM can be performed by applying a unitary $U$ followed by the PVM given by projectors to the \emph{computational basis} (i.e. standard basis) vectors $\{\ket{x}: x\in\{0,1\}^n\}$. 

We say that \emph{PVM $\{\Pi_s\}_s$ can be prepared by a quantum circuit of depth $L$} if there exists such a $U$ that can be prepared by a geometrically-local quantum circuit containing only single- and two-qubit gates on neighboring qubits of depth $L$. 

\paragraph{Haar random unitaries and states.}The \emph{Haar measure $\muh$ on the unitary group $U(d)$} is the unique probability measure that is invariant under left- and right-multiplication by any $U\in U(d)$. The \emph{Haar measure on $d$-dimensional states} is the unique rotation invariant measure on states given by $\ket{\psi} = U\ket{0}$ with $U\sim\muh$ and $\ket{0} \in \mathbb{S}^{d-1}_{\mathbb{C}}$ the all-zeros state. We will often abuse notation and write $\psi\sim\muh$. Moments of the Haar measure can be computed with the following folklore formula (see e.g.~\cite{harrow2013church}):
\begin{align*}
\E_{\psi\sim\muh}\bigl[(\ketbra{\psi})^{\otimes k}\bigr]=\frac{\cS_k}{{{d+k-1} \choose k}}=\frac{1}{d\cdots (d+k-1)}\sum_{\pi\in S_k} \pi^d,
\end{align*}
where $\cS_k$ denotes the projector onto the symmetric subspace $\mathrm{Sym}^k(\mathbb{C}^d)$, $S_k$ is the symmetric group on $k$ elements, and $\pi^d$ acts on $(\mathbb{C}^d)^{\otimes k}$ via
\begin{align*}
\pi^d\ket{i_1,\ldots,i_k}=\ket{i_{\pi^{-1}(1)},\ldots, i_{\pi^{-1}(k)}}.
\end{align*}
Based on this formula, we can also deduce the following formula for any $\ket{\varphi}$:
\begin{align*}
\E_{\psi\sim\muh}\left[\abs{\braket{\psi|\varphi}}^{\otimes 2k}\right]=\E_{\psi\sim\muh}\left[\tr\left[(\ketbra{\psi})^{\otimes k}(\ketbra{\varphi})^{\otimes k}\right]\right]=\binom{d-k-1}{k}^{-1}.
\end{align*}
We will also need the following fact:
\begin{fact}[Lemma 25 of~\cite{nietner2023average}]\label{fact:haar_divide}
Let $\ket{i_1},...,\ket{i_t}$ with $i-1,...,i_t \in [d]$ be mutually orthogonal state vectors and $\lambda=(\lambda_1,...,\lambda_t)$ a partition of $k$ for $k\leq d$. We have
\begin{align*}
\E_{\psi\sim\muh}\left[\prod_{l=1}^t\abs{\braket{\psi|i_l}}^{2\lambda_l}\right]=\frac{\prod_{l=1}^t\lambda_l!}{d...(d+k-1)}.
\end{align*}
\end{fact}

\paragraph{Unitary and state designs.} 
A \emph{unitary (resp. state) $k$-design} is a probability distribution $\nu$ over $U\in U(d)$  (resp. over states $\ket{\psi}\in\C^d$) such that the $k$-th moments are the same as $k$-th moments of the Haar measure $\muh$.

\begin{definition}[Unitary (resp. state) design]
A probability distributions $\nu$ over unitary $U\in U(d)$ (pure states in $\mathbb{C}^d$) is called a \emph{unitary (resp. state) $k$-design} if the following inequalities are satisfied respectively:
\begin{align*}
&\E_{U\sim\mu}f(U)=\E_{U\sim\muh}f(U),\\
&\E_{\psi\sim\mu}(\ketbra{\psi})^{\otimes k}=\frac{\cS_{k}}{\binom{d+k-1}{k}},
\end{align*}
for any polynomial $f$ of degree at most $k$.
\end{definition}

\noindent By definition, $k$-designs are related to low-degree hardness: If all expectation values of degree $k$-polynomials look Haar, then the ensemble is $k$-degree hard to distinguish from Haar.
We will also need the more subtle concept of \emph{approximate} designs:

\begin{definition}[Approximate unitary (state) design]\label{def:approx_design}
A probability distribution $\nu$ over unitaries $U\in U(d)$ (resp. pure states in $\mathbb{C}^d$) is called an \emph{$\epsilon$-approximate unitary (resp. state) $k$-design} if the following inequalities are satisfied respectively:
\begin{align*}
&\max_{A\in\C^{d^2\times d^2}:\norm{A}_1\leq 1}\bigl\|\{(\Phi^{(k)}(\nu)-\Phi^{(k)}(\muh))\otimes \cI_n\}(A)\bigr\|\leq\frac{\epsilon}{d^k},\\
&(1-\epsilon)\E_{\psi\sim\mu}(\ketbra{\psi})^{\otimes k}\leq\frac{\cS_{k}}{\binom{d+k-1}{k}}\leq(1+\epsilon)\E_{\psi\sim\mu}(\ketbra{\psi})^{\otimes k},
\end{align*}
where $\cI_n$ is the identity transformation on a $d=2^n$-dimensional ancillary system and
\begin{align*}
\Phi^{(k)}(\mu)(B)=\int U^{\otimes k} BU^{\dagger\otimes k}d\mu(U)
\end{align*}
for any $B\in\C^{d\times d}$.
\end{definition}

\subsection{Low-degree likelihood ratio}\label{sec:basic_low_degree}

In this work we focus on \emph{hypothesis testing} problems. In preparation for discussing the particulars of the quantum setting we will consider in Section~\ref{sec:general}, here we review existing frameworks for establishing statistical and computational lower bounds for such problems in the \emph{classical setting}.

In a \emph{many-versus-one distinguishing problem} over $\R^n$, the goal is to distinguish between the following two cases, given samples $x_1,\ldots,x_m \in \R^n$:
\begin{itemize}
    \item $H_0$ (null case): the samples were drawn i.i.d. from a \emph{null distribution} $D_\emptyset$ over $\R^n$.
    \item $H_1$ (alternative case): the samples were drawn i.i.d. from some distribution $D_u$ coming from a family $\cS=\{D_u\}_{u\sim\mu}$, with $u$ sampled at the outset from prior distribution $\mu$.
\end{itemize}
We will assume \emph{$m=\poly(n)$ throughout this paper} as we only consider sample-efficient algorithms. Denote all $m$ samples as $\x=(x_1,...,x_m)\in\R^{m\times n}$.

\begin{definition}
    A distinguisher $f: (\R^n)^m \to \{0,1\}$ is said to achieve \emph{$\eta$-detection} if
    \begin{equation*}
        |\Pr_{D^{\otimes m}_\emptyset}[f(\x) = 1] - \Pr_{u\sim \mu}\Pr_{D^{\otimes m}_\mu}[f(\x) = 1]| > \eta\,.
    \end{equation*}
    If $\eta = \Omega(1)$ (resp. $\eta = 1 - o_n(1)$), we refer to this as \emph{weak (resp. strong) detection}.
\end{definition}

Given a hypothesis testing problem, the hardness of achieving weak detection for a class of distinguishers indicates that all such distinguishers fail to distinguish (and thus learn) with any $> 50\%$ constant probability of success.

\paragraph{Information-theoretic bounds.}

The central object for understanding how large $m$ must be before detection is possible is the \emph{likelihood ratio}. Throughout this work, we will assume that $D_u$ is absolutely continuous relative to $D_\emptyset$ for all $u$. We can then define the single-and $m$-sample likelihood ratios by 
\begin{equation*}
    \overline{D}_u(x) = D_u(x) / D_\emptyset(x) \qquad \text{and} \qquad \overline{D}^{\otimes m}_u(\x) = D_u^{\otimes m}(\x) / D_\emptyset^{\otimes m}(\x)\,.
\end{equation*}
The classic Neyman-Pearson lemma implies the following:
\begin{fact}\label{fact:NP}
    The likelihood ratio test $f(\x) \triangleq \bone[\E_u\overline{D}^{\otimes m}_u(\x) > 1]$ achieves $\eta$-detection if and only if $d_{\mathrm{TV}}(\E_u D^{\otimes m}_u(\x), D^{\otimes m}_\emptyset(\x)) > \eta$. 
\end{fact}

To motivate the low-degree likelihood ratio, we can consider the following ``$L_2$'' formulation of the likelihood ratio test. Instead of considering $L^\infty$-bounded distinguishers $f: (\R^n)^m\to\{0,1\}$, let us consider $L^2$-bounded distinguishers $p: (\R^n)^m\to\R$, i.e. test functions for which $\norm{p(\x)}_{D^{\otimes m}_\emptyset} \le 1$. 
One can readily verify that the choice of $p$ which maximizes the \emph{distinguishing advantage}, i.e. which solves \begin{equation}
    \sup_{\norm{p}_{D_\emptyset^{\otimes m}} = 1} \Bigl|\E_{D^{\otimes m}_\emptyset}[p(\x)] - \E_{u\sim \mu}\E_{D^{\otimes m}_\mu}[p(\x)]\Bigr| \label{eq:distinguishL2}
\end{equation}
is given by the shifted and normalized  likelihood ratio:
\begin{equation*}
    p^*(\x) = \frac{\E_u\overline{D}^{\otimes m}_u(\x) - 1}{\norm{\E_u\overline{D}^{\otimes m}_u - \textbf{1}}_{D_\emptyset}}
\end{equation*}
where $\textbf{1}$ denotes the constant $1$ function. This achieves distinguishing advantage $\norm{\E_u\overline{D}^{\otimes m}_u - \textbf{1}}_{D_\emptyset}$; in this work, this quantity will always be finite as we only consider distributions $D_u, D_\emptyset$ over discrete domains.

By Jensen's inequality, this optimal distinguishing advantage upper bounds $d_{\rm TV}(\E_u D^{\otimes m}_u, D_\emptyset^{\otimes m})$, so by Fact~\ref{fact:NP}, we obtain the following condition for information-theoretic hardness of detection:
\begin{lemma}
    If the optimal distinguishing advantage in Eq.~\eqref{eq:distinguishL2} is $o_n(1)$, then weak detection is not achievable by any distinguisher.
\end{lemma}

\paragraph{Computational bounds.} The low-degree likelihood framework offers a computational analogue of the above reasoning. Instead of considering arbitrary $L^2$-bounded distinguishers, we restrict our attention to distinguishers $p: (\R^n)^m\to \R$ which are given by polynomials of some degree $k$. Note that any such polynomial can be evaluated in time $(nm)^{O(k)}$.

The motivation for this first arose in the context of establishing sum-of-squares lower bounds for statistical estimation~\cite{barak2019nearly,hopkins2017efficient,hopkins2017power}, and more generally it is supported by the fact that many of the most powerful computationally efficient algorithms for hypothesis testing, e.g. spectral methods, message passing, and method of moments, can be implemented as low-degree polynomial estimators (see the survey~\cite{kunisky2019notes} for an overview).

Let $(\chi_\alpha)_{|\alpha|\le k}$ denote an orthonormal basis for the space $\mathcal{V}^{\le k}$ of polynomials of degree at most $k$ in $L^2(D_\emptyset^{\otimes m})$. Given $f\in L^2(D_\emptyset^{\otimes m})$, we denote by $f^{\le k}$ the projection of $f$ to this subspace. We will focus almost exclusively (with the exception of Section~\ref{sec:conn_planted_clique}) on $D_\emptyset$ which is uniform over the Boolean cube, in which case $(\chi_\alpha)_{|\alpha| \le k}$ are simply the Fourier characters of degree at most $k$ and $f^{\le k}$ is simply the degree-$\le k$ Fourier truncation of $f$. We can now consider the following modification to the optimization problem in Eq.~\eqref{eq:distinguishL2}:
\begin{equation}
    \sup_{p\in \mathcal{V}^{\le k}: \norm{p}_{D_\emptyset^{\otimes m}} = 1} \Bigl|\E_{D^{\otimes m}_\emptyset}[p(\x)] - \E_{u\sim \mu}\E_{D^{\otimes m}_\mu}[p(\x)]\Bigr|\,.
    \label{eq:distinguishL2_k}
\end{equation}
The optimal $p$ is given, up to normalization and shifting, by the likelihood ratio \emph{truncated to degree at most $k$}:
\begin{equation*}
    p^*(\x) = \frac{\E_u(\overline{D}^{\otimes m}_u)^{\le k}(\x) - 1}{\norm{\E_u(\overline{D}^{\otimes m}_u)^{\le k} - \textbf{1}}_{D_\emptyset}}\,.
\end{equation*}
We refer to $\E_u (\overline{D}^{\otimes m}_u)^{\leq k}$ as the \emph{($m$-sample) low-degree likelihood ratio}.

$p^*$ achieves distinguishing advantage $\norm{\E_u (\overline{D}^{\otimes m}_u)^{\le k} - \textbf{1}}_{D_\emptyset}$. In particular, if the likelihood ratio has the Fourier expansion $\E_u \overline{D}^{\otimes m}_u(\x) = \sum_\alpha \wh{D}[\alpha] \chi_\alpha(\x)$, then the optimal distinguishing advantage achieved by any degree-$\le k$ distinguisher is given by $\Adv{k}$ for
\begin{equation*}
    \Adv{k}^2 \triangleq \sum_{0<|\alpha| \le k} \wh{D}[\alpha]^2 \qquad \text{for} \qquad \wh{D}[\alpha] \triangleq \E_{D^{\otimes m}_\emptyset}[\E_u\overline{D}^{\otimes m}_u(\x)\chi_\alpha(\x)]\,.
\end{equation*}
We will refer to $\Adv{k}$ as the \emph{degree-$k$ advantage}.

A long line of work in the classical literature on complexity of statistical inference has established bounds on this quantity for a wide range of tasks. Barring hypothesis testing problems with exceptional algebraic structure, the recurring finding is that boundedness of the low-degree likelihood ratio appears to coincide with the failure of computationally efficient distinguishers, giving rise to the following informal conjecture (see Conjecture 2.2.4 in~\cite{hopkins2018statistical} for a formalization).
\begin{conjecture}[Low-degree conjecture, informal]\label{conj:informal}
    Let $k = k_n$ be some growing function of $n$. For ``natural'' distinguishing problems, there is no $(nm)^{\tilde{O}(k)}$-time distinguisher achieving weak detection between $\E_u D^{\otimes m}$ and $D_\emptyset^{\otimes m}$ if $\norm{\E_u (\overline{D}^{\otimes m}_u)^{\le k} - \textbf{1}}_{D_\emptyset} = o(1)$\,.

    In particular, if $m = \mathrm{poly}(n)$ and $k = \omega(\log n)$, then boundedness of the low-degree likelihood ratio coincides with hardness for all $\mathrm{poly}(n)$-time distinguishers.
\end{conjecture}

\paragraph{Sample-wise degree.}
In some parts of this work, we will consider an alternative notion of degree that allows for a wider range of efficient distinguishers.

We say that a polynomial $p: (\R^n)^m\to\R$ has \emph{sample-wise degree} $(\mathcal{D},k)$ if $p(x_1,\ldots,x_m)$ is a linear combination of functions which each have degree at most $\mathcal{D}$ in each $x_i$ and nonzero degree in at most $k$ of the $x_i$'s (the set of $k$ indices can vary across functions in the linear combination). Note that any such $p$ must have total degree at most $\mathcal{D}k$, and conversely any polynomial $p$ with total degree at most $k$ has sample-wise degree $(k,k)$. We will denote the orthogonal projection (with respect to $L^2(D^{\otimes m}_\emptyset)$) of any distinguisher $p$ to the subspace of polynomials of sample-wise degree $(\mathcal{D},k)$ by $p^{\le \mathcal{D},k}$. The above discussion on optimal distinguishing advantage for low-degree polynomial distinguishers carries over immediately to this alternative setting by replacing all appearances of $\le k$ with $\le \mathcal{D},k$. In particular, the relevant quantity to bound here would be $\E_u (\overline{D}^{\otimes m}_u)^{\le \mathcal{D},k}$. We can also define the \emph{degree-$(\cD,k)$ advantage} $\Adv{(\cD,k)}$ accordingly.

\section{Discussion and related work}
\label{sec:related}

Below we discuss relevant prior work in an effort to situate our contributions within the broader literature on the complexity of quantum and classical learning.

\subsection{Information-theoretic bounds}

Recent works have precisely pinned down the sample complexity for a number of fundamental quantum learning tasks~\cite{chen2022exponential,chen2022tight,chen2022tight2,chen2024tight,chen2024optimal,chen2024optimalshadow,chen2024optimalstate,gong2024sample,odonnell2015quantum,odonnell2016efficient,haah2016sample,anshu2020sample}. A proper survey of this literature is beyond the scope of this paper. We note that a large number of these works have focused on the question of distinguishing whether a state is maximally mixed or from some ensemble of states that are far from maximally mixed, and we too will focus on this problem, albeit from the perspective of computational efficiency.

In our setting, the ensembles of states we consider are typically very structured in the sense that each state can be described by a polynomial number of parameters. For states from such ensembles, we make the (folklore) observation that learning is information-theoretically tractable by a straightforward application of known results:

\begin{remark}\label{remark:info_theoretic}
    We note that given any class $\mathcal{C}$ of $n$-qubit pure states with $\exp(\mathrm{poly}(n))$ covering number, it is always information-theoretically possible to learn states in $\mathcal{C}$ by performing simple single-copy measurements on polynomially many copies. All of the classes mentioned on the first page, e.g. Gibbs states of local Hamiltonians, shallow circuit states, stabilizer states, and matrix product states, satisfy this property.
    
    The reason this problem is information-theoretically possible is that we can perform brute force over a suitable net $\{\ket{\psi}\}$ over $\mathcal{C}$.
    Techniques like classical shadows~\cite{huang2020predicting} allow for the estimation of the fidelity of the given state with all $\ket{\psi}$ in the net using only $O(\log |\mathcal{C}|) = \mathrm{poly}(n)$ copies.
    Furthermore, the measurements needed are efficiently preparable: recent work Ref.~\cite{schuster2024random} proves that shadow estimation for fidelities can be performed using only measurements prepared by $\log(n)$-depth circuits in a 1D nearest-neighbor layout. 
    Consequently, polynomially many of these shallow measurements are sufficient to learn the  circuit description of state $\ket{\psi}$ from the net that has high fidelity with the unknown state.
    
    However, this strategy is computationally inefficient as it requires brute force enumeration and estimation of all $\exp(\mathrm{poly}(n))$ fidelities in classical post-processing.
\end{remark}

\subsection{Classical low-degree method}

The low-degree method emerged from~\cite{barak2019nearly,hopkins2017efficient,hopkins2017power,hopkins2018statistical} and has since been used extensively to obtain sharp predictions for a wide range of problems in classical statistical inference (see~\cite{kunisky2019notes} for a survey from 2019). It is impossible for us to do justice to this vast literature in this section, so here we highlight the parts of this literature most relevant to this work. 

In~\cite{brennan2020statistical}, the authors established a formal connection between the low-degree method and the statistical query model (see below). This allowed them to transfer lower bounds for one framework to lower bounds for the other, recovering sharp predictions in both models for various problems. Most relevant to the present work are their results on planted clique and planted dense subgraph, in particular their bipartite analogues, which recover the tight statistical query lower bounds of~\cite{feldman2015complexity}. In addition to drawing upon this result in Section~\ref{sec:general_equivalence} when we establish a similar connection in the quantum setting, in Section~\ref{sec:overview_pds} we discuss how our lower bound for quantum planted biclique builds upon calculations performed in the classical setting. We also note the work of~\cite{dhawan2025detection} which provides classical low-degree bounds for a hypergraph version of this problem, as well as prior works of~\cite{ma2015computational,bresler2023detection,brennan2018reducibility,brennan2020reducibility,hajek2015computational} which studied this same question from the perspective of average-case reductions rather than low-degree hardness.

Our results in Section~\ref{sec:conn_agnostic_product} also touch upon the well-studied problem of tensor PCA~\cite{montanari2014statistical,lesieur2017statistical}, where the goal is to distinguish between a Gaussian tensor and a Gaussian tensor plus a rank-1 spike. A sharp prediction for hardness for this problem via the low-degree method can be found in the note of~\cite{kunisky2019notes}, and~\cite{potechin2022sub} has also proven a sum-of-squares lower bound for this problem; these hardness results are matched by low-degree algorithms via sum-of-squares~\cite{hopkins2015tensor} and via message passing~\cite{wein2019kikuchi}. Tensor PCA has also been studied in quantum contexts, with prior work giving quartic speedups over the best known classical algorithms for this problem~\cite{hastings2020classical,schmidhuber2025quartic}.

\subsection{Hardness of quantum learning from pseudorandomness}

Another popular method to prove computational lower bounds on quantum learning tasks relies on recent constructions of pseudorandom quantum states~\cite{ji2018pseudorandom} and unitaries~\cite{ma2024construct}. These are ensembles over states or unitaries that cannot be distinguished by efficient algorithms from Haar-random states or unitaries conditional on mild cryptographic assumptions. The fact that these ensembles cannot even be distinguished from their Haar-random counterparts implies that they also cannot be learned by efficient algorithms. Crucially, because these states or unitaries require only polynomially-sized circuits or even, in specialized architectures, circuits with subextensive lightcones \cite{schuster2024random} to generate, this then implies that learning even extremely shallow circuits from measurements is computationally hard. 

The existence of quantum pseudorandomness has had profound implications for physics. One key application is in ``spoofing" quantum properties such as entanglement~\cite{aaronson2024quantum} and magic~\cite{Gu_2024} -- thought to be measures of the ``quantumness" of a state --  with polynomial-sized circuits. This immediately places lower bounds on our ability to manipulate the amounts of these resources in a given quantum state by efficient algorithms. Our work immediately raises the question of whether these bounds could be strengthened by adopting a `low-degree' notion of efficiency.

\subsection{Quantum statistical query model}
\label{sec:qsq}

Within the quantum literature, perhaps most relevant to the present work is the \emph{quantum statistical query (QSQ)} model, introduced in~\cite{arunachalam2020quantum}, which offers a complementary perspective on predicting information-computation gaps for quantum learning.

The model is motivated by the \emph{statistical query (SQ)} model~\cite{kearns1998efficient} from the classical literature, another popular framework for proving hardness under a restricted model of computation. Roughly, the SQ considers a model where, instead of directly observing samples, the learning algorithm can only query population-level statistics of the data distribution to some level of precision. This framework has been used to obtain accurate predictions for a range of problems like Gaussian mixtures~\cite{diakonikolas2017statistical}, learning neural networks~\cite{chen2022hardness,chen2022learning,goel2020statistical,goel2020superpolynomial,diakonikolas2017statistical,diakonikolas2020algorithms}, and robust linear classification~\cite{diakonikolas2022near}. As with the low-degree method, the literature here is too vast to do justice to here.

Analogous to SQ, QSQ is a model where, instead of performing arbitrary multi-outcome measurements on the copies of the unknown state $\rho$, the learning algorithm can only query two-outcome observables of $\rho$ of their choice up to some level of precision, called \emph{tolerance}. Below, we briefly discuss the relative merits of these two models. 

One advantage of QSQ is that it is possible to prove lower bounds even when the queries can be adaptively chosen, whereas in our low-degree framework, it is necessary to limit the adaptivity to avoid various ``cheating'' algorithms (see Section~\ref{sec:general_adaptivity_full}). 

But the QSQ model is limiting in some respects. For one, it assumes that the two-outcome observable values are given up to small \emph{adversarial error}, which can be overly pessimistic~\cite{abbe2021power}. In addition, the notion of sample complexity is somewhat conflated with the notion of tolerance: in this model, if observable values are given to error $\tau$, the sample complexity of the algorithm is effectively regarded as the number of queries times $1/\tau^2$. But this imposes a rather rigid constraint on how samples get used in QSQ algorithms, which can sometimes lead to incorrect predictions for information-computation gaps; this is discussed at length in~\cite{brennan2020statistical}.

Perhaps the most important drawback of the QSQ model is that focusing on protocols that use two-outcome measurements can be overly restrictive. As the following example demonstrates, there can even be cases where the QSQ complexity of a problem exponentially overshoots both the true complexity and what is predicted by the low-degree method.

\begin{example}
    Consider the task of distinguishing between the maximally mixed state and a state which is a Haar-random pure state in the first $\log_2 n$ qubits and maximally mixed elsewhere. It is known that by measuring $O(\sqrt{n})$ copies in the computational basis, one can computationally efficiently distinguish between these two scenarios. The idea is that the distribution over readouts from measuring the first $\log_2 n$ qubits under the second scenario is the uniform distribution over $\{0,1\}^{\log_2 n}$, whereas the distribution in the first scenario is $1/\sqrt{n}$-far from uniform in $L_2$ distance. By counting the number of collisions between samples~\cite{chan2014optimal}, one can distinguish between these two cases. Counting collisions among samples from a distribution over $\{0,1\}^{\log_2 n}$ can be implemented with a degree-$O(\log n)$ estimator, meaning that the above task is tractable under the low-degree framework. In contrast, Ref.~\cite{arunachalam2023role} showed that for QSQ algorithms, $\exp(\Omega(n))$ many QSQ queries are needed for this problem.
\end{example}

\noindent Nevertheless, QSQ offers a fruitful complementary viewpoint on hardness of quantum learning. For instance, in~\cite{arunachalam2023role}, it was shown that ensembles that are approximate state $2$-designs are hard to learn under QSQ, strengthening prior work of~\cite{hinsche2023one,nietner2023average} which focused on \emph{diagonal} two-outcome measurements. This mirrors our findings in the low-degree framework and serves as additional evidence in favor of the conjecture that random quantum circuits are pseudorandom. 
Even more interestingly, the mechanism by which $2$-designs manifest in the analysis for QSQ is quite different from how they appear in our low-degree arguments. In their case, the $2$-design property is crucial for relating the variance of the response to a QSQ (as a random variable in $\rho\sim\mathcal{E}$) to the variance if $\rho$ were Haar-random. In our case, the $2$-design property is crucial for establishing local indistinguishability of copy-wise rotations of the quantum data.

Finally, as shown by~\cite{brennan2020statistical} in the classical case, under certain conditions there is even a formal equivalence between the two models. In Section~\ref{sec:general_equivalence}, we extend this equivalence to the quantum setting by observing that the conditions of~\cite{brennan2020statistical} are satisfied by the distribution over readouts provided the measurements are corrupted by some stochastic noise.

Finally, note that~\cite{arunachalam2023role} also studied a quantum generalization of the planted biclique problem, but despite sharing the same name, their version and motivation is entirely different: they take the classical distribution over right-neighbors of a random left-vertex of the bipartite graph and consider a coherent encoding of this. Their main finding is that if the size of the planted biclique is $\lambda \ge\Omega(\log n)$, then there is an algorithm for their problem that makes $n^{O(\lambda)}$ QSQs to precision $\sqrt{\lambda/n}$, whereas any classical algorithm requires precision $\lambda/n$.

\subsection{Quantum learning algorithms as low-degree estimators}
\label{sec:lowdegree_examples}

\paragraph{Observable averages.} As a trivial first example of a low-degree estimator, consider a protocol that estimates an observable $O$ by measuring multiple copies each with the two-outcome POVM $\{O, \Id - O\}$. The low-degree algorithm in this case is simply the proportion of copies that resulted in the former outcome, which is a degree-1 estimator. 

\paragraph{Shallow circuit learning.} 

It was recently shown in Refs.~\cite{landau2024learning,kim2024learning,huang2024learning} that states prepared by quantum circuits of constant depth can be learned efficiently in the following sense: Given a state prepared by a constant depth circuit and the ability to perform measurements, using only estimations of $O(1)$-local marginals allows one can find another constant-depth circuit that prepares a high-fidelity state.
The first step of learning the local marginals can be done by learning the coefficients of all local Pauli operators, which is a degree $O(1)$ estimator. 
Then, in Ref.~\cite{landau2024learning} a local inversion $V$ is found that rotates the state to one of the form $|0^l\rangle\otimes |\psi'\rangle$, for $l=O(1)$ qubits by inverting the local gates in the lightcone of the $l$ qubits. 
As unitaries are inverted by taking the adjoint, this too is a low-degree polynomial map. 
These local inversions are combined with fixed SWAP operations, which results in an output that is low-degree in the measurement readouts.

\paragraph{Learning Hamiltonians via polynomial system solving.} In~\cite{haah2022optimal}, the authors give a sample- and time-optimal algorithm for learning high-temperature Hamiltonians from Gibbs state access. The protocol first performs measurements to obtain local Pauli observables of the Gibbs state, and then post-processes these by solving a suitable polynomial system using the Newton-Raphson method. This latter algorithm proceeds by repeatedly applying low-degree polynomial maps to the iterate followed by projection, and can because the algorithm converges rapidly, the overall output can be well-approximated by a low-degree polynomial in the readouts.

The recent breakthrough of~\cite{bakshi2024learning} gave an algorithm for learning \emph{arbitrary-temperature} local Hamiltonians from Gibbs state access. They devised a sum-of-squares algorithm for post-processing readouts which were again given by local Pauli measurements. Roughly speaking, they formulated a certain polynomial system in the unknown coefficients of the Hamiltonian and showed that a pseudodistribution of sufficient degree over solutions to this system could be rounded to a good approximation to the true coefficients. While it is unclear how to rigorously formulate this as a low-degree polynomial estimator and consequently how to rule out such algorithms when trying to prove information-computation gaps, in the classical literature, low-degree hardness is typically regarded as initial evidence for sum-of-squares hardness. Indeed, the low-degree advantage $\|\overline{D}^{\le k} - \mathbf{1}\|^2$ emerges naturally in pseudo-calibration~\cite{barak2019nearly}, and showing it is small is a necessary first step in proving a sum-of-squares lower bound.

\paragraph{Agnostic tomography of product states.}
Recent work~\cite{bakshi2024learninga,chen2024stabilizer} shows how to obtain, for any given pure state, a (discrete) product state that is $\epsilon$-close to the optimal fidelity. For agnostic tomography of discrete product states, the algorithm proposed by~\cite{chen2024stabilizer} is adaptive and low-degree by looking for the high-correlated \emph{single-qubit} project for each qubit via adaptively chosen local measurements followed by an adaptive local optimization procedure employing the measurement outcomes as a low-degree function. For agnostic tomography of product states when the optimal fidelity is larger than $5/6$, Ref.~\cite{bakshi2024learninga} gave the faster agnostic tomography algorithm of product states. The algorithm first uses a divide-and-conquer strategy and sews the local tomography of each qubit, and then performs a local product state optimization that only uses local unitary transformations and low-degree post-processing. The resulting learning algorithm as a whole is thus low-degree in the local measurement results.

For agnostic tomography of product states with any optimal fidelity, although the algorithm proposed in~\cite{bakshi2024learning} is highly adaptive, the intuition and most intermediate subroutines are captured by low-degree functions. In particular, the algorithm maintains a cover over all product states that have large fidelity and only optimize fidelity in a low-degree way over extensions of good product states over a small subsystem.

\subsection{Corner cases}

Yet, there are instances where the classical low-degree framework, by limiting its attention to low-degree algorithms, makes incorrect predictions about the computational hardness of a problem. One notable example is $k$-XOR-SAT, a special case of learning parities. The low-degree method predicts that the problem should be computationally hard, even though it is in fact efficiently solvable by the high-degree algorithm of Gaussian elimination. 

Our low-degree framework has the same issue. In \Cref{sec:stabilizer_state}, we show low-degree hardness for learning quantum stabilizer states even though there exist time-efficient high-degree algorithms for the task that effectively perform a Gaussian elimination on measurement readouts \cite{montanaro2017learning}. 

In general, the low-degree method is particularly susceptible to such failures on problems (whether quantum or classical) with a high degree of algebraic structure~\cite{holmgren2020counterexamples}. Nevertheless, algorithms that exploit such algebraic structure to run efficiently tend to be especially brittle -- as exemplified by the famous contrast between the efficiency of learning parities, and the intractability of learning parities with noise. As with the classical low-degree method, we expect that our framework correctly captures the limits of
learning quantum states in the presence of noise. While the current agnostic tomography algorithms~\cite{grewal2024improved,chen2024stabilizer} of learning stabilizer states provide computational efficient protocols even when the fidelity between the underlying (noise-corrupted) state and any stabilizer states drop to $o(1)$, these algorithms depend on adaptive post-selection and Bell difference sampling that requires two-copy measurements. To the best of our knowledge, there is no noise-robust algorithm based on subroutines such as computational difference sampling that only employs single-copy measurements. This can be evidence that our framework may correctly predict the vulnerability of high-degree protocols with single-copy measurements against noise.

Finally, we remark that very recently, Buhai et al.\cite{buhai2025quasipolynomiallowdegreeconjecturefalse} constructed a distinguishing problem for which the degree-$\Omega(n^{1 - O(\epsilon)})$ likelihood ratio vanishes, yet there is a \emph{quasipolynomial}-time algorithm. Furthermore, unlike Gaussian elimination, their algorithm is noise-robust, based on a list-decoding algorithm for noisy polynomial interpolation. Notably, unlike a previous construction which also leveraged a connection to error correction~\cite{holmgren2020counterexamples}, the distinguishing problem of~\cite{buhai2025quasipolynomiallowdegreeconjecturefalse} has permutational symmetry and thus yields a counterexample to a certain formalization of Conjecture~\ref{conj:informal} in~\cite{hopkins2018statistical}. While their specialized construction offers a note of caution in interpreting predictions made by the low-degree method, the method remains useful for obtaining initial evidence of computational hardness and initial estimates for the thresholds at which hardness manifests. Furthermore, in cases where the best known existing algorithms fall under the umbrella of polynomial distinguishers, like the examples documented in Section~\ref{sec:lowdegree_examples}, low-degree hardness articulates a natural barrier faced by contemporary algorithmic techniques.

\section{Overview of techniques}
\label{sec:overview}

Here, we provide a high-level discussion of the key technical ingredients in this work. For most of this work, we will focus on the well-studied \emph{hypothesis testing} setting where the learner is given copies of an unknown quantum state $\rho$ which is either maximally mixed or drawn from an ensemble $\mathcal{E}$ and would like to distinguish between these two cases by performing measurements on the copies of $\rho$.

\subsection{Low-degree hardness with non-adaptive measurements}
\paragraph{Warmup: single-qubit measurements.}To start with, we consider a direct extension of the classical low-degree model for quantum data as a warm-up in \Cref{sec:general_single_qubit}. The learner makes non-adaptive single-qubit measurements $\{\bigotimes_{r=1}^n\ketbra{\phi_{s_i,r}}\}_{s_i}$ for $i\in[m]$ on $m=\poly(n)$ copies, and computes the output as a function of degree $k$ on the classical output in $\bm{s}=(s_1,...,s_m)\in\{0,1\}^{m\cdot n}$ of measurements. 

As the classical outputs of single-qubit measurements are Boolean, for any degree $k$ we can write the $m$-sample low-degree likelihood ratio $\E_\rho(\overline{D}_\rho^{\otimes m})^{\leq k}$ as a Fourier expansion 
\begin{equation}
    \E_\rho(\overline{D}_\rho^{\otimes m})^{\leq k}(\bm{s})=\sum_{T}\hat{D}[T]\chi_T(\bm{s})
\end{equation}
for $T\subseteq[mn]$, $|T|\leq k$. Here $\chi_T(\bm{s})=(-1)^{\sum_{(i,r)\in T}s_{i,r}}$ is the usual Fourier basis. We then compute the degree-$k$ advantage 
\begin{equation}
    \chi_{\leq k}^2=\sum_{0<\abs{T}\leq k}\hat{D}[T]^2\,.    
\end{equation}
Recall from Section~\ref{sec:basic_low_degree} that this quantity is the optimal distinguishing advantage achieved by any distinguisher which takes the bits of the measurement outcomes and applies a polynomial of degree at most $k$. 

As there are at most $\sum_{i=1}^k\binom{mn}{i}\leq(mn)^k$ choices of $T$, we can bound the degree-$k$ advantage by $o(1)$ if we can ensure $\hat{D}[T]\leq 2^{-\tOmega(k\log n)}$. We then show that, if every subsystem $T$ of $\rho^{\otimes m}$ of at most $k$ qubits is $2^{-\tOmega(k\log n)}$-close to the maximally mixed state in trace distance, we can ensure $\hat{D}[T]\leq 2^{-\tOmega(k\log n)}$ and thus that weak detection is degree-$k$ hard (see \Cref{coro:general_single_qubit} for details).

\paragraph{Application: Quantum error mitigation.}
As an application of the above warmup, we show in \Cref{sec:qem} that it is degree-$\omega(\log n)$ hard to distinguish between the maximally mixed state and output state of geometrically local noisy quantum circuits of depth $\omega(\log n)$ at extremely low local depolarization noise rate $O(1/n^{1-\delta'})$ for any $\delta'>0$ as in \Cref{thm:EM}. To obtain this result, we consider the ensemble of quantum circuits that interleaves a constant number of random circuits of depth $\omega(\log(n))$ that form approximate unitary $2$-designs and local depolarization noise of noise rate $O(1/n^{1-\delta'})$. We then show that the local subsystems of output states of random states in this ensemble are indistinguishable from the maximally mixed state. The hardness result then follows by combining the local indistinguishability and the above warmup calculation for single-qubit measurements.

\paragraph{Ancilla-assisted single-copy measurements. }In \Cref{sec:general_ancilla}, we consider the setting in \Cref{thm:general_informal} when the protocol can perform $m=\poly(n)$ non-adaptive single-copy PVMs $\{\ketbra{\phi_{s_i}}\}_{s_i}$ for $i\in[m]$, where each measurement is implemented using at most $n'=O(n)$ ancillary qubits. The challenge compared to the previous setting is that here, the measurements can nontrivially entangle different qubits of the system, precluding a direct argument about the reduced density matrices of $\rho^{\otimes m}$.

As each single-copy PVM with $n'$ ancillary qubits can be written as $\{U_i\ket{s}\bra{s}U_i^\dagger\}_{s\in\{0,1\}^{n+n'}}$ for a unitary $U_i$, the protocol can be regarded as performing computational-basis single-qubit measurements on the ``rotated" batch of $m$ samples $\{U_i\cdot\rho\otimes\ketbra{0}^{\otimes n'}\cdot U_i^\dagger\}_{i=1}^m$. According to the previous discussion on protocols with single-qubit measurements, weak detection is degree-$k$ hard if any subsystem (can be chosen across samples) in the $m$ rotated samples up to $k$ qubits is $2^{-\tOmega(k\log n)}$ indistinguishable from the maximally mixed state. 

To prove \Cref{thm:general_informal}, the main step is then to show that an ensemble $\mathcal{E}$ that forms a $2^{-\tOmega(k\log n)}$-approximate state $2$-design satisfies such local indistinguishability for any sequence of unitaries $\{U_i\}_{i=1}^m$. We provide the intuition here and leave the detailed proof to \Cref{coro:general_ancilla}.

Our first step is to show that, given $m$ copies of $\rho$ sampled from a state $2$-design ensemble and a sequence of $n$-qubit unitaries $\{V_i\}_{i=1}^m$, any subsystem in the $m$ rotated samples $\{V_i\rho V_i^\dagger\}_{i=1}^m$ up to $k$ qubits is $2^{-\Theta(n)}$-close to maximally mixed state. We compute trace distance here using the norm inequality and the square $2$-norm. As $\rho$ forms a state $2$-design, for the purposes of our analysis which implicitly only uses degree-$2$ moments of $\mathcal{E}$, it behaves as if Haar-random chosen and is thus invariant after rotation by $\{V_i\}_{i=1}^m$. For each single copy, the local subsystem is $2^{-\Theta(n)}$-close to maximally mixed state~\cite{hayden2006aspects} and local indistinguishability then follows by a union bound over copies. 

This is already enough to establish the special case of Theorem~\ref{thm:general_informal} when there are $n' = 0$ ancilla. But projective measurements without ancilla preclude even very simple two-outcome measurements. To address this, we use tools from unitary decomposition~\cite{shende2005synthesis} to decompose each $(n+n')$-qubit unitary $U_i$ in $\{U_i\cdot\rho\otimes\ketbra{0}^{\otimes n'}\cdot U_i^\dagger\}$ into $2^{\Theta(n')}$ terms of product between the $n$ qubits and $n'$ ancillary qubits. This decomposition makes it possible for us to ensure local indistinguishability for any sequence of unitaries $\{U_i\}_{i=1}^m$ given $n'\leq O(n)$; we defer the details to~\Cref{sec:general_ancilla}.

\begin{remark}
    Our findings about approximate $2$-design property implying local indistinguishability can be interpreted as a version of the statement that states from a 2-design are approximately \emph{$k$-uniform states} up to any quantum experiment that takes $\mathrm{poly}(n)$ samples. A $k$-uniform state of $n$ qudits is one on which any subsystem of $k$-qudits is maximally-mixed\cite{kuniform1,kuniform2,kuniform3}. While our finding follows from standard calculations on designs, we were not able to find an articulation of this point in the literature and believe it is of conceptual interest.
\end{remark}

\paragraph{Application: Gibbs states of sparse Hamiltonians. } As an application of \Cref{thm:general_informal}, we show in \Cref{sec:hamiltonian_gibbs} that there is a $D=O(n)$ such that it is degree-$D$ hard to distinguish between the maximally mixed state and the Gibbs state of a Hamiltonian given by a random signed sum of $O(n^3)$ random Pauli strings using single-copy PVMs with $O(n)$ ancillary qubits as claimed in \Cref{coro:gibbs_informal}. This can be shown by combining the following two observations. Firstly, the eigenvalues of a random Hamiltonian concentrate within the interval $[-3,3]$ with high probability~\cite{chen2024sparse} (see \eqref{eq:RSPS_eig} for details). Secondly, the Hamiltonian ensemble matches the Gaussian unitary ensemble up to the third moment. The eigenvectors of the Hamiltonian thus form a state $3$-design. Therefore, the Gibbs states of Hamiltonian given by a random signed sum of $O(n^3)$ random Pauli strings can be regarded as a classical combination of states chosen from a state $3$-design.

\subsection{Low-degree hardness with restricted adaptivity} 

Here, we consider protocols that can perform $m=\poly(n)$ adaptive single-copy PVMs, in which each measurement can be chosen adaptively based on previous measurement outcomes. The case when the learners have adaptivity is surprisingly subtle. As a tool to analyze the performance, we introduce the quantum copy-wise low-degree model, a quantum analog of the sample-wise degree model, which includes all post-processing as degree-$(\cD,k)$ functions that are degree $\cD$ in at most $k$ copies (note that all degree-$k$ functions in the standard low-degree model are degree-$(\cD,k)$ functions). In \Cref{lem:k_design_copy} and \Cref{coro:general_sample} (see \Cref{sec:general_copywise}), we extend the reasoning from the previous section to show that weak detection is degree-$(\cD,k)$ hard for non-adaptive single-copy PVMs provided the ensemble $\mathcal{E}$ forms a $2^{-\tOmega(k\log n)}$-approximate state \emph{$k$-design}. To see this, we consider the terms in the Fourier expansion of the degree-$(\cD,k)$ advantage $\chi_{\leq(\cD,k)}^2$, recalling that each is a square of a Fourier coefficient of likelihood ratio $\E_\rho(\overline{D}_\rho^{\otimes m})$. For each fixed $T$, the Fourier coefficient $\hat{D}[T]$ is bounded by the $\abs{T}$-th (the number of copies in $T$) moment of the degree-$\leq\cD$ advantage of a single copy by H\"{o}lder's inequality (see \eqref{eq:holder} for details). We show via a moment calculation that if $\mathcal{E}$ is an approximate $k$-design, this ensures the $\abs{T}$-th moment of the degree-$\cD$ advantage of a single copy is bounded by $2^{-\tOmega(k\log n)}$ \emph{for any $\cD$}.

\paragraph{Round-based model with adaptivity within blocks. }The most important benefit of using the copy-wise low-degree model is that it enables analyses on protocols with adaptivity. In this work, we consider adaptive protocols under two different ``round-based models," which divide the $m$ samples into $m_1$ blocks, each of $m_0$ samples. 

Let us first consider the round-based model with adaptivity \emph{within} blocks in \Cref{sec:general_adaptivity_within_block}, where we have $m_1=\poly(n)$ blocks, each consisting of $m_0=\polylog(n)$ copies that can be measured adaptively. While adaptivity is allowed {\em within} the blocks, no adaptivity is allowed {\em between blocks}.

The first step is to note that if we can ensure a small $k$-th moment of the degree-$(\cD,k)$ advantage for each \emph{individual block}, we can hope to bound the degree-$(\cD,k)$ advantage of the overall protocol using the prior non-adaptive analysis, where we now replace individual copies of $\rho$ with \emph{blocks} of copies. So let us focus on the $k$-th moment of the degree-$(\cD,k)$ advantage for each individual block of $m_0$ copies. As the measurements within a block can be adaptively chosen, we cannot trace out the copies that are not included in $T$ when we are computing the Fourier coefficient $\hat{D}[T]$. To address this issue, we propose a delicate inductive argument. We divide all subsets $T\subset[m_0]$ into two classes: the subsets that contain the last copy and the subsets that do not. The contribution for the first class of $T$'s sums up to the degree-$(\cD,k)$ advantage of a block containing $m_0-1$ copies. To compute the contribution for the second class of $T$'s, we break all such $T$'s into a union of the last copy and the remaining copies in \eqref{eq:induction_within}. By one application of Cauchy-Schwarz, we can bound the contribution by a value that depends on the $2m_0$-th moment of a single-copy degree-$\mathcal{D}$ advantage and the $2(m_0-1)$-th moment of the likelihood ratio for a single-copy outputting a particular classical measurement outcome (see Eq.~\eqref{eq:induction_second_within}). By induction, we are able to compute the degree-$(\cD,k)$ advantage for each individual block of $m_0$ copies using these quantities. Finally, we generalize the above analysis to compute the degree-$(\cD,k)$ advantage of the whole protocol using the $2km_0$-th moment of a single-copy degree-$k$ advantage and the $2k(m_0-1)$-th moment of the likelihood ratio for a single copy outputting a particular classical measurement outcome, where the $k$ overhead in the order of moments comes from the fact that each $T$ can participate in at most $k$ blocks.

To prove the first bullet point of \Cref{thm:adaptivity_informal}, we employ the fact that a $2^{-\mathrm{poly}(k,\log n)}$-approximate state $O(km_0)$-design provides a bound on both the $2km_0$-th moment of a single-copy degree-$\cD$ advantage and the $2k(m_0-1)$-th moment of the likelihood ratio for a single copy outputting a particular classical measurement outcome. The detailed calculation is provided in \Cref{coro:general_adaptivity_within_block}.

\paragraph{Round-based model with adaptivity among blocks. }We also consider the round-based model with adaptivity among blocks in \Cref{sec:general_adaptivity_among_block}, where we have $m_1=\text{polylog}(n)$ adaptive blocks of $m_0=\text{poly}(n)$ non-adaptive copies within each block. At the start of each block, we can prepare a fixed measurement strategy to apply to the next $\mathrm{poly}(n)$ copies, where the strategy is chosen fully adaptively based on all readouts from the previous blocks. 

In this model, we use the combination of the recursive induction and the moment bounds in \emph{reverse order} relative to what we did in the previous model. Given a new block, as the measurements in this block can be adaptively predetermined based on the previous history, we compute the Fourier coefficient for block-wise subsets $T\subset[m_1]$. As a dual to the previous inductive argument, we divide all subsets $T\subset[m_1]$ into two classes: the subsets that contain the last block and the subsets that do not. The contribution for the first class of $T$'s sums up to the degree-$(\cD,k)$ advantage of another round-based protocol with adaptivity among blocks containing $m_1-1$ copies. To compute the contribution for the second class of $T$'s, we break all such $T$'s into a union of the last block and the remaining blocks in \eqref{eq:induction_within}. Again by an application of Cauchy-Schwarz, we can bound the contribution by a value that depends on the $2k$-th moment of a single-copy degree-$\cD$ advantage and the $2(m_1-1)$-th moment of the likelihood ratio for a single \emph{block} outputting a particular classical measurement outcome. To prove the second bullet point of \Cref{thm:adaptivity_informal}, we employ the fact that a $2^{-\mathrm{poly}(k,\log n)}$-approximate state $O(m_1)$-design provides a bound on both the $2k$-th moment of a single-copy degree-$\cD$ advantage and the $2(m_1-1)$-th moment of the likelihood ratio for a single block outputting a particular classical measurement outcome. The detailed calculation is provided in \Cref{coro:general_adaptivity_among_block}.

\paragraph{Learning random quantum circuits. }As an application, we show in \Cref{sec:circuit} that weak detection is degree-$\polylog(n)$ hard for non-adaptive and round-based adaptive protocols to distinguish between the maximally mixed state and a state prepared by a random polylogarithmic-depth geometrically local quantum circuit. This is because random polylogarithmic-depth geometrically local quantum circuits form a $\polylog$-design up to inverse superpolynomially small approximation factor according to our calculation in \Cref{sec:basic_circuit_moment} and~\cite{schuster2024random}.

\subsection{Quantum planted biclique}
\label{sec:overview_pds}

Here we sketch the key technical ingredients in our proof of low-degree hardness for quantum planted biclique. In this problem, one is either given $n$ copies of the maximally mixed state on $n$ $d$-dimensional qudits, or $n$ copies of a state which is given by planting $\lambda/n$ copies of the same Haar-random pure state $\rho$ into a random subset $S$ of the qudits and corrupting the state with global depolarizing noise. As mentioned in the introduction, when $d = 2$ and $\rho$ is simply $\ket{1}$, measuring these states in the computational basis will result in an instance of the classical planted biclique problem. 

For general $d$ and general measurements, the starting point is the fact that if one measures a $d$-dimensional Haar-random state in the computational basis, the likelihood of seeing one of the first $d/2$ measurement outcomes is roughly $1/2 + \Theta(1/\sqrt{d})$. The algorithmic upper bound for local measurements (see \Cref{lem:local_upper_pds}) simply uses this to reduce quantum planted biclique to an instance of classical \emph{planted dense subgraph} where the edges in the planted subgraph have elevated probability $1/2 + \Theta(1/\sqrt{d})$ whereas all other edges have probability $1/2$. A standard second moment calculation shows that this would give rise to the $n^{1/2}d^{1/4}$ scaling in Theorem~\ref{thm:pds_informal}, and indeed, an existing calculation for this classical problem (see Claim 8.7 in~\cite{brennan2020statistical}) would predict the same threshold on the lower bound side \emph{assuming the protocol simply measures every qudit in the computational basis and checks whether the readout is within $\{1,\ldots,d/2\}$ or $\{d/2+1,\ldots,d\}$}.

The key challenge lies in going from this rough intuition to a lower bound against all local measurements. \emph{A priori} the above strategy appears to be wasting significant information by throwing all measurement outcomes into one of two buckets, and furthermore it might seem advantageous to measure the different qudits in different bases. The content of our lower bound is that neither of these modifications would help, and that the naive protocol above is actually optimal. The proof is ultimately a delicate moment calculation, but here we highlight some ingredients unique to the quantum setting that may be of independent interest.

One point is that in the classical setting of planted dense subgraph, it is trivial to compute the Fourier coefficients of the likelihood ratio as an explicit function of the edges that appear in the overall graph. The Fourier coefficients of the \emph{expected} likelihood ratio, over the randomness of the planted structure, can be read off easily from this. In the quantum setting, this is far from the case as the likelihood ratio is ultimately an unwieldy multilinear polynomial in the measurement vectors. Instead of getting a clean expression for the likelihood ratio for a given $\rho, S$, we have to carry around this multilinear polynomial for most of the proof. When we expect over the planted structure at the end by taking a Haar integral, this results in a weighted linear combination of permutation operators applied to the measurement vectors. While similar expressions have appeared in the related context of \emph{purity testing} in~\cite{chen2022exponential}, there the challenge was to derive \emph{lower bounds} on these quantities, whereas we need to derive \emph{upper bounds}. Finally, another hurdle that we encounter is the need to work with the harmonic basis over $\mathbb{Z}_d$ instead of the Boolean Fourier basis, as local measurement of any qudit in the null hypothesis results in the uniform distribution over $\{1,\ldots,d\}$. Bounding the low-degree coefficients under this harmonic basis requires careful bookkeeping to avoid $d^k$ factors. We defer the technical details of all of this to Section~\ref{sec:lower_pds}.

\subsection{Quantum hardness for hybrid and classical problems}
We also show low-degree hardness for two problems, Learning Stabilizer with Noise and agnostic tomography of product states, by encoding tasks with classical or classical-quantum hybrid input that are known to be hard for classical low-degree functions into quantum input. We construct these encodings in a low-degree way to show even the average-case hardness of these tasks with encoded quantum input data in the quantum low-degree model.

First, the {\em Learning Stabilizers with Noise} ($\mathsf{LSN}$) problem has recently been proposed as a hardness assumption that generalizes the {\em Learning Parities with Noise} problem to the quantum setting \cite{poremba2024learning}. $\mathsf{LPN}$ is known to be classically low-degree hard in the worst-case, via a special case of the problem known as $k$-XORSAT where every row of the generator is $k$-sparse. 
The $k$-XORSAT problem, which corresponds to a worst-case noiseless parity learning problem, is hard for degree-$k$ algorithms but tractable for unbounded-degree algorithms. We show that $\mathsf{LSN}$ has these exact same properties, bolstering the analogy between $\mathsf{LSN}$ and $\mathsf{LPN}$. The proof of low-degree hardness proceeds by embedding $k$-XORSAT in a worst-case noiseless $\mathsf{LSN}$ instance and then performing an exact worst-to-average-case reduction. The embedding from classical data into quantum data is constructed as a low-degree function such that any quantum algorithm with degree-$k$ in our quantum low-degree model can be regarded as a classical function with degree-$2k$ on all the classical bits (see \Cref{sec:other_lsn_avg} for details). The reduction is necessary because $\mathsf{LSN}$ is an average-case problem.  

We then consider the task of agnostic tomography of product states, which, given an unknown state $\rho$ that is guaranteed to have a non-negligible overlap with a product state, outputs a product state with near-optimal overlap with $\rho$. Here, we encode into quantum data the classical task of tensor principal component analysis (PCA)~\cite{kunisky2019notes}, which distinguishes between a random Gaussian tensor in $\C^{n\times n\times n\times n}$, and a mixture of a random spike component and a Gaussian random component. Tensor PCA is known to be degree-$k$ hard if the ratio $\lambda$ between the spike signal and the Gaussian random component is bounded by $\lambda\leq O(n^{-1}k^{-1/2})$~\cite{kunisky2019notes}.  In \Cref{lem:normal_tensor_pca}, we first show the classical low-degree hardness for a \emph{normalized} version of tensor PCA where each tensor has a unit Frobenius norm through the concentration of Gaussian distribution. We then transfer the classical tensor into a quantum state using the encoding introduced in~\cite{bakshi2024learninga}, which we show in \Cref{lem:agnostic_reduction} is a low-degree encoding in the sense that every qubit is a low-degree function on the entries of the classical tensor. Finally, we show that a low-degree algorithm for finding the best product approximation to these encoded quantum states would yield a low-degree algorithm for normalized tensor PCA, a contradiction.

\section{Outlook}

Further refinements are possible to our framework and applications. Below, we mention some concrete future directions.

\paragraph{How to model full adaptivity.}  A key contribution of this work was to establish, for the first time, a framework for low-degree hardness against adaptive measurement strategies. Yet, here we can handle only a limited notion of adaptivity, where the adaptive choice of measurement basis depends only on a partial history of previous measurement outcomes. As we show in Section~\ref{sec:general_adaptivity_full}, a variety of natural ways of modeling adaptivity run into fundamental issues or are as hard as proving circuit lower bounds. We leave as an important open question how to satisfactorily model adaptivity within the low-degree framework.

\paragraph{Other quantum applications.}Finding further applications of our framework in different quantum tasks is an immediate open question. The candidate learning problems that one could hope to prove information-computation gaps for include $>1$-dimensional tensor network states, ground states of local Hamiltonians, shadow tomography with efficiently preparable observables, and agnostically learning shallow circuits.

\paragraph{Protocols with arbitrary POVMs.}For non-adaptive measurement strategies, our framework falls a hair short of encompassing algorithms that use fully general POVMs, as we can only deal with ancilla-assisted PVMs with at most $\lesssim n/11$ ancillas -- which do not suffice to simulate a general POVM.

\paragraph{Quantum hardness and advantage from encoding classical problems.}In this work, we encode classical problems that are known to be average-case hard in the classical low-degree framework into quantum data to show average-case hardness in the quantum low-degree framework. It is natural to wonder if we can find a general reduction from classical average-case hardness to quantum average-case hardness in the low-degree setting. Our work also opens up the exciting possibility of finding {\em low-degree} quantum advantage; that is, a classical problem, embedded into a quantum state, on which a low-degree quantum algorithm can perform better than a low-degree classical algorithm. It is also interesting to see if there are classical problems that are hard using \emph{any} quantum encoding, indicating their hardness for any quantum algorithms.

\paragraph{Mixture vs. mixture hypothesis testing problems.}The general framework established in this work focuses on a point vs. mixture distinguishing task. It is natural to ask if we can prove low-degree hardness for distinguishing between two random ensembles of quantum states, i.e. mixture vs. mixture state distinguishing tasks. This is of particular interest if one wants to extend our framework beyond single-copy measurements: as soon as one considers $t = 2$-copy measurements, SWAP test provides another example of a low-degree learning algorithm. Low-degree lower bounds for mixture vs. mixture testing have been proven in classical settings, e.g. in the work~\cite{rush2023easier}.

\paragraph{Distinguishing and learning quantum channels.}While our general framework focuses on hypothesis testing tasks among quantum states, it is interesting to see if our framework can be generalized to hypothesis testing tasks among quantum channels, where we should also take into account the computational complexity of preparing input quantum states.

\paragraph{Existence of an information-computation gap for quantum error mitigation.} In this work, we show that quantum error mitigation is low-degree hard for $1$D geometrically local noisy circuit of super-logarithmic depth and inverse polynomially small noise rate and protocols using single-qubit measurements. While there exists a statistically efficient algorithm using joint measurements for the hypothesis testing problem we construct for our lower bound~\cite{buadescu2021improved}, it is unknown whether there is a sample-efficient algorithm using single-qubit measurements. It thus remains open whether there is an actual information-computation gap in this setting.

\paragraph{Other quantum generalizations of classical inference problems.}The problem of quantum planted biclique that we proposed in this work is just one of many possible ``quantizations'' of classical inference problems. One could have also considered a quantum version of (non-bipartite) planted clique, a quantum version of distinguishing random graphs from random geometric graphs~\cite{bangachev2024fourier,liu2022testing}, or any of a number of other models with planted sparse structure. Even for the quantum planted biclique question, it remains unclear what the full complexity landscape of the problem is when one allows general single-copy measurements. Given the extent to which such problems have shaped our understanding of the complexity of classical learning, we expect such problems to be a fruitful testbed for future work in the quantum learning community.

\section{The general framework for non-adaptive measurements}\label{sec:general}

In this section we present our low-degree framework for quantum learning, specialized to the case of non-adaptive measurement strategies. As mentioned in Section~\ref{sec:basic_low_degree}, we will focus on \emph{hypothesis testing} problems. More specifically, for most of this work we will focus on one of the most pervasive classes of hypothesis testing problems within the quantum learning literature~\cite{odonnell2015quantum,chen2022exponential}: given access to copies of an unknown $n$-qubit state $\rho$ and given an ensemble $\mathcal{E}$ of ``structured states,'' distinguish between
\begin{itemize}
    \item $H_0$ (null case): $\rho$ is equal to the maximally mixed state $\rhomm=I/2^n$. 
    \item $H_1$ (alternative case): $\rho$ was sampled at the outset from $\mathcal{E}$
\end{itemize}

Our notion of low-degree hardness depends on the choice of a class of measurement strategies. In general, one could consider measurement strategies which in each round process $t$ copies of $\rho$ by performing a projective measurement on those copies and $n'$ ancilla qubits \--- by Naimark's theorem, we can implement arbitrary $2^{n'}$-outcome POVMs using a PVM with $n'$ ancillary qubits. Each such measurement outputs a $(tn + n')$-bit string, and we perform $m$ such rounds of measurement and collect $m(tn+n')$ bits. A \emph{degree-$k$} estimator in our setting is then any polynomial $f: \{0,1\}^{m(tn+n')} \to \R$ whose variance under the null case is normalized to $1$, such that each monomial in $f$ depends on at most $k$ bits. For any such $f$, we can consider the \emph{distinguishing advantage} achieved by $f$. Recall from Eq.~\eqref{eq:distinguishL2_k} that this is the absolute difference between the expectation of $f$ under the null case and the expectation under the alternative case, and that the optimal distinguishing advantage is given by the variance of the projection of the \emph{likelihood ratio} between alternative case and null case to the space of degree-$\le k$ polynomials (we will review this calculation in the warm-up below).

For a given class of measurements, specified by $t$ and $n'$ and miscellaneous additional constraints on what kinds of measurements are allowed, we say that the hypothesis testing problem is \emph{degree-$k$ hard} for that class (throughout, we will think of $k$ as $\omega(\log n)$) if \emph{for any sequence of (non-adaptively chosen) measurements from the class}, the optimal distinguishing advantage $\chi_{\le k}$ achieved by any degree-$k$ estimator is $o(1)$. Note that the classical low-degree model considered in \Cref{sec:basic_low_degree}, specialized to distributions over the Boolean cube and null case given by the uniform distribution, is the special case of the above framework when $\mathcal{E}$ only consists of states with diagonal density matrices. Indeed, the key challenge in the quantum setting is the additional quantification over which measurements are selected, whereas in the classical setting, the only measurements to consider are computational basis measurements.

Throughout this work, we will focus on the simple but practically relevant setting of single-copy projective measurements, i.e., where $t = 1$ (see \Cref{fig:model_single} for an illustration). Below, we first provide some warm-up calculations in the special case of local (i.e., single-\emph{qubit}) measurements (\Cref{sec:general_single_qubit}). We then turn to the main result of this section: if $\mathcal{E}$ forms an approximate $2$-design, then hypothesis testing is low-degree hard for general single-copy measurements that use up to $O(n)$ ancillary qubits (\Cref{sec:general_ancilla}).

\begin{figure}[htbp]
    \centering
    \includegraphics[width=0.80\linewidth]{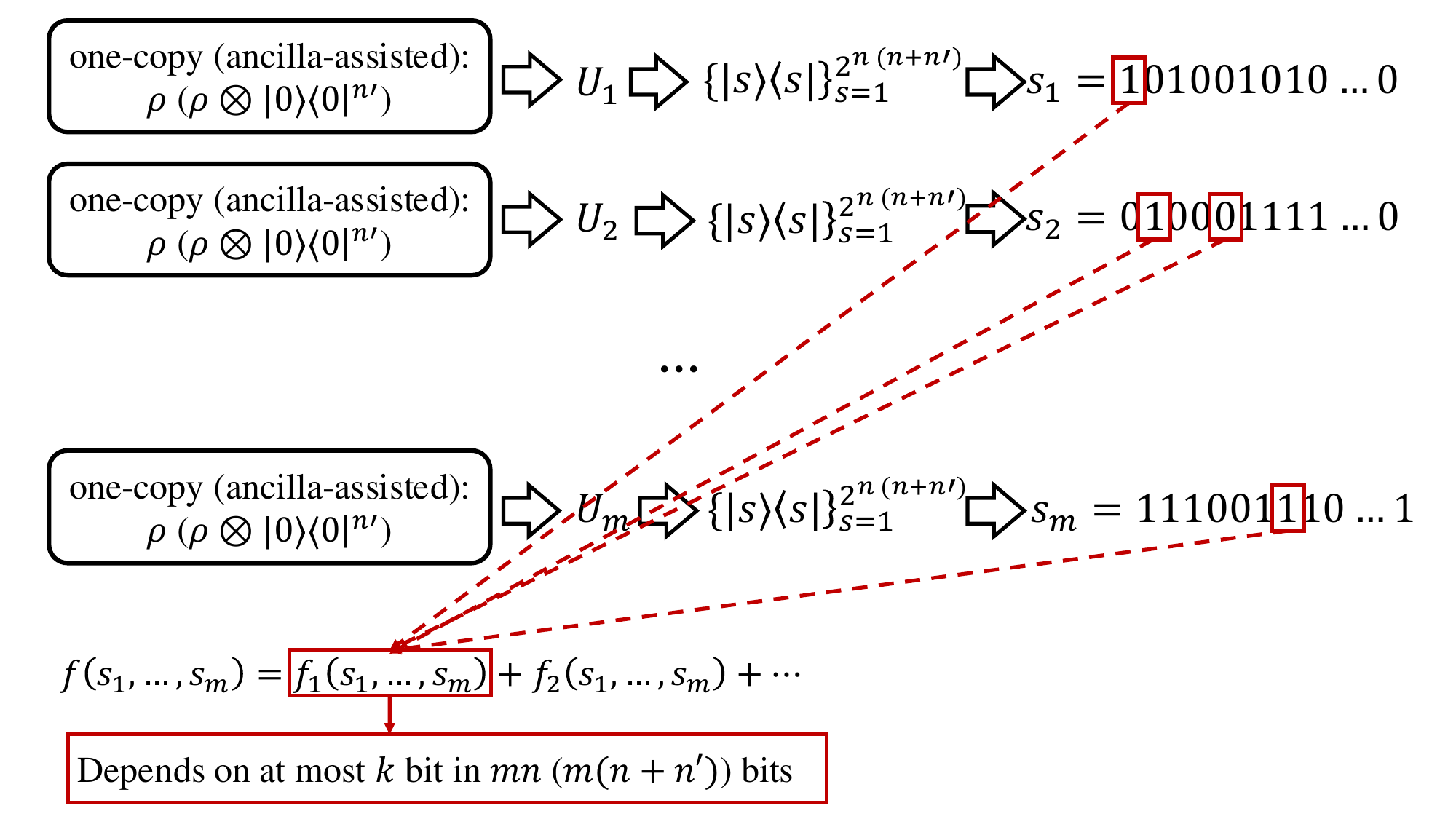}
    \caption{Non-adaptive single-copy PVM on one copy of $\rho$ (expressions in parentheses correspond to the case where the PVM involves $n'$ ancilla qubits).}
    \label{fig:model_single}
\end{figure}

\subsection{Warm-up: Low-degree hardness for local measurements}\label{sec:general_single_qubit}
We start with protocols that only use local (single-qubit) PVMs with zero ancilla, i.e. PVMs of the form $\{\ket{\phi_{s_i}}\bra{{\phi_{s_i}}}\}_{s_i}$ on the $i$-th copy, where $\ket{\phi_{s_i}}=\bigotimes_{r=1}^n\ket{\phi_{s_{i,r}}}$ and $s_i=(s_{i,1},...,s_{i,n}) \in \{0,1\}^n$. For any outcome $s_i$, the likelihood of observing $s_i$ in the null case is always the same: \begin{equation*}
    D_\emptyset(s_i)=\bra{\phi_{s_i}}\rhomm\ket{\phi_{s_i}}=2^{-n}\,.
\end{equation*}
The likelihood ratio between observing $s_i$ under the alternative case versus under the null case, for a specific choice of $\rho$ sampled from the ensemble $\mathcal{E}$, is thus given by
\begin{equation*}
    \overline{D}_{\rho}(s_i)=2^n\tr\left(\rho\bigotimes_{r=1}^n\ket{\phi_{s_{i,r}}}\bra{\phi_{s_{i,r}}}\right)\,,
\end{equation*}
and the likelihood ratio for observing a history $\bm{s} = (s_1,\ldots,s_m) \in \{0,1\}^{mn}$ after measuring $m$ copies of $\rho$ is given by
\begin{equation*}
    \overline{D}_{\rho,m}(\bm{s})=2^{mn}\prod_{i=1}^m\tr\left(\rho\bigotimes_{r=1}^n\ket{\phi_{s_{i,r}}}\bra{\phi_{s_{i,r}}}\right)\,.
\end{equation*}

Let $\chi_T(\bm{s})$ denote the Fourier basis function $\chi_T(\bm{s})=(-1)^{\sum_{(i,r)\in T}s_{i,r}}$, noting that because the distribution over histories under the null case is uniform, the functions $(\chi_T)$ form an orthonormal basis. We can then decompose the likelihood ratio for the whole history as 
\begin{align*}
\overline{D}_{\rho,m}(\bm{s})=\sum_{T\subseteq[mn]}\alpha_T^\rho\chi_T(\bm{s})\,, \qquad \text{for} \ \alpha_T^{\rho}=\mathbb{E}_{\bm{s}}\overline{D}_{\rho,m}(\bm{s})\chi_T(\bm{s})\,.
\end{align*}

Now we can consider the low-degree likelihood ratio $\overline{D}_{\rho,m}^{\leq k}(\bm{s})$, recalling from Section~\ref{sec:basic_low_degree} that this is the degree-$k$ Fourier truncation $\overline{D}_{\rho,m}^{\leq k}(\bm{s}) \triangleq \sum_{\abs{T}\leq k}\alpha_T^\rho\chi_T(\bm{s})$ and that the optimal distinguishing advantage achieved by any degree-$k$ estimator is given by:
\begin{align*}
\chi^2_{\le k} \triangleq \norm{\mathbb{E}_{\rho\sim\mathcal{E}}\overline{D}_{\rho,m}^{\leq k}-\textbf{1}}^2
=\sum_{\abs{T}\leq k,T\neq\emptyset}\left(\E_\rho[\alpha^{\rho}_T]\right)^2\,.
\end{align*}
In the following, we consider the Fourier coefficient $\hat{D}[T] \triangleq \E_\rho[\alpha^{\rho}_T]$ for a particular $T\subseteq[mn]$ and $\abs{T}\leq k$. We denote $\bm{s}[T]$ to be the set of $s_{i,r}$ with $(i,r)\in T$. Note that the parts not touched by $T$ get traced out when we average over histories, yielding:
\begin{align}\label{eq:Fourier_single_qubit}
\hat{D}[T]=\E_\rho[\alpha^{\rho}_T]=\E_{\rho\sim\mathcal{E}}\E_{\bm{s}}\overline{D}_{\rho,m}(\bm{s})\chi_T(\bm{s})=\E_{\rho\sim\mathcal{E}}\E_{\bm{s}[T]}2^{\abs{T}}\tr\left(\underset{[mn]\backslash T}{\tr}(\rho^{\otimes m})\bigotimes_{(i,r)\in T}|\phi_{s_{i,r}}\rangle\langle\phi_{s_{i,r}}|\right).
\end{align}
To show low-degree hardness, we wish to show that the sum of squares of these quantities over $|T| \le k$ is small. Motivated by this, we consider the following:

\begin{condition}[Bounded $k$-local likelihood ratio]\label{assume:local_likelihood}
For a hypothesis testing problem $\rhomm$ vs. $\rho\sim\mathcal{E}$, define the \emph{$k$-local likelihood ratio} as
\begin{align*}
k\text{-LLR} := \max_{\substack{T\in [mn]:|T|=k \\ \Lambda \in \text{single-qubit PVMs}}} \norm{\mathbb{E}_{\rho\sim\mathcal{E}}\overline{D}_{\tr_{[mn]\backslash T}(\rho^{\otimes m})}-\textbf{1}}^2\,,
\end{align*}
where we maximize over all size-$T$ subsets of the total $mn$ qubits, and over all single-qubit PVM measurement strategies $\Lambda$ (such strategies apply the PVM $\{\ket{\phi_{s_i}}\bra{{\phi_{s_i}}}\}_{s_i \in \{0,1\}^n}$ on the $i$-th copy), and where the quantity $\overline{D}_{\tr_{[mn]\backslash T}(\rho'^{\otimes m})}$ is given by
\begin{align*}
    \overline{D}_{\tr_{[mn]\backslash T}(\rho^{\otimes m})} \triangleq 2^{\abs{T}}\tr\left(\underset{[mn]\backslash T}{\tr}(\rho^{\otimes m})\bigotimes_{(i,r)\in T}|\phi_{s_{i,r}}\rangle\langle\phi_{s_{i,r}}|\right).
\end{align*}
\end{condition} 

\noindent Often it is more convenient to check that the following stronger quantity, which we call \emph{$k$-local indistinguishability}, is bounded:
\begin{equation}
d_{\tr}\left(\mathbb{E}_{\rho\sim\mathcal{E}}\tr_{[mn]\backslash T}\left(\rho^{\otimes m}\right),\frac{I}{2^{\abs{T}}}\right)\,, \label{eq:indist}
\end{equation}
i.e. any $k$-qubit subsystem of the $m$ copies is close to maximally mixed. This is the natural quantum analogue of $k$-wise independence. It is easy to show that that a bound on $k$-local distinguishability also implies a bound on $k$-LLR 
(see \Cref{sec:coro_kLLR} for the proof):
\begin{corollary}[$k$-local indistinguishability implies $k$-LLR condition]\label{coro:general_single_qubit}
If the hypothesis testing problem $\rhomm$ vs. $\rho\sim\mathcal{E}$ satisfies $d_{\tr}\left(\mathbb{E}_{\rho\sim\mathcal{E}}\tr_{[mn]\backslash T}\left(\rho^{\otimes m}\right),\frac{I}{2^{\abs{T}}}\right)\leq \epsilon$ for any $T\in[mn]$ and $\abs{T}\leq k$, then $k\text{-LLR} \leq \epsilon^2 2^{2k}$.
\end{corollary}

\noindent We can now conclude this warm-up section with the basic result that bounded $k$-LLR implies low-degree hardness. Later we will give examples of natural problems that satisfy this condition.

\begin{proposition}[Low-degree hardness for single-qubit measurements]\label{thm:general_single_qubit}
For a hypothesis testing problem $\rhomm$ vs. $\rho\sim\mathcal{E}$ whose $k$-LLR is bounded by $2^{-\tOmega(k\log n)}$, weak detection is degree-$k$ hard for protocols that make non-adaptive single-qubit measurements on $\mathrm{poly}(n)$ copies of $\rho$.
\end{proposition}
\begin{proof}
Suppose the $k$-LLR is bounded by $\epsilon^2$. Then 
\begin{align*}
\mathbb{E}_{\rho,\rho'\sim\mathcal{E}}\mathbb{E}_{\bm{s},\bm{s}'}\overline{D}_{\rho,m}(\bm{s})\overline{D}_{\rho',m}(\bm{s}')\chi_T(\bm{s})\chi_T(\bm{s}')\leq\norm{\mathbb{E}_{\rho\sim\mathcal{E}}\overline{D}_{\tr_{[mn]\backslash T}(\rho^{\otimes m})}-1}^2 < \epsilon^2.
\end{align*}
Therefore, we can bound the square degree-$k$ advantage by
\begin{align*}
\norm{\mathbb{E}_{\rho\sim\mathcal{E}}\overline{D}_{\rho,m}^{\leq k}-\textbf{1}}^2&=\sum_{\abs{T}\leq k,T\neq\emptyset}\mathbb{E}_{\rho,\rho'\sim\mathcal{E}}\mathbb{E}_{\bm{s},\bm{s}'}\overline{D}_{\rho,m}(\bm{s})\overline{D}_{\rho',m}(\bm{s}')\chi_T(\bm{s})\chi_T(\bm{s}')\\
&=\epsilon^2\sum_{i=1}^{k}\binom{mn}{i} \leq\epsilon^2k(mn)^k\,.
\end{align*}
Choosing $\epsilon=2^{-\tOmega(k\log n)}$ we can finish the proof.
\end{proof}

\noindent By \Cref{coro:general_single_qubit}, the same conclusion (degree-$k$ hardness of weak detection) follows when the $k$-local distinguishability measure in Eq.~\eqref{eq:indist} is bounded by $\epsilon^2$. In addition, we are able to extend this result beyond local measurements to arbitrary measurements prepared by bounded-depth circuits using light cone arguments:
\begin{corollary}[low-degree hardness for measurements with bounded depth]\label{coro:general_bounded_depth}
For the hypothesis testing problem $\rhomm$ vs. $\rho\sim\mathcal{E}$ satisfying
\begin{align*}
d_{\tr}\left(\mathbb{E}_{\rho\sim\mathcal{E}}\tr_{[mn]\backslash T}\left(\rho^{\otimes m}\right),\frac{I}{2^{\abs{T}}}\right)\leq 2^{-\tOmega(k\log n)}
\end{align*}
for any $T\in[mn]$ and $\abs{T}\leq 2^Lk$ ($\abs{T}\leq(2L+k)^{\mathscr{D}}$), weak detection is degree-$k$ hard for any protocol that makes polynomially many non-adaptive measurements prepared by ($\mathscr{D}$-dimensional geometrically-local) circuits of depth at most $L$.
\end{corollary}
\noindent We defer the details of this to~\Cref{sec:general_bounded_depth}:

\begin{remark}
     The reader might wonder how the low-degree framework for \emph{local} measurements differs from a more naive setting where one simply gets access to $\mathrm{poly}(n)$ many size-$\log(n)$ patches of $\rho^{\otimes m}$, performs full tomography on those patches in polynomial time, and post-processes the resulting classical information in polynomial time. In the latter model, it is trivial to show that local indistinguishability implies hardness: if the patches are $1/\mathrm{poly}(n)$-close to maximally mixed, then no classical post-processing, no matter how powerful, can extract information about the underlying state.
     
    We clarify that our low-degree model is strictly more powerful than this latter model. For instance, in the latter model, if one performs tomography on a deterministic choice of patches and only those patches are indistinguishable whereas other patches of $\rho^{\otimes m}$ are not, then low-degree protocols would be able to detect them but the weaker model would not. On the other hand, even if algorithms in the weaker model could choose patches randomly, the resulting estimator will still correspond to a convex combination of polynomial functions that are $\poly(n)$-sparse, degree-$\log(n)$, and variance-$1$, which need not encompass all degree-$\log(n)$, variance-$1$ polynomial functions that would be considered in our low-degree model.
\end{remark}

\subsection{\texorpdfstring{$2$}{2}-design implies low-degree hardness for single-copy measurements}\label{sec:general_ancilla}

We now present the main result of this section: a sufficient condition for low-degree hardness for the most general non-adaptive setting where one can perform arbitrary single-copy PVMs with some number of ancillas, as depicted in \Cref{fig:model_single}. Unlike in the previous subsection, each PVM can now arbitrarily entangle different qubits in a system of size $n + n'$.

More formally, we consider protocols using $m=\text{poly}(n)$ copies of $\rho$ and applying ancilla-assisted PVMs on $\rho\otimes\ketbra{0}^{\otimes n'}$. We abuse the notation slightly and denote the measurement for each round as $\{\ket{\phi_{s_i}}\bra{\phi_{s_i}}\}_{s_i}$ for $i=1,...,m$, wher the $s_i$'s are $n+n'$-bit strings. Given an output $s_i$, the likelihoods under the null case and the alternative case for a fixed $\rho$ from the ensemble are given by
\begin{align*}
D_\emptyset(s_i)=\tr(\rhomm\otimes\ketbra{0}^{\otimes n'}\cdot\ketbra{\phi_{s_i}}),\quad D_{\rho}(s_i)=\tr(\rho\otimes\ketbra{0}^{\otimes n'}\cdot\ketbra{\phi_{s_i}}).
\end{align*}
Therefore, the likelihood ratio for each copy is given by
\begin{align*}
\overline{D}_{\rho}(s_i)=\frac{\tr(\rho\otimes\ketbra{0}^{\otimes n'}\cdot\ketbra{\phi_{s_i}})}{\tr(\rhomm\otimes\ketbra{0}^{\otimes n'}\cdot\ketbra{\phi_{s_i}})}
\end{align*}
for $\rho\sim \mathcal{E}$. Given a history $\bm{s}=(s_1,...,s_m)$, the likelihood ratio for the whole history is given as 
\begin{align*}
\overline{D}_{\rho,m}(\bm{s})=\prod_{i=1}^m\frac{\tr(\rho\otimes\ketbra{0}^{\otimes n'}\cdot\ketbra{\phi_{s_i}})}{\tr(\rhomm\otimes\ketbra{0}^{\otimes n'}\cdot\ketbra{\phi_{s_i}})}.
\end{align*}

\noindent We first show that a generalization of the notion of local indistinguishability from the previous section implies low-degree hardness for ancilla-assisted PVMs:

\begin{theorem}[Low-degree hardness for ancilla-assisted PVMs]\label{thm:general_ancilla}
For the hypothesis testing problem $\rhomm$ vs. $\rho\sim\mathcal{E}$ satisfying
\begin{align*}
d_{\tr}\left(\mathbb{E}_{\rho\sim\mathcal{E}}\underset{[m(n+n')]\backslash T}{\tr}\left(\bigotimes_{i=1}^{m}U_i\left(\rho\otimes\ketbra{0}^{\otimes n'}\right) U_i^\dagger\right),\underset{[m(n+n')]\backslash T}{\tr}\left(\bigotimes_{i=1}^{m}U_i\left(\rhomm\otimes\ketbra{0}^{\otimes n'}\right) U_i^\dagger\right)\right)\leq 2^{-\tilde{\Omega}(k\log n)},
\end{align*}
for some $0<n'\leq n$, any $T\in[m(n+n')]$, and $\abs{T}\leq k$, weak detection is degree-$k$ hard for any protocol that makes polynomially many non-adaptive PVMs $\{\{U_i\ket{x}\bra{x}U_i^{\dagger}\}_{x=1}^{2^{n+n'}}\}_{i=1}^m$ on a single copy of $\rho$ and $n'$ ancillary qubits.
\end{theorem}

\noindent The proof is similar to \Cref{thm:general_single_qubit}, see \Cref{sec:pf_general_ancilla} for details.

With this in hand, we now show the main result of this section, namely that when $\rho\sim\mathcal{E}$ forms an approximate state $2$-design, the hypothesis testing problem satisfies the condition in \Cref{thm:general_ancilla} and is thus low-degree hard for protocols with non-adaptive PVMs using a bounded number $n'$ of ancillary qubits. 

\begin{theorem}[Approximate state $2$-design indicates low-degree hardness for single-copy measurements]\label{coro:general_ancilla}
For the hypothesis testing problem $\rhomm$ vs. $\rho\sim\mathcal{E}$ with $\mathcal{E}$ a $\epsilon$-approximate state $2$-design for $\epsilon\leq 2^{-\tOmega(k\log n)}$, weak detection is degree-$k$ hard for any protocol that makes polynomially many non-adaptive PVMs $\{\{U_i\ket{x}\bra{x}U_i^{\dagger}\}_{x=1}^{2^{n+n'}}\}_{i=1}^m$ acting on a single copy and at most $n'=\Theta(n)$ ancillary qubits.
\end{theorem}

\begin{proof}
We consider each single copy with $n'$ ancillary qubits, for each fixed unitary $U\in SU(2^n)$ we have
\begin{align*}
&\E_\psi\norm{\tr_{[n+n']\backslash A}\left[U\left(\ketbra{\psi}\otimes\ketbra{0}^{\otimes n'}\right)U^\dagger\right]-\tr_{[n+n']\backslash A}\left[U\left(\frac{I}{2^n}\otimes\ketbra{0}^{\otimes n'}\right)U^\dagger\right]}_2^2\\
\leq&\E_\psi\tr_{[n+n']\backslash A}^2\left[U\left(\ketbra{\psi}\otimes\ketbra{0}^{\otimes n'}\right)U^\dagger\right]+\tr_{[n+n']\backslash A}^2\left[U\left(\frac{I}{2^n}\otimes\ketbra{0}^{\otimes n'}\right)U^\dagger\right]-\\
&\quad 2\E_\psi\tr_{[n+n']\backslash A}\left[U\left(\ketbra{\psi}\otimes\ketbra{0}^{\otimes n'}\right)U^\dagger\right]\tr_{[n+n']\backslash A}\left[U\left(\frac{I}{2^n}\otimes\ketbra{0}^{\otimes n'}\right)U^\dagger\right]+3\epsilon\\
=&\E_\psi\tr_{[n+n']\backslash A}^2\left[U\left(\ketbra{\psi}\otimes\ketbra{0}^{\otimes n'}\right)U^\dagger\right]-\tr_{[n+n']\backslash A}^2\left[U\left(\frac{I}{2^n}\otimes\ketbra{0}^{\otimes n'}\right)U^\dagger\right]+3\epsilon\\
\leq&\tr\left[U\otimes U\left(\left(\frac{\SWAP_{12}}{2^n(2^n+1)}-\frac{I_{2n}}{4^n(2^n+1)}\right)\otimes\ketbra{0}^{\otimes 2n'}\right)U^\dagger\otimes U^\dagger\cdot I_{[n+n']\backslash T}^{\otimes 2}\otimes \SWAP_{T}\right]+3\epsilon\\
\leq&\tr\left[U\otimes U\left(\left(\frac{\SWAP_{12}}{2^n(2^n+1)}\right)\otimes\ketbra{0}^{\otimes 2n'}\right)U^\dagger\otimes U^\dagger\cdot I_{[n+n']\backslash T}^{\otimes 2}\otimes \SWAP_{T}\right]+2^{-n}+3\epsilon,
\end{align*}
where the second inequality follows from the fact that the first moment of Haar random states is the same as the maximally mixed state.

Note that we can always decompose any $n$-qubit unitary $U$ as:
\begin{align*}
U=\sum_i\alpha_iV_i\otimes W_i
\end{align*}
with at most constant terms, $\abs{\alpha_i}\leq 1$, and $V_i\in SU(2)$ and $W_i\in SU(2^{n-1})$ are single- and $n-1$-qubit unitaries. Some examples along this line include the quantum Shannon decomposition and the Cosine-Sine Decomposition~\cite{shende2005synthesis}, which provides this constant as $8$. We can repeat this procedure for $n'$ times and have
\begin{align*}
U=\sum_i\alpha_iV_i\otimes W_i
\end{align*}
with at most $2^{3n'}$ terms, $\abs{\alpha_i}\leq 1$, and $V_i\in SU(2^{n'})$ and $W_i\in SU(2^{n-n'})$ are $n'$- and $n-n'$-qubit unitaries. We thus have
\begin{align*}
&\tr\left[U\otimes U\left(\left(\frac{\SWAP_{12}}{2^n(2^n+1)}\right)\otimes\ketbra{0}^{\otimes 2n'}\right)U^\dagger\otimes U^\dagger\cdot I_{[n+n']\backslash T}^{\otimes 2}\otimes \SWAP_{T}\right]\\
=&\sum_{ijkl}\alpha_i\alpha_j\alpha_k^*\alpha_l^*\tr\left[W_i\otimes W_j\left(\frac{\SWAP_{12}}{2^n(2^n+1)}\right) W_k^\dagger\otimes W_l^\dagger\otimes V_i\otimes V_j\left(\ketbra{0}^{\otimes 2n'}\right)V_k^\dagger\otimes V_l^\dagger\cdot I_{[n+n']\backslash T}^{\otimes 2}\otimes \SWAP_{T}\right]\\
\leq& \sum_{ijkl}\alpha_i\alpha_j\alpha_k^*\alpha_l^*2^{\abs{T}-n-n'}\leq 2^{\abs{T}+11n'-n}.
\end{align*}
Thus, we have 
\begin{align*}
&\E_\psi\norm{\tr_{[n+n']\backslash A}\left[U\left(\ketbra{\psi}\otimes\ketbra{0}^{\otimes n'}\right)U^\dagger\right]-\tr_{[n+n']\backslash A}\left[U\left(\frac{I}{2^n}\otimes\ketbra{0}^{\otimes n'}\right)U^\dagger\right]}_2^2\\
\leq&2\cdot 2^{\abs{A}+11n'-n}+3\epsilon.
\end{align*}
Then by Markov's inequality, we have
\begin{align*}
&\Pr\left[\norm{\tr_{[n+n']\backslash A}\left[U\left(\ketbra{\psi}\otimes\ketbra{0}^{\otimes n'}\right)U^\dagger\right]-\tr_{[n+n']\backslash A}\left[U\left(\frac{I}{2^n}\otimes\ketbra{0}^{\otimes n'}\right)U^\dagger\right]}_2\geq\left(2\cdot2^{\abs{A}+11n'-n}+3\epsilon\right)^{1/4}\right]\\
\leq&\left(2\cdot2^{\abs{A}+11n'-n}+3\epsilon\right)^{1/2}.
\end{align*}
By the norm inequality, we have
\begin{align*}
&\Pr\left[d_{\tr}\left(\tr_{[n+n']\backslash A}\left[U\left(\ketbra{\psi}\otimes\ketbra{0}^{\otimes n'}\right)U^\dagger\right],\tr_{[n+n']\backslash A}\left[U\left(\frac{I}{2^n}\otimes\ketbra{0}^{\otimes n'}\right)U^\dagger\right]\right)\geq2^{\abs{A}/2}\left(2\cdot2^{\abs{A}+11n'-n}+3\epsilon\right)^{1/4}\right]\\
&\leq\left(2\cdot2^{\abs{A}+11n'-n}+3\epsilon\right)^{1/2},
\end{align*}
for any fixed $U$. 

We can now consider the full scenario on $m$ copies of $[n]$ and denote $T=T_1\cup\ldots\cup T_m$, where $T_i$ denotes the elements of $T$ that are contained within the $i$th copy. We always have $\abs{T_i}\leq\abs{T}$. Consider matrices $A,B,C,D$ with $\norm{A}_1=\norm{D}_1=1$. Then,
\begin{align*}
\norm{A\otimes B-C\otimes D}_1&\leq\norm{A\otimes (B-D)+(A-C)\otimes D}_1\leq\norm{A}_1\norm{B-D}_1+\norm{A-C}_1\norm{D}_1\leq\norm{B-D}_1+\norm{A-C}_1.
\end{align*}
Applying this bound $m$ times and then use the union bound and the triangle inequality, we obtain:
\begin{align}\label{eq:trace_dis_ancilla}
\begin{split}
&d_{\tr}\left(\mathbb{E}_\psi\underset{[m(n+n')]\backslash T}{\tr}\left(\bigotimes_{i=1}^{m}U_i\left(\ketbra{\psi}\otimes\ketbra{0}^{\otimes n'}\right) U_i^\dagger\right),\underset{[m(n+n')]\backslash T}{\tr}\left(\bigotimes_{i=1}^{m}U_i\left(\frac{I}{2^n}\otimes\ketbra{0}^{\otimes n'}\right) U_i^\dagger\right)\right)\\
&\leq m\cdot 2^{\abs{T}/2}\left(2\cdot2^{\abs{T}+11n'-n}+3\epsilon\right)^{1/4},
\end{split}
\end{align}
with probability at least $1-m\cdot\left(2\cdot2^{\abs{T}+11n'-n}+3\epsilon\right)^{1/2}$. Finally, note that $\abs{T}\leq k$, we only need to choose $n'\leq (n-\tOmega(k\log n))/11=\Theta(n)$ and $\epsilon=2^{-\tOmega(k\log n)}$, which will guarantee that
\begin{align*}
d_{\tr}\left(\mathbb{E}_\psi\underset{[m(n+n')]\backslash T}{\tr}\left(\bigotimes_{i=1}^{m}U_i\left(\ketbra{\psi}\otimes\ketbra{0}^{\otimes n'}\right) U_i^\dagger\right),\underset{[m(n+n')]\backslash T}{\tr}\left(\bigotimes_{i=1}^{m}U_i\left(\frac{I}{2^n}\otimes\ketbra{0}^{\otimes n'}\right) U_i^\dagger\right)\right)
\end{align*}
is bounded by $2^{-\tOmega(k\log n)}$ and finishes the proof for \Cref{coro:general_ancilla}.
\end{proof}

\section{Frameworks for adaptive measurements}
\label{sec:adapt}

A feature which is largely unique to the quantum setting is the ability for the learner to perform \emph{adaptively chosen} measurements. This turns out to be quite subtle to model within the low-degree framework; indeed, in Section~\ref{sec:general_adaptivity_full} we catalog a number of no-go results for natural attempts. We remark that in classical settings where the low-degree method has been applied and where there is the possibility of adaptive choice~\cite{mardia2024low,bangachev2024fourier}, it is also open how to model adaptivity.

In this section, we provide three frameworks for reasoning about measurement strategies with some nontrivial level of adaptivity. To do this, we first need to turn to a stronger \emph{copy-wise} notion of low-degree, which is the quantum analogue of the classical sample-wise notion mentioned at the end of Section~\ref{sec:basic_low_degree}. We provide an overview of this stronger notion in~\Cref{sec:general_copywise} for details (see~\Cref{fig:model_copy}(a)); roughly, low-degree protocols in this category compute functions that are linear combination of terms that have degree at most $\cD$ in the bits of each copy, while each depending on at most $k$ copies.

Under this stronger copy-wise model, we provide two \emph{round-based} models of adaptivity, which divide all $m$ samples into $m_1$ rounds each of size $m_0$. In particular, we consider two kinds of round-based models. The first model has adaptivity with each block of size $m_0=\polylog(n)$ (see \Cref{sec:general_adaptivity_within_block} for details), and the second model has adaptivity among $m_1=\polylog(n)$ blocks (see \Cref{sec:general_adaptivity_among_block} for details).  
Finally, we also consider protocols based on statistical query (SQ) post-processing of readouts from non-adaptive measurements. In this model, while the measurements are still non-adaptive, the statistical queries, which can be thought of as an (arbitrarily high-degree) coarsening of our measurements to
two-outcome observables, are chosen adaptively. We will show a general equivalence between SQ models and non-adaptive copy-wise low-degree models in the presence of readout noise in \Cref{sec:general_equivalence}, building on the classical result of~\cite{brennan2020statistical}. 

Finally, in \Cref{sec:general_adaptivity_full}, we explain how several natural notions of adaptivity in the low-degree method cannot be used to prove meaningful lower bounds.

\begin{figure}[htbp]
    \centering
    \includegraphics[width=1.00\linewidth]{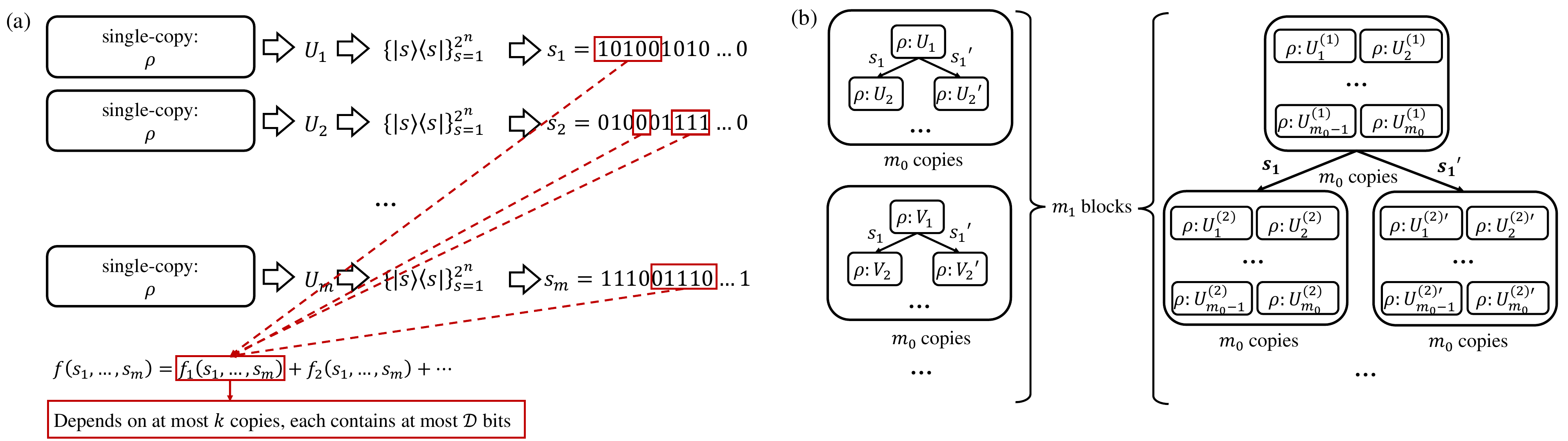}
    \caption{(a) The non-adaptive copy-wise low-degree model where each term of the final function depends on at most $k$ copies each of degree at most $\cD$. (b) Protocols using single-copy PVMs with restrictive adaptivity (round-based models), which can be decomposed into $m_1$ blocks each of $m_0$ copies. In the left model, we allow adaptivity within each block. In the right model, we allow adaptivity among blocks.}
    \label{fig:model_copy}
\end{figure}

\subsection{Copy-wise low-degree model} \label{sec:general_copywise}
First, we introduce the copy-wise low-degree model. Recall that we consider protocols using $m=\text{poly}(n)$ copies of $\rho$ and applying a single-copy PVM on each copy. Without loss of generality, we can denote the measurement for each round as $\{\ket{\phi_{s_i}}\bra{\phi_{s_i}}\}_{s_i}$ for $i=1,...,m$. We will assume that $s_i$ can be any $n$-bit string by \Cref{sec:basic_quantum}. Given an output $s_i$, the likelihood ratio for each copy is given by
\begin{align*}
\overline{D}_{\rho}(s_i)=2^n\tr(\rho\ket{\phi_{s_i}}\bra{\phi_{s_i}})
\end{align*}
for $\rho\sim \mathcal{E}$. Given a history $\bm{s}=(s_1,...,s_m)$, the likelihood ratio for the whole history is given as 
\begin{align*}
\overline{D}_{\rho,m}(\bm{s})=2^{mn}\prod_{i=1}^m\tr(\rho\ket{\phi_{s_i}}\bra{\phi_{s_i}}).
\end{align*}
Here, we consider copy-wise $(\cD,k)$-low-degree likelihood ratio $D_{\rho,m}^{\leq\cD,k}$ instead of the low-degree likelihood ratio $D_{\rho,m}^{\leq k}$ considered in the previous parts. Recall from the end of~\Cref{sec:basic_low_degree} that this is given by restricting the Fourier expansion of $D_{\rho,m}^{\leq\cD,k}$ to the subspace spanned by Fourier characters which each have degree at most $\mathcal{D}$ in each bit and nonzero degree in at most $k$ of the bits. The square degree-$(\cD,k)$ advantage can be computed as
\begin{align}\label{eq:holder}
\begin{split}
\norm{\mathbb{E}_{\rho\sim\mathcal{E}}\overline{D}_{\rho,m}^{\leq \cD,k}-\textbf{1}}^2&=\mathbb{E}_{\rho,\rho'\sim\mathcal{E}}\expval{\overline{D}_{\rho,m}^{\leq \cD,k},\overline{D}_{\rho',m}^{\leq \cD,k}}-1\\
&=\mathbb{E}_{\rho,\rho'\sim\mathcal{E}}\expval{\left(\bigotimes_{i=1}^{m}\overline{D}^{(i)}_{\rho}\right)^{\leq \cD,k},\left(\bigotimes_{i=1}^{m}\overline{D}^{(i)}_{\rho'}\right)^{\leq \cD,k}}-1\\
&=\mathbb{E}_{\rho,\rho'\sim\mathcal{E}}\sum_{T\in[m],\abs{T}\leq k,T\neq\emptyset}\prod_{i\in T}\left(\expval{\overline{D}_{\rho}^{(i)\leq \cD},\overline{D}_{\rho'}^{(i)\leq \cD}}-1\right)\\
&\leq\sum_{T\in[m],\abs{T}\leq k,T\neq\emptyset}\left(\prod_{i\in T}\mathbb{E}_{\rho,\rho'\sim\mathcal{E}}\left(\expval{\overline{D}_{\rho}^{(i)\leq \cD},\overline{D}_{\rho'}^{(i)\leq \cD}}-1\right)^{\abs{T}}\right)^{1/\abs{T}}\\
&\leq\sum_{T\in[m],\abs{T}\leq k,T\neq\emptyset}\mathbb{E}_{\rho,\rho'\sim\mathcal{E}}\left(\expval{\overline{D}_{\rho}^{\leq \cD},\overline{D}_{\rho'}^{\leq \cD}}-1\right)^{\abs{T}},
\end{split}
\end{align}
where $\overline{D}^{(i)}_{\rho}$ in the second line is the likelihood ratio for the $i$-th sample, the third equality follows almost identically to the proof of Claim 3.3 in \cite{brennan2020statistical}, the fourth line follows from H\"{o}lder's inequality, and $\overline{D}_{\rho}$ in the last line is the likelihood ratio for any single-copy measurements. 

Similar to \Cref{assume:local_likelihood}, we also identify a condition under which the hypothesis testing problem is low-degree hard under the copy-wise model. 
\begin{condition}[Bounded $(\cD,k)$-likelihood ratio for each copy]\label{assume:copy_likelihood}The hypothesis testing problem $\rhomm$ vs. $\rho\sim\mathcal{E}$ has an $\epsilon$-bounded $(\cD,k)$-likelihood ratio for each copy if for any single-copy measurement strategy, we have
\begin{align*}
\mathbb{E}_{\rho,\rho'\sim\mathcal{E}}\left(\expval{\overline{D}_{\rho}^{\leq \cD},\overline{D}_{\rho'}^{\leq \cD}}-1\right)^{k}\leq\epsilon.
\end{align*}
\end{condition}

\begin{theorem}[Copy-wise low-degree hardness for single-copy measurements]\label{thm:general_sample}
For a hypothesis testing problem $\rhomm$ vs. $\rho\sim\mathcal{E}$ satisfying \Cref{assume:copy_likelihood} at $\epsilon=2^{-\tOmega(k\log n)}$, i.e., 
\begin{align*}
\mathbb{E}_{\rho,\rho'\sim\mathcal{E}}\left(\expval{\overline{D}_{\rho}^{\leq \cD},\overline{D}_{\rho'}^{\leq \cD}}-1\right)^{k}\leq2^{-\tOmega(k\log n)}
\end{align*}
for any single-copy measurement, weak detection is degree-$(\mathcal{D},k)$ hard for any protocol that uses polynomially many non-adaptive single-copy PVMs $\{\{U_i\ket{x}\bra{x}U_i^{\dagger}\}_{x=1}^{2^n}\}_{i=1}^m$.
\end{theorem}
\noindent We note that the hardness results in the above theorem is the same as \Cref{thm:general_ancilla} while \Cref{assume:copy_likelihood} is stronger than approximate state $2$-designs. Therefore, \Cref{thm:general_sample} is actually a weaker theorem than \Cref{thm:general_ancilla}.
\begin{proof}
We assume hypothesis testing problem $\rhomm$ vs. $\rho\sim\mathcal{E}$ satisfies $\epsilon$ $(\cD,k)$-likelihood ratio for each copy in \Cref{assume:copy_likelihood}, i.e., 
\begin{align*}
\mathbb{E}_{\rho,\rho'\sim\mathcal{E}}\left(\expval{\overline{D}_{\rho}^{\leq \cD},\overline{D}_{\rho'}^{\leq \cD}}-1\right)^{k}\leq\epsilon
\end{align*}
for any single-copy measurement. We have the square degree-$(\cD,k)$ advantage:
\begin{align*}
\norm{\mathbb{E}_{\rho\sim\mathcal{E}}\overline{D}_{\rho,m}^{\leq \cD,k}-\textbf{1}}^2&\leq\sum_{T\in[m],\abs{T}\leq k,T\neq\emptyset}\left(\prod_{i\in T}\mathbb{E}_{\rho,\rho'\sim\mathcal{E}}\left(\expval{\overline{D}_{\rho}^{(i)\leq \cD},\overline{D}_{\rho'}^{(i)\leq \cD}}-1\right)^{\abs{T}}\right)^{1/\abs{T}}\\
&\leq\sum_{t=1}^k\binom{m}{t}\epsilon\\
&\leq km^k\epsilon.
\end{align*}
Choosing $\epsilon=2^{-\tOmega(k\log n)}$ we can finish the proof.
\end{proof}

\noindent A priori it is unclear when~\Cref{assume:copy_likelihood} is satisfied. In the following, we show that approximate state $k$-designs provide one such example:

\begin{lemma}[Approximate state $k$-design satisfies \Cref{assume:copy_likelihood}]\label{lem:k_design_copy}
Given a hypothesis testing problem $\rhomm$ vs. $\rho\sim\mathcal{E}$ with $\mathcal{E}$ an $\epsilon$-approximate state $k$-design with $0\leq\epsilon\leq 1/2$, we have
\begin{align*}
\mathbb{E}_{\rho,\rho'\sim\mathcal{E}}\left(\expval{\overline{D}_{\rho}^{\leq \infty},\overline{D}_{\rho'}^{\leq \infty}}-1\right)^{k}\leq O(k^22^{k\log_2 k}(\epsilon+2^{-n}).
\end{align*}
\end{lemma}

\begin{proof}
We consider the expression
\begin{align}
\mathbb{E}_{\rho,\rho'\sim\mathcal{E}}\left(\expval{\overline{D}_{\rho}^{\leq \infty},\overline{D}_{\rho'}^{\leq \infty}}-1\right)^{k}&=\mathbb{E}_{\rho,\rho'\sim\mathcal{E}}\left(\expval{\overline{D}_{\rho},\overline{D}_{\rho'}}-1\right)^{k}\nonumber\\ \label{eq:momentsofinnterproduct}
&=\sum_{j=0}^k{k\choose j} (-1)^j \E_{\rho,\rho'\sim\mathcal{E}} 2^{jn}\left(\sum_{s\in\{0,1\}^n} D_{\rho}(s)D_{\rho'}(s)\right)^j.
\end{align}
Let's further look at the expression
\begin{align}\label{eq:expansion}
2^{jn}\E_{\rho,\rho'\sim\mathcal{E}} \left(\sum_{s\in\{0,1\}^n} D_{\rho}(s)D_{\rho'}(s)\right)^j=2^{jn}\sum_{s_1,\ldots,s_j} \E_{\rho} D_{\rho}(s_1)\cdots D_{\rho}(s_j) \E_{\rho'}D_{\rho'}(s_1)\cdots D_{\rho'}(s_j).
\end{align}
For a given list $s_1,\ldots,s_j$, we denote by $\# y(s_1,\ldots,s_j)$ the number of times the bit string $y$ shows up in the list .
We split the sum in~\eqref{eq:expansion} into $j$ parts summing over the number $r$ of distinct bit strings in the tuples $(s_1,\ldots,s_j)$. In the following, $B(j)$ denotes the $j$th Bell number, which counts the number of partitions of a set of size $j$ and $S(j,r)$ denotes the Stirling number of second kind, which counts the number of partitions of a set of size $j$ into $r$ non-empty subsets.
Suppose $\rho$ and $\rho'$ are drawn from the exact Haar measure, we have
\begin{align*}
&2^{jn}\sum_{s_1,\ldots,s_j} \E_{\rho} D_{\rho}(s_1)\cdots D_{\rho}(s_j) \E_{\rho'}D_{\rho'}(s_1)\cdots D_{\rho'}(s_j)\\
=&\frac{2^{jn}}{d^2(d+1)^2\cdots (d+j-1)^2}\left( \sum_{s_1,\ldots,s_j} \prod_{y\in\{0,1\}^n}(\# y(s_1,\ldots,s_j))!^2\right) \\
\leq& \frac{2^{jn}}{d^2(d+1)^2\cdots (d+j-1)^2}\left( \sum_{r=1}^j (j-r)!^2 S(j,r)d(d-1)\cdots (d-r+1)\right),
\end{align*}
where the first equality follows from \Cref{fact:haar_divide}.
The first inequality follows from organizing the sum by the number $r$ of distinct bitstrings that appear, and then by the number of ways $j$ bitstrings can have $r$ distinct bitstrings
Here, $(j-r)!^2$ is the worst case upper bound on $\prod_{y\in \{0,1\}^n} (\#y(s_1,\ldots,x_j))!^2$ when $(s_1,\ldots,s_j)$ contains $r$ distinct strings.
Moreover, $S(j,r) d(d-1)\cdots (d-r+1)$ is a loose upper bound on the number of all such lists $(s_1,\ldots,s_j)$.
Using this, we can now further expand this sum
\begin{align*}
\leq& \frac{2^{jn}}{d^2(d+1)^2\cdots (d+j-1)^2}\left(d\cdots (d-j+1) + \sum_{r=1}^{j-1} (j-r)!^2 B(j)d(d-1)\cdots (d-r+1)\right) \\
\leq& 1+B(j)\left(2j^22^{-n}+O\left(2^{-n}\sum_{l=0}^j 2^{2l\log_2(l)-ln}\right)\right)\\
=&1+2^{j\log(j)}\left(2j^22^{j\log(j)}2^{-n}+O\left(2^{-n} \frac{1}{1-2^{2\log_2(j)-n}}\right)\right)\\
=&1+O(j^22^{j\log_2 j}2^{-n})),
\end{align*}
Here we used that $S(j,j)=1$ and a simple upper bound $B(j)\leq j!$.
This is because any partition can be achieved as the cycle structure of a permutation.

However, $\rho,\rho'$ are drawn from an $\epsilon$-approximate state $k$-design here instead. Therefore, we have 
\begin{align*}
(1-2\epsilon)(1+O(j^22^{-n}))\leq \eqref{eq:expansion}\leq (1+2\epsilon)(1+O(j^22^{-n})),
\end{align*}
where we used
\begin{align*}
&1-2\epsilon\leq 1/(1+\epsilon)\\
&1/(1-\epsilon)\leq 1+2\epsilon.
\end{align*}
for $0<\epsilon\leq 1/2$. Plugging this equation into Eq.~\eqref{eq:momentsofinnterproduct} yields
\begin{align*}
\mathbb{E}_{\rho,\rho'\sim\mathcal{E}}\left(\expval{\overline{D}_{\rho}^{\leq \infty},\overline{D}_{\rho'}^{\leq \infty}}-1\right)^{k}&=\sum_{j=0}^k{k\choose j} (-1)^j \E_{\rho,\rho'\sim\mathcal{E}} 2^{jn}\left(\sum_{s\in\{0,1\}^n} D_{\rho}(s)D_{\rho'}(s)\right)^j\\
&\leq \sum_{j=0}^k {k\choose j}(-1)^j (1+(-1)^j2\epsilon)(1+O(j^22^{j\log_2 j}2^{-n}))\\
&\leq \sum_{j=0}^k {k\choose j}(-1)^j+O(k^22^{-n}+\epsilon)\sum_{j=0}^k {k\choose j}\\
&\leq O(k^22^{-n}+\epsilon)\sum_{j=0}^k {k\choose j}\\
&\leq O(2^{k}(k^22^{k\log_2 k}2^{-n}+\epsilon)),
\end{align*}
which finishes the proof for \Cref{lem:k_design_copy}.
\end{proof}

\noindent Combining \Cref{lem:k_design_copy} and \Cref{thm:general_sample}, we have the following corollary:
\begin{corollary}[Approximate state $k$-design indicates copy-wise low-degree hardness]\label{coro:general_sample}
For the hypothesis testing problem $\rhomm$ vs. $\rho\sim\mathcal{E}$ with $\mathcal{E}$ a $2^{-\tOmega(k\log n)}$-approximate state $k$-design, weak detection is degree-$(\cD,k)$ hard for any protocol that measures a polynomial number of non-adaptive single-copy PVMs $\{\{U_i\ket{x}\bra{x}U_i^{\dagger}\}_{x=1}^{2^n}\}_{i=1}^m$.
\end{corollary}

\subsection{Round-based models with adaptivity within blocks}\label{sec:general_adaptivity_within_block}

As shown in \Cref{fig:model_copy}(b), we consider a framework for learning with single-copy PVMs and restricted adaptivity called \emph{round-based models}.  
In this model, we divide the $m$ samples into $m_1$ blocks, each of $m_0$ samples. 

The first model we consider is when we have $m_1=\poly(n)$ blocks, each consisting of $m_0=\polylog(n)$ copies that can be measured adaptively. While adaptivity is allowed {\em within} the blocks, no adaptivity is allowed {\em between blocks} (see LHS of \Cref{fig:model_copy}(b) for illustration). Roughly speaking, in this model the measurements within each block are allowed to be chosen arbitrarily adaptively based on readouts from earlier in the same block, but at the start of each new block the adaptivity ``resets.'' We consider a different notion of adaptivity in the next section.
We will focus on the $(\cD,k)$ copy-wise low-degree model throughout this subsection for technical convenience. We use $\prod$ between $\overline{D}$ to denote adaptive multiplication, and $\bigotimes$ to denote non-adaptive product.

We analyze adaptivity using the learning tree framework introduced in~\cite{aharonov2022quantum,bubeck2020entanglement,chen2022exponential,chen2022complexity,chen2023efficientPauli,chen2024optimal}, Roughly speaking, in this approach
one models the learning protocol as a decision tree. Starting from the root node, we pick a PVM and measure a copy of the unknown quantum state using the picked PVM. Based on different measurement outcomes, we pick different PVMs in the different child nodes. For adaptive protocols with $m$ measurements, we repeat this procedure for $m$ times, resulting in a decision tree of depth $m$. As we choose a PVM with $2^n$ outcomes per measurement, we can map all possible histories to a $2^{mn}$ bit string. This fact is our crucial interface with the Boolean analytic techniques used in the previous section. 

Formally, we define the protocol as below
\begin{definition}[Round-based model with adaptivity within blocks]
Given $m=\poly(n)$ copies of unknown $n$-qubit state $\rho$, a round-based model with adaptivity within blocks contains $m_1=\poly(n)$ blocks each of $m_0=\polylog(n)$ copies. The blocks are independent of each other, and no adaptivity is allowed among blocks. Each block can be represented as a rooted tree $\cT_i,i\in[m_1]$ of depth $m_0$ with each node on the tree recording the measurement outcomes of the block so far. Each tree $\cT_i$ has the following properties:
\begin{itemize}
    \item We assign a probability $p^\rho(u_i)$ for any node $u_i$ on the tree $\cT_i$.
    \item The probability assigned to the root $r_i$ of the tree $\cT_i$ is $p^\rho(r_i)=1$.
    \item At each non-leaf node $u_i$, we measure $\rho$ using an adaptively chosen PVM $\{U_{u_i}\ket{s}\bra{s}U_{u_i}^\dagger\}$, which results in a classical outcome $s$. Each child node $v_i$ corresponding to the classical outcome $s$ of the node $u_i$ is connected through the edge $e_{u_i,s}$.
    \item If $v_i$ is the child of $u_i$ through the edge $e_{u_i,s}$, the probability to traverse this edge is $p^{\rho}(s|u_i)=\braket{s|U_{u_i}^\dagger\rho U_{u_i}|s}$. Then
    \begin{align*}
    p^\rho(v_i)=p^\rho(u_i)\cdot p^\rho(s|u_i)=p^\rho(u_i)\cdot\braket{s|U_{u_i}^\dagger\rho U_{u_i}|s}.
    \end{align*}
    \item Each root-to-leaf path is of length $m_0$. Note that for a leaf node $\ell_i$, $p^\rho(\ell_i)$ is the probability for the classical memory to be in state $\ell_i$ at the end of the learning protocol. Denote the set of leaves of $\cT_i$ by $\mathrm{leaf}(\cT_i)$
\end{itemize}
\end{definition}

We will need the following condition:
\begin{condition}[$(M,\epsilon,\cD,k)$-condition for restricted adaptivity within blocks]\label{assume:general_adaptivity_within_block}
Consider the hypothesis testing problem $\rhomm$ vs. $\rho\sim\mathcal{E}$. We denote the output history as $\bm{s}=(s_1,...,s_{m_0})\in\{0,1\}^{m_0n}$ with each $s_i\in\{0,1\}^n$ for $i\in[m_0]$. We assume that for any composite block of $km_0$ samples (with history $\bm{s}=(s_1,...,s_{km_0})$), we have 
\begin{align*}
\mathbb{E}_{\bm{s}_{1:km_0-1}}\mathbb{E}_{\rho\sim\mathcal{E}}\left[\left(\prod_{i=1}^{km_0-1}\overline{D}_{\rho}^{\leq \cD}(s_i)\right)^{2(km_0-1)}\right]\leq M.
\end{align*}
We further assume that for any single-copy measurement, we have
\begin{align*}
\mathbb{E}_{\rho,\rho'\sim\mathcal{E}}\left(\expval{\overline{D}_{\rho}^{\leq \cD},\overline{D}_{\rho'}^{\leq \cD}}-1\right)^{2km_0}\leq\epsilon.
\end{align*}
\end{condition}

Given the above condition, we can prove the following result:

\begin{theorem}[Low-degree hardness for round-based model with adaptivity within blocks]\label{thm:general_adaptivity_within_block}
For the hypothesis testing problem $\rhomm$ vs. $\rho\sim\mathcal{E}$ satisfying an $(M,2^{-\Omega(k^3m_0\log k\log m_1\log m_0\log M)},\cD,k)$-condition for restricted adaptivity within blocks (see \Cref{assume:general_adaptivity_within_block}), i.e., 
\begin{align*}
&\mathbb{E}_{\bm{s}_{1:km_0-1}}\mathbb{E}_{\rho\sim\mathcal{E}}\left[\left(\prod_{i=1}^{km_0-1}\overline{D}_{\rho}^{\leq \cD}(s_i)\right)^{2(km_0-1)}\right]\leq M, \quad{\text{and}}\\
&\mathbb{E}_{\rho,\rho'\sim\mathcal{E}}\left(\expval{\overline{D}_{\rho}^{\leq \cD},\overline{D}_{\rho'}^{\leq \cD}}-1\right)^{2km_0}\leq2^{-\Omega(k^3m_0\log k\log m_1\log m_0\log M)}
\end{align*}
for each block of $m_0$ samples and any single-copy measurement, weak detection is degree-$(\cD,k)$ hard for any adaptivity within round-based model protocol that makes polynomially many PVMs $\{\{U_i\ket{x}\bra{x}U_i^{\dagger}\}_{x=1}^{2^n}\}_{i=1}^m$. These protocols use $m_1=\poly(n)$ blocks each with $m_0=m/m_1=\polylog(n)$ samples, and adaptivity is only allowed in each block. 
\end{theorem}

\begin{proof}
We assume the $(M,\epsilon,\cD,k)$-assumption for restricted adaptivity within blocks in \Cref{assume:general_adaptivity_within_block} holds.We start with a single block of size $m_0$. Recall that 
\begin{align*}
\overline{D}^{(i)}_{\rho}(s_i)=2^n\tr(\rho\ket{\phi_{s_i}}\bra{\phi_{s_i}}),\ \overline{D}_{\rho,{m_0}}(\bm{s})=2^{{m_0}n}\prod_{i=1}^{m_0}\tr(\rho\ket{\phi_{s_i}}\bra{\phi_{s_i}}).
\end{align*}

Our first step is to write $\overline{D}_{\rho,{m_0}}$ as a product of functions with degree $\cD$ on each copy:
\begin{align*}
\overline{D}_{\rho,{m_0}}^{\leq \cD,\infty}(\bm{s})=\prod_{i=1}^{{m_0}}\overline{D}^{(i)\leq \cD}_{\rho'}(s_i).
\end{align*}
We have 
\begin{align*}
\norm{\mathbb{E}_{\rho\sim\mathcal{E}}\overline{D}_{\rho,{m_0}}^{\leq \cD,k}-\textbf{1}}^2&=\mathbb{E}_{\bm{s}}(\mathbb{E}_{\rho\sim\mathcal{E}}\overline{D}_{\rho,{m_0}}^{\leq \cD,k}-1)^2\\
&=\mathbb{E}_{\rho,\rho'\sim\mathcal{E}}\expval{\overline{D}_{\rho,{m_0}}^{\leq \cD,k},\overline{D}_{\rho',{m_0}}^{\leq \cD,k}}-1\\
&=\mathbb{E}_{\bm{s}}\mathbb{E}_{\rho,\rho'\sim\mathcal{E}}\overline{D}_{\rho,{m_0}}^{\leq \cD,k}(\bm{s})\overline{D}_{\rho',{m_0}}^{\leq \cD,k}(\bm{s})-1.
\end{align*}
Instead of choosing $\chi_T(\bm{s})$ in the qubit-wise, we now choose $\chi_T(\bm{s})$ in the copy-wise. We denote
\begin{align*}
\overline{D}^{\leq \cD}_{\rho,{m_0}}(\bm{s})=\sum_{T\subseteq[{m_0}]}\alpha_T^\rho\chi_T(\bm{s}),\quad \alpha_T^{\rho}=\frac{1}{2^{n{m_0}}}\sum_{\bm{s}}\overline{D}_{\rho,{m_0}}^{\leq \cD}(\bm{s})\chi_T(\bm{s})=\mathbb{E}_{\bm{s}}\overline{D}_{\rho,{m_0}}^{\leq\cD}(\bm{s})\chi_T(\bm{s}).
\end{align*}
Again, we focus on the Fourier basis in this work. The low-degree likelihood ratio satisfies:
\begin{align*}
\overline{D}_{\rho,{m_0}}^{\leq\cD, k}(\bm{s})=\sum_{\abs{T}\leq k}\alpha_T^\rho\chi_T(\bm{s}).
\end{align*}
We compute the square degree-$(\cD,k)$ advantage as
\begin{align*}
\norm{\mathbb{E}_{\rho\sim\mathcal{E}}\overline{D}_{\rho,{m_0}}^{\leq \cD,k}-\textbf{1}}^2&=\mathbb{E}_{\bm{s}}\mathbb{E}_{\rho,\rho'\sim\mathcal{E}}\overline{D}_{\rho,{m_0}}^{\leq \cD,k}(\bm{s})\overline{D}_{\rho',{m_0}}^{\leq \cD,k}(\bm{s})-1\\
&=\mathbb{E}_{\bm{s}}\mathbb{E}_{\rho,\rho'\sim\mathcal{E}}\sum_{\abs{T},\abs{T'}\leq k}\alpha_T^\rho\alpha_{T'}^{\rho'}\chi_T(\bm{s})\chi_{T'}(\bm{s})-1\\
&=\mathbb{E}_{\rho,\rho'\sim\mathcal{E}}\sum_{\abs{T}\leq k}\alpha_T^\rho\alpha_{T}^{\rho'}-1\\
&=\mathbb{E}_{\rho,\rho'\sim\mathcal{E}}\mathbb{E}_{\bm{s},\bm{s}'}\sum_{\abs{T}\leq k}\overline{D}^{\leq \cD}_{\rho,{m_0}}(\bm{s})\overline{D}^{\leq \cD}_{\rho',{m_0}}(\bm{s}')\chi_T(\bm{s})\chi_T(\bm{s}')-1\\
&=\sum_{\abs{T}\leq k,T\neq\emptyset}\mathbb{E}_{\rho,\rho'\sim\mathcal{E}}\mathbb{E}_{\bm{s},\bm{s}'}\overline{D}^{\leq \cD}_{\rho,{m_0}}(\bm{s})\overline{D}^{\leq \cD}_{\rho',{m_0}}(\bm{s}')\chi_T(\bm{s})\chi_T(\bm{s}'),
\end{align*}
where the fourth line follows from the fact that $\expval{\chi_T(\bm{s}),\chi_{T'}(\bm{s})}=\delta_{T,T'}$. However, we can not use the previous H\"{o}lder inequality directly due to adaptivity. To this end, we have to use the recursive trick. We denote
\begin{align*}
A_{m_0,k}=\norm{\mathbb{E}_{\rho\sim\mathcal{E}}\overline{D}_{\rho,{m_0}}^{\leq \cD,k}-\textbf{1}}^2
\end{align*}
for convenience. We thus have:
\begin{align}\label{eq:induction_within}
A_{m_0,k}=A_{m_0-1,k}+\sum_{\abs{T'}\leq k-1,T'\neq\emptyset}\mathbb{E}_{\rho,\rho'\sim\mathcal{E}}\mathbb{E}_{\bm{s},\bm{s}'}\overline{D}^{\leq \cD}_{\rho,{m_0}}(\bm{s})\overline{D}^{\leq \cD}_{\rho',{m_0}}(\bm{s}')\chi_{T'\cup[m_0]}(\bm{s})\chi_{T'\cup[m_0]}(\bm{s}').
\end{align}

We now consider the second term. We have (we always assume $T'\neq\emptyset$ in the following)
\begin{align}\label{eq:induction_second_within}
\begin{split}
&\sum_{\abs{T'}\leq k-1}\mathbb{E}_{\rho,\rho'\sim\mathcal{E}}\mathbb{E}_{\bm{s},\bm{s}'}\overline{D}^{\leq\cD}_{\rho,{m_0}}(\bm{s})\overline{D}^{\leq\cD}_{\rho',{m_0}}(\bm{s}')\chi_{T'\cup[m_0]}(\bm{s})\chi_{T'\cup[m_0]}(\bm{s}')\\
=&\sum_{\abs{T'}\leq k-1}\mathbb{E}_{\rho,\rho'\sim\mathcal{E}}\mathbb{E}_{\bm{s}_{1:m_0-1},\bm{s}'_{1:m_0-1}}\overline{D}^{\leq\cD}_{\rho,{m_0-1}}(\bm{s}_{1:m_0-1})\overline{D}^{\leq\cD}_{\rho',{m_0-1}}(\bm{s}'_{1:m_0-1})\chi_{T'}(\bm{s}_{1:m_0-1})\chi_{T'}(\bm{s}'_{1:m_0-1})\cdot\\
&\quad \mathbb{E}_{s_{m_0},s'_{m_0}}[\overline{D}_\rho^{(m_0)\leq k}(s_{m_0})\overline{D}_{\rho'}^{(m_0)\leq k}(s'_{m_0})\chi_{[s_{m_0}]}(s_{m_0})\chi_{[s'_{m_0}]}(s'_{m_0})]\\
\leq&\left(\underset{\rho,\rho'\sim\mathcal{E}}{\mathbb{E}}\underset{\bm{s}_{1:m_0-1},\bm{s}'_{1:m_0-1}}{\mathbb{E}}\left(\sum_{\abs{T'}\leq k-1}\overline{D}^{\leq\cD}_{\rho,{m_0-1}}(\bm{s}_{1:m_0-1})\overline{D}^{\leq\cD}_{\rho',{m_0-1}}(\bm{s}'_{1:m_0-1})\chi_{T'}(\bm{s}_{1:m_0-1})\chi_{T'}(\bm{s}'_{1:m_0-1})\right)^2\right)^{\frac{1}{2}}\cdot\epsilon^{\frac{1}{2km_0}}\\
\leq&km_0^k\epsilon^{\frac{1}{2km_0}}\mathbb{E}_{\rho\sim\mathcal{E}}\mathbb{E}_{\bm{s}_{1:m_0-1}}\left[\left(\prod_{i=1}^{m-1}\overline{D}_{\rho}^{(i)\leq \cD}(s_i)\right)^2\right]\\
\leq&km_0^k\epsilon^{\frac{1}{2km_0}}\mathbb{E}_{\bm{s}_{1:m_0-1}}\mathbb{E}_{\rho\sim\mathcal{E}}\left[\left(\overline{D}_{\rho}^{(i)\leq \cD}(s_i)\right)^{2(m_0-1)}\right],
\end{split}
\end{align}
where the second step follows from implementing H\"{o}lder's inequality twice and $(M,\epsilon,\cD,k)$-assumption for restricted adaptivity within blocks in \Cref{assume:general_adaptivity_within_block}, and the fact that
\begin{align*}
\left(\sum_{\abs{T'}\leq k-1}\chi_{T'}(\bm{s}_{1:m_0-1})\chi_{T'}(\bm{s}'_{1:m_0-1})\right)^2\leq\left(\sum_{t=1}^{k-1}\binom{m_0}{k}\right)^2\leq k^2m_0^k
\end{align*}
and the third step follows from the Cauchy-Schwarz inequality. This indicates that 
\begin{align}\label{eq:induction_block_within}
\norm{\mathbb{E}_{\rho\sim\mathcal{E}}\overline{D}_{\rho,{m_0}}^{\leq \cD,k}-\textbf{1}}^2=A_{m_0,k}=km_0^{k+1}\epsilon^{\frac{1}{2km_0}}\mathbb{E}_{\bm{s}_{1:m_0-1}}\mathbb{E}_{\rho\sim\mathcal{E}}\left[\left(\overline{D}_{\rho}^{(i)\leq \cD}(s_i)\right)^{2(m_0-1)}\right].
\end{align}
We then consider the full likelihood ratio $\overline{D}_{\rho,m}(\bm{S})$ for the complete bit string output $\bm{S}=(\bm{s}_1,...,\bm{s}_{m_1})\in\{0,1\}^{m_0m_1}$ for $m=m_0m_1$ copies using the non-adaptive result:
\begin{align*}
\norm{\mathbb{E}_{\rho\sim\mathcal{E}}\overline{D}_{\rho,{m}}^{\leq \cD,k}-\textbf{1}}^2&=\mathbb{E}_{\bm{S}}(\mathbb{E}_{\rho\sim\mathcal{E}}\overline{D}_{\rho,{m}}^{\leq \cD,k}-1)^2=\mathbb{E}_{\rho,\rho'\sim\mathcal{E}}\expval{\overline{D}_{\rho,{m}}^{\leq \cD,k},\overline{D}_{\rho',{m}}^{\leq \cD,k}}-1\\
&=\mathbb{E}_{\bm{S}}\mathbb{E}_{\rho,\rho'\sim\mathcal{E}}\overline{D}_{\rho,{m}}^{\leq \cD,k}(\bm{S})\overline{D}_{\rho',{m}}^{\leq \cD,k}(\bm{S})-1\\
&\leq\sum_{t=1}^{k}\binom{m}{t}\mathbb{E}_{\rho,\rho'\sim\mathcal{E}}\left(\expval{\overline{D}_{\rho,{m_0}}^{\leq\cD, k},\overline{D}_{\rho',{m_0}}^{\leq\cD, k}}-1\right)^t\\
&\leq km_1^k\mathbb{E}_{\rho,\rho'\sim\mathcal{E}}\left(\expval{\overline{D}_{\rho,{m_0}}^{\leq\cD, k},\overline{D}_{\rho',{m_0}}^{\leq\cD, k}}-1\right)^k\\
&\leq km_1^k\sum_{j=0}^k\binom{k}{j}\mathbb{E}_{\rho,\rho'\sim\mathcal{E}}\left(\mathbb{E}_{\bm{s}\in\{0,1\}^{m_0n}}\overline{D}_{\rho,{m_0}}^{\leq\cD, k}(\bm{s})\overline{D}_{\rho',{m_0}}^{\leq\cD, k}(\bm{s})\right)^j.
\end{align*}
We note that
\begin{align*}
\E_{\rho,\rho'\sim\mathcal{E}}\left(\mathbb{E}_{\bm{s}\in\{0,1\}^{m_0n}}\overline{D}_{\rho,{m_0}}^{\leq\cD, k}(\bm{s})\overline{D}_{\rho',{m_0}}^{\leq\cD, k}(\bm{s})\right)^j&=\E_{\rho,\rho'\sim\mathcal{E}}\E_{\bm{s}_1,...,\bm{s}_j\in\{0,1\}^{m_0n}}\overline{D}_{\rho,{m_0}}^{\leq\cD, k}(\bm{s}_1)...\overline{D}_{\rho,{m_0}}^{\leq\cD, k}(\bm{s}_j)\overline{D}_{\rho',{m_0}}^{\leq\cD, k}(\bm{s}_1)...\overline{D}_{\rho',{m_0}}^{\leq\cD, k}(\bm{s}_j)\\
&\leq \E_{\rho,\rho'\sim\mathcal{E}}\E_{\bm{s}\in\{0,1\}^{n}}\overline{D}_{\rho}^{\leq\cD, jk}(\bm{s})\overline{D}_{\rho',{m_0}}^{\leq\cD, jk}(\bm{s})\\
&\leq 1+A_{km_0,jk}
\end{align*}
as choosing $\bm{s_1},...,\bm{s}_j$ is a special case for adaptively choosing all $jm_0$ copies and there can be at most copy-wise degree of $(\cD,jk)$. We thus have
\begin{align*}
\norm{\mathbb{E}_{\rho\sim\mathcal{E}}\overline{D}_{\rho,{m}}^{\leq\cD, k}-\textbf{1}}^2&\leq km_1^k\sum_{j=0}^k\binom{k}{j}\mathbb{E}_{\rho,\rho'\sim\mathcal{E}}\left(\mathbb{E}_{\bm{s}\in\{0,1\}^{m_0n}}\overline{D}_{\rho,{m_0}}^{\leq \cD,k}(\bm{s})\overline{D}_{\rho',{m_0}}^{\leq \cD,k}(\bm{s})\right)^j\\
&\leq km_1^k2^kA_{km_0,k^2}\\
&\leq k2^km_1^k\cdot k^2(km_0)^{k^2}\epsilon^{\frac{1}{2km_0}}\mathbb{E}_{\bm{s}_{1:km_0-1}}\mathbb{E}_{\rho\sim\mathcal{E}}\left[\left(\overline{D}_{\rho}^{(i)\leq\cD}(s_i)\right)^{2(km_0-1)}\right].
\end{align*}
Recall that $(M,\epsilon,\cD,k)$-assumption for restricted adaptivity within blocks in \Cref{assume:general_adaptivity_within_block} indicates that 
\begin{align*}
\mathbb{E}_{\bm{s}_{1:km_0-1}}\mathbb{E}_{\rho\sim\mathcal{E}}\left[\left(\prod_{i=1}^{km_0-1}\overline{D}_{\rho}^{\leq \cD}(s_i)\right)^{2(km_0-1)}\right]\leq M.
\end{align*}
We can thus compute the overall square degree-$(\cD,k)$ advantage by
\begin{align*}
\norm{\mathbb{E}_{\rho\sim\mathcal{E}}\overline{D}_{\rho,{m}}^{\leq\cD, k}-\textbf{1}}^2&\leq k2^km_1^k\cdot k^2(km_0)^{k^2}\epsilon^{\frac{1}{2km_0}}\mathbb{E}_{\bm{s}_{1:km_0-1}}\mathbb{E}_{\rho\sim\mathcal{E}}\left[\left(\overline{D}_{\rho}^{(i)\leq\cD}(s_i)\right)^{2(km_0-1)}\right]\\
&\leq k2^km_1^k\cdot k^2(km_0)^{k^2}\epsilon^{\frac{1}{2km_0}}M\\
&=2^{k^2\log k\log m_0\log m_1}\epsilon^{\frac{1}{2km_0}}M
\end{align*}
By choosing
\begin{align*}
\epsilon=2^{-\Omega(k^3m_0\log k\log m_1\log m_0\log M)},
\end{align*}
we can ensure the square degree-$(\cD,k)$ advantage below $o(1)$, which finishes the proof for \Cref{thm:general_adaptivity_within_block}.
\end{proof}
 
\noindent At a first glance, the $(M,\epsilon,\cD,k)$-assumption for restricted adaptivity within blocks in \Cref{assume:general_adaptivity_within_block} in \Cref{thm:general_adaptivity_within_block} seems complicated. In the following, we show that approximate state designs satisfy such an assumption. Formally, we have the following corollary (see \Cref{sec:pf_general_adaptivity_within_block} for the proof):
\begin{corollary}[Approximate state $2km_0$-design indicates hardness for round-based model with adaptivity within blocks]\label{coro:general_adaptivity_within_block}
For the hypothesis testing problem $\rhomm$ vs. $\rho\sim\mathcal{E}$ with $\mathcal{E}$ a $2^{-\Omega(k^4m_0^2\log^2 k\log m_1\log^2 m_0)}$-approximate state $2km_0$-design, weak detection is degree-$(\cD,k)$ hard for any adaptivity within round-based model protocol that makes polynomially many PVMs $\{\{U_i\ket{x}\bra{x}U_i^{\dagger}\}_{x=1}^{2^n}\}_{i=1}^m$. These protocols use $m_1=\poly(n)$ blocks each with $m_0=m/m_1=\polylog(n)$ samples, and adaptivity is only allowed in each block. 
\end{corollary}

\subsection{Round-based models with adaptivity among blocks}\label{sec:general_adaptivity_among_block}

The second model we consider is when we have $m_1=\text{polylog}(n)$ adaptive blocks of $m_0=\text{poly}(n)$ non-adaptive copies within each block. Roughly speaking, at the start of each block, we can prepare a fixed measurement strategy to apply to the next $\mathrm{poly}(n)$ copies, where the strategy is chosen fully adaptively based on all readouts from the previous blocks. We formalize this as follows:

\begin{definition}[Round-based model with adaptivity among blocks]
Given $m=\poly(n)$ copies of unknown $n$-qubit state $\rho$, a round-based model with adaptivity within blocks contains $m_1=\polylog(n)$ blocks each of $m_0=\poly(n)$ copies. While no adaptivity is allowed within each block, how to measure a whole new block can be adaptively decided by previous history. The whole process can be as a rooted tree $\cT$ of depth $m_1$ with each node on the tree recording the measurement outcomes of a block. It has the following properties:
\begin{itemize}
    \item We assign a probability $p^\rho(u)$ for any node $u$ on the tree $\cT$.
    \item The probability assigned to the root $r$ of the tree $\cT$ is $p^\rho(r)=1$.
    \item At each non-leaf node $u$, we measure $m_1$ copies of $\rho$ using an $m_0$ independent PVMs $\{U_{u,i}\ket{s}\bra{s}U_{u,i}^\dagger,i\in[m_0]\}$, which results in a classical outcome $\bm{s}_u=s_u^{(1)},...,s_u^{(m_0)}$. Each child node $v$ corresponding to the classical outcome $\bm{s}_u$ of the node $u$ is connected through the edge $e_{u,s}$.
    \item If $v$ is the child of $u$ through the edge $e_{u,s}$, the probability to traverse this edge is $p^{\rho}(s|u)=\prod_{i=1}^{m_0}\braket{s|U_{u,i}^\dagger\rho U_{u,i}|s}$. Then
    \begin{align*}
    p^\rho(v)=p^\rho(u)\cdot p^\rho(s|u)=p^\rho(u)\cdot\prod_{i=1}^{m_0}\braket{s|U_{u,i}^\dagger\rho U_{u,i}|s}.
    \end{align*}
    \item Each root-to-leaf path is of length $m_1$. Note that for a leaf node $\ell$, $p^\rho(\ell)$ is the probability for the classical memory to be in state $\ell$ at the end of the learning protocol. Denote the set of leaves of $\cT$ by $\mathrm{leaf}(\cT)$
\end{itemize}
\end{definition}

We will need the following condition:

\begin{condition}[$(M',\epsilon,\cD,k)$-condition for restricted adaptivity among blocks]\label{assume:general_adaptivity_among_block}
Consider the hypothesis testing problem $\rhomm$ vs. $\rho\sim\mathcal{E}$. We assume that for any single-copy measurement with output $s$, we have 
\begin{align*}
\mathbb{E}_{\rho\sim\mathcal{E}}\left(\overline{D}_{\rho}^{\leq \cD}(s)-1\right)^{2(m_1-1)}\leq M'.
\end{align*}
We further assume that for any single-copy measurement, we have
\begin{align*}
\mathbb{E}_{\rho,\rho'\sim\mathcal{E}}\left(\expval{\overline{D}_{\rho}^{\leq \cD},\overline{D}_{\rho'}^{\leq \cD}}-1\right)^{2k}\leq\epsilon.
\end{align*}
\end{condition}

Given the above condition, we can prove the following result:

\begin{theorem}[Low-degree hardness for round-based model with adaptivity among blocks]\label{thm:general_adaptivity_among_block}
For the hypothesis testing problem $\rhomm$ vs. $\rho\sim\mathcal{E}$ satisfying a $(M',2^{-\Omega(km_1\log k\log m_0\log M')},\cD,k)$-condition for restricted adaptivity among blocks (see \Cref{assume:general_adaptivity_among_block}), i.e., 
\begin{align*}
&\mathbb{E}_{\rho\sim\mathcal{E}}\left(\overline{D}_{\rho}^{\leq \cD}(s)-1\right)^{2(m_1-1)}\leq M'\\
&\mathbb{E}_{\rho,\rho'\sim\mathcal{E}}\left(\expval{\overline{D}_{\rho}^{\leq \cD},\overline{D}_{\rho'}^{\leq \cD}}-1\right)^{2k}\leq2^{-\Omega(km_1\log k\log m_0\log m_1\log M')}
\end{align*}
for any single-copy PVM, weak detection is degree-$(\cD,k)$ hard for any adaptivity among round-based model protocol that makes polynomially many single-copy PVMs $\{\{U_i\ket{x}\bra{x}U_i^{\dagger}\}_{x=1}^{2^n}\}_{i=1}^m$. These protocols use $m_1=\polylog(n)$ blocks each with $m_0=m/m_1=\poly(n)$ samples, and adaptivity is only allowed among blocks. 
\end{theorem}

\noindent Intuitively, the proof of \Cref{thm:general_adaptivity_among_block} also uses the recursive induction trick. However, instead of inducting on each single copy, we should induct on each block instead.

\begin{proof}
The first step is to project the overall likelihood-ratio $\overline{D}^{\leq\cD}_{\rho,m}(\bm{S})$ for $\bm{S}=(\bm{s}_1,...,\bm{s}_{m_1})\in\{0,1\}^{mn}$ into product of low-degree blocks as $\overline{D}_{\rho,m}^{\leq \cD}(\bm{S})=\prod_{i=1}^{m_1}\overline{D}_{\rho,m_0}^{\leq \cD}(\bm{s}_i)$, where each $\bm{s}_i\in\{0,1\}^{m_0n}$ is the history in block $i=1,...,m_1$. Note that in the copy-wise low-degree model, we always have $\overline{D}_{\rho,m}^{\leq \cD,k}(\bm{S})$ included in $\prod_{i=1}^{m_1}\overline{D}_{\rho,m_0}^{\leq \cD,k}(\bm{s}_i)$. In the following, we abuse the notation a little and instead denote
\begin{align*}
\overline{D}_{\rho,m}^{\leq \cD}=\prod_{i=1}^{m_1}\overline{D}_{\rho,m_0}^{\leq \cD,k}(\bm{s}_i).
\end{align*}
We use the \emph{Fourier basis for each block} $\chi_T(\bm{S})$ for $T\subseteq[m_1]$ here and decompose $\overline{D}_{\rho,m}^{\leq \cD}(\bm{S})$ as
\begin{align*}
\overline{D}_{\rho,{m}}^{\leq \cD}(\bm{S})=\sum_{\abs{T}\subseteq[m_1]}\alpha_T^\rho\chi_T(\bm{S}),\quad \alpha_T^\rho=\frac{1}{2^{nm}}\sum_{\bm{S}}\overline{D}_{\rho,m}^{\leq \cD}(\bm{S})\chi_T(\bm{S})=\mathbb{E}_{\bm{S}}\overline{D}_{\rho,m}^{\leq\cD}(\bm{S})\chi_T(\bm{S}).
\end{align*}
By the definition of copy-wise square degree-$(\cD,k)$ advantage, we have
\begin{align*}
\norm{\mathbb{E}_{\rho\sim\mathcal{E}}\overline{D}_{\rho,m}^{\leq \cD,k}-\textbf{1}}^2&\leq\sum_{\abs{T}\leq k,T\neq\emptyset}\mathbb{E}_{\rho,\rho'\sim\mathcal{E}}\mathbb{E}_{\bm{S},\bm{S}'}\overline{D}^{\leq\cD}_{\rho,m}(\bm{S})\overline{D}^{\leq \cD}_{\rho',m}(\bm{S}')\chi_T(\bm{S})\chi_T(\bm{S}').
\end{align*}
We denote the above quantity as $B_{m_1,k}$. We have the following induction by dividing $T$ into those with the last block or without:
\begin{align*}
B_{m_1,k}=B_{m_1-1,k}+\sum_{\abs{T'}\leq k-1,T'\neq\emptyset}\mathbb{E}_{\rho,\rho'\sim\mathcal{E}}\mathbb{E}_{\bm{S},\bm{S}'}\overline{D}^{\leq \cD}_{\rho,m}(\bm{S})\overline{D}^{\leq \cD}_{\rho',m}(\bm{S}')\chi_{T'\cup[m_1]}(\bm{S})\chi_{T'\cup[m_1]}(\bm{S}').
\end{align*}

We now consider the second term. We have (we always assume $T'\neq\emptyset$ in the following)
\begin{align*}
&\sum_{\abs{T'}\leq k-1,T'\neq\emptyset}\mathbb{E}_{\rho,\rho'\sim\mathcal{E}}\mathbb{E}_{\bm{S},\bm{S}'}\overline{D}^{\leq\cD}_{\rho,m}(\bm{S})\overline{D}^{\leq\cD}_{\rho',m}(\bm{S}')\chi_{T'\cup[m_1]}(\bm{S})\chi_{T'\cup[m_1]}(\bm{S}')\\
=&\sum_{\abs{T'}\leq k-1}\mathbb{E}_{\rho,\rho'\sim\mathcal{E}}\mathbb{E}_{\bm{S}_{1:m-m_0},\bm{S}'_{1:m-m_0}}\overline{D}^{\leq\cD}_{\rho,{m-m_0}}(\bm{S}_{1:m_0-1})\overline{D}^{\leq\cD}_{\rho',{m-m_0}}(\bm{S}'_{1:m-m_0})\chi_{T'}(\bm{S}_{1:m-m_0})\chi_{T'}(\bm{S}'_{1:m-m_0})\cdot\\
&\quad \mathbb{E}_{\bm{s}_{m_0},\bm{s}'_{m_0}}[\overline{D}_{\rho,m_0}^{\leq\cD,k}(\bm{s}_{m_0})\overline{D}_{\rho',m_0}^{\leq\cD,k}(\bm{s}'_{m_0})\chi_{[\bm{s}_{m_0}]}(\bm{s}_{m_0})\chi_{[\bm{s}'_{m_0}]}(\bm{s}'_{m_0})]\\
\leq&\left(\underset{\bm{S}_{1:m-m_0}\bm{S}'_{1:m-m_0}\atop\rho,\rho'\sim\mathcal{E}}{\mathbb{E}}\left(\sum_{\abs{T'}\leq k-1}\overline{D}^{\leq\cD}_{\rho,{m-m_0}}(\bm{S}_{1:m-m_0})\overline{D}^{\leq\cD}_{\rho',{m-m_0}}(\bm{S}'_{1:m-m_0})\chi_{T'}(\bm{S}_{1:m-m_0})\chi_{T'}(\bm{S}'_{1:m-m_0})\right)^2\right)^{\frac{1}{2}}\\
&\left(\underset{\bm{S}_{1:m-m_0}\bm{S}'_{1:m-m_0}\atop\rho,\rho'\sim\mathcal{E}}{\mathbb{E}}\left(\expval{\overline{D}_{\rho,m_0}^{\leq\cD,k},\overline{D}_{\rho',m_0}^{\leq\cD,k}}-1\right)^2\right)^{1/2}\\
\leq&km_1^k\mathbb{E}_{\bm{S}_{1:m-m_0},\rho\sim\mathcal{E}}\left[\overline{D}^{\leq\cD}_{\rho,{m-m_0}}(\bm{S}_{1:m-m_0})^2\right]\cdot\left(\mathbb{E}_{\rho,\rho'\sim\mathcal{E}}\left(\sum_{B\neq\emptyset,\abs{B}\leq k\atop B\subseteq[m_0]}\prod_{j\in B}\left(\expval{\overline{D}_{\rho}^{(j)\leq \cD},\overline{D}_{\rho'}^{(j)\leq \cD}}-1\right)\right)^2\right)^{1/2}\\
\leq &k^2m^k\mathbb{E}_{\bm{S}_{1:m-m_0},\rho\sim\mathcal{E}}\left[\overline{D}^{\leq\cD}_{\rho,{m-m_0}}(\bm{S}_{1:m-m_0})^2\right]\cdot\left(\mathbb{E}_{\rho,\rho'\sim\mathcal{E}}\left(\prod_{B\neq\emptyset,\abs{B}\leq k\atop B\subseteq[m_0],j\in B}\left(\expval{\overline{D}_{\rho}^{(j)\leq\cD},\overline{D}_{\rho'}^{(j)\leq\cD}}-1\right)\right)^2\right)^{1/2}\\
\leq &k^2m^k\mathbb{E}_{\bm{S}_{1:m-m_0},\rho\sim\mathcal{E}}\left[\overline{D}^{\leq\cD}_{\rho,{m-m_0}}(\bm{S}_{1:m-m_0})^2\right]\cdot\left(\mathbb{E}_{\rho,\rho'\sim\mathcal{E}}\left(\expval{\overline{D}_{\rho}^{\leq \cD},\overline{D}_{\rho'}^{\leq \cD}}-1\right)^{2k}\right)^{1/2},
\end{align*}
where the second step follows from Cauchy-Swartz inequality, the third step uses the fact that 
\begin{align*}
\sum_{\abs{T'}\leq k-1}\chi_{T'}(\bm{S}_{1:m-m_0})\chi_{T'}(\bm{S}'_{1:m-m_0})\leq\sum_{t=1}^{k-1}\binom{m_1-1}{t}\leq km_1^k
\end{align*}
and the non-adaptive result for each block, and the fifth step follows from H\"{o}lder's inequality. By the $(M',\epsilon,\cD,k)$-assumption for restricted adaptivity among blocks, we further have
\begin{align*}
&\sum_{\abs{T'}\leq k-1,T'\neq\emptyset}\mathbb{E}_{\rho,\rho'\sim\mathcal{E}}\mathbb{E}_{\bm{S},\bm{S}'}\overline{D}^{\leq\cD}_{\rho,m}(\bm{S})\overline{D}^{\leq\cD}_{\rho',m}(\bm{S}')\chi_{T'\cup[m_1]}(\bm{S})\chi_{T'\cup[m_1]}(\bm{S}')\\
&\leq k^2m^k\epsilon^{1/2}\mathbb{E}_{\bm{S}_{1:m-m_0},\rho\sim\mathcal{E}}\left[\overline{D}^{\leq\cD}_{\rho,{m-m_0}}(\bm{S}_{1:m-m_0})^2\right].
\end{align*}

Now, we consider the term $\mathbb{E}_{\bm{S}_{1:m-m_0},\rho\sim\mathcal{E}}\left[\overline{D}^{\leq cD}_{\rho,{m-m_0}}(\bm{S}_{1:m-m_0})^2\right]$. It can be computed by H\"{o}lder's inequality as:
\begin{align*}
\mathbb{E}_{\bm{S}_{1:m-m_0},\rho\sim\mathcal{E}}\left[\overline{D}^{\leq \cD}_{\rho,{m-m_0}}(\bm{S}_{1:m-m_0})^2\right]&=\mathbb{E}_{\bm{S}_{1:m-m_0}}\mathbb{E}_{\rho\sim\mathcal{E}}\left[\prod_{i=1}^{m_1-1}\overline{D}_{\rho,m_0}^{\leq\cD,k}(\bm{s}_i)^2\right]\\
&\leq \mathbb{E}_{\bm{S}_{1:m-m_0}}\prod_{i=1}^{m_1-1}\left(\mathbb{E}_{\rho\sim\mathcal{E}}\left[D_{\rho,m_0}^{\leq\cD}(\bm{s}_i)^{2(m_1-1)}\right]\right)^{1/(m_1-1)}
\end{align*}
Here, we consider the maximal value of $\mathbb{E}_{\rho\sim\mathcal{E}}\left[D_{\rho,m_0}^{\leq \cD,k}(\bm{s}_i)^{2(m_1-1)}\right]$. Given a term $\overline{D}_{\rho,m_0}^{\leq \cD,k}(\bm{s}_i)^l$, fix $B_i^i,...,B_i^l\subseteq[m_0]$ with size $\leq k$, we can denote $C_0,...,C_l$ as the maximal subset of $[m_0]$ that is covered by exactly $0,...,l$ sets within $B_i^i,...,B_i^l$. We thus have
\begin{align*}
&\mathbb{E}_{\rho\sim\mathcal{E}}\left[D_{\rho,m_0}^{\leq \cD,k}(\bm{s}_i)^{2(m_1-1)}\right]\\
\leq&\mathbb{E}_{\rho\sim\mathcal{E}}\left[\sum_{B_i^1,...,B_i^{m_1-1}}1^{\otimes C_0}\otimes\bigotimes_{j\in C_1}(\overline{D}_{\rho}^{\leq \cD}(s_{i,j})-1)\otimes...\otimes\bigotimes_{j\in C_{2(m_1-1)}}(\overline{D}_\rho^{\leq \cD}(s_{i,j})-1)^{2(m_1-1)}\right]\\
\leq& k^{2(m_1-1)}m_0^{2k(m_1-1)}\mathbb{E}_{\rho\sim\mathcal{E}}\left[\bigotimes_{j\in \cup_{u=0}^{2(m_1-1)}C_u}(\overline{D}_\rho^{\leq \cD}(s_{i,j})-1)^{2(m_1-1)}\right],\\
\end{align*}
where $s_{i,j}\in\{0,1\}^n$ is the output bit string of the $j$-th copy in the $i$-th block. By the $(M',\epsilon,\cD,k)$-assumption for restricted adaptivity among blocks, we further have
\begin{align*}
\mathbb{E}_{\rho\sim\mathcal{E}}\left[D_{\rho,m_0}^{\leq \cD,k}(\bm{s}_i)^{2(m_1-1)}\right]\leq M'^{km_1}k^{2(m_1-1)}m_0^{2k(m_1-1)}.
\end{align*}
We thus have
\begin{align*}
B_{m_1,k}&=\norm{\mathbb{E}_{\rho\sim\mathcal{E}}\overline{D}_{\rho,m}^{\leq\cD, k}-1}^2\\
&\leq m_1k^2m^kM'^{km_1}k^{2(m_1-1)}m_0^{2k(m_1-1)}\epsilon^{1/2}\\
&=2^{O(km_1\log M'\log m_0\log m_1\log k)}\epsilon^{1/2}.
\end{align*}
By choosing
\begin{align*}
\epsilon=2^{-\Omega(km_1\log k\log m_1\log m_0\log M')},
\end{align*}
we can ensure the square degree-$(\cD,k)$ advantage below $o(1)$, which finishes the proof for \Cref{thm:general_adaptivity_among_block}.
\end{proof}

Again, we show that approximate state designs satisfy $(M',\epsilon,\cD,k)$-assumption for restricted adaptivity among blocks in \Cref{assume:general_adaptivity_among_block}. Formally, we have the following corollary (see \Cref{sec:pf_general_adaptivity_among_block} for the proof):
\begin{corollary}[Approximate state $2k$-design indicates hardness for round-based model among adaptivity within blocks]\label{coro:general_adaptivity_among_block}
For the hypothesis testing problem $\rhomm$ vs. $\rho\sim\mathcal{E}$ with $\mathcal{E}$ a $2^{-\Omega(km_1^2\log^2 m_1\log m_0\log k)}$-approximate state $2k$-design, weak detection is degree-$(\cD,k)$ hard for any adaptivity among round-based model protocol that makes polynomially many single-copy PVMs $\{\{U_i\ket{x}\bra{x}U_i^{\dagger}\}_{x=1}^{2^n}\}_{i=1}^m$. These protocols use $m_1=\polylog(n)$ blocks each with $m_0=m/m_1=\poly(n)$ samples, and adaptivity is only allowed among blocks. 
\end{corollary}

\subsection{The statistical query model and its equivalence with the copy-wise low-degree model}\label{sec:general_equivalence}

In this section, we make the simple observation that the copy-wise low-degree model that we consider is equivalent to the \emph{statistical query} model applied to readouts from non-adaptive measurements, provided the readouts in the low-degree model are noisy. The conclusion follows in a straightforward fashion from an analogous result previously proved in the classical setting~\cite{brennan2020statistical}, and the main contribution of this section is conceptual: to identify a physically motivated scenario under which the conditions for their result hold in the quantum setting.

\subsubsection{Review of statistical query model}\label{sec:basic_sq}

Here, we introduce another computational model of \emph{statistical query (SQ)} model. The SQ model was first proposed to design and analyze noise-tolerant algorithms~\cite{kearns1998efficient}. It is further developed to be a popular restricted model of computation for studying information-computation tradeoffs~\cite{diakonikolas2017statistical,feldman2017statistical,feldman2015complexity}. In this work, we focus on the VSTAT$(m)$ algorithm, which access a distribution $p$ over $\R^n$ via queries $\phi:\R^n\to[0,1]$ to an oracle. For an query $\phi$, it returns to a value $\E_{x\sim p}[\phi(x)]+\xi$ with an adversarial chosen $\xi\in\R$ with $\abs{\xi}\leq\max\left\{\frac 1m,\sqrt{\frac{\E[\phi](1-\E[\phi])}{m}}\right\}$. The approximation of $\E[\phi]$ is the same as an $m$-sample empirical estimate under the guarantees of concentration inequality. An algorithm that makes $q$ queries to VSTAT$(m)$ is roughly a proxy for an algorithm running in time $q$ on $m$ samples despite that $\phi$ need not to be polynomial-time computable but have must be a function of only a single sample.

To prove lower bounds against SQ algorithms, we introduce a complexity measure on hypothesis testing problems called \emph{statistical dimension}. We will use the following definition of statistical dimension from~\cite{brennan2020statistical}: 
\begin{definition}[Definition 1.2 of~\cite{brennan2020statistical}]\label{def:statistical_query}
Given the hypothesis testing problem $D_\emptyset$ vs. $\cS=\{D_u\}_{u\sim\mu}$, the statistical dimension $\text{SDA}(\cS,\mu,m)$ is defined to be:
\begin{align}\label{eq:def_SDA}
\text{SDA}(\cS,\mu,m)\coloneqq\max\left\{N\in\N:\E_{u,v\in\mu}\left[\abs{\expval{D_u,D_v}-1}|A\right]\leq\frac1m,\ \forall\text{ events }A\text{ s.t. }\Pr_{u,v\sim\mu}(A)\geq \frac{1}{N^2}\right\}.
\end{align}
\end{definition}
\noindent The intuition for the above definition might be obscure at first glance. Recall that $\expval{D_u,D_v}-1$ is the centered average of the likelihood ratio $\overline{D}_v$ over samples from $D_u$. When this quantity is bounded below by $\delta$, there can be a common event such that one can simultaneously distinguish $D_u$ and $D_v$ from $D_\emptyset$ with probability at least some $\delta'$ depending only on $\delta$. The statistical dimension in \eqref{eq:def_SDA} thus quantifies the measure of pairs of distributions with no such common events according to the prior distribution $\mu$. It is shown that the statistical dimension provides a lower bound on the SQ complexity of hypothesis testing with a VSTAT oracle~\cite{feldman2017statistical,brennan2020statistical}:

\begin{lemma}[Theorem 2.7 of~\cite{feldman2017statistical} and Theorem A.5 of~\cite{brennan2020statistical}]
Given the hypothesis testing problem $D_\emptyset$ vs. $\cS=\{D_u\}_{u\sim\mu}$. Any (possibly) randomized statistical query algorithm which solves the hypothesis testing problem with probability at least $1-\delta$ requires at least $(1-\delta)\text{SDA}(\cS,\mu,m)$ queries to VSTAT$(m/3)$.
\end{lemma}

\begin{remark}
We remark that the statistical dimension definition we used here is mildly stronger than the standard statistical dimension defined in~\cite{feldman2017statistical} and the statistical query model used in previous results in learning distributions of quantum circuits~\cite{nietner2023average,hinsche2021learnabilityoutputdistributionslocal}. In other words, the lower bound proved in our work automatically holds for both models mentioned above.
\end{remark}

\subsubsection{Equivalence}
\label{sec:SQ}
We now show the equivalence between copy-wise low-degree models introduced in the previous part and statistical query models in \Cref{sec:basic_sq} following the framework in~\cite{brennan2020statistical}. 

To begin with, we go back to the classical distinguishing task $D_\emptyset$ vs. $\cS=\{D_u\}_{u\sim\mu}$ on $\R^n$ as in \Cref{sec:basic_low_degree}. We now introduce the concept of niceness:
\begin{definition}[$(\delta,k)$-nice~\cite{brennan2020statistical}]\label{def:nice}
Given a null distribution $D_\emptyset$ on $\R^n$, a polynomial $p:(\R^n)^k\to\R$ of $k$ samples $x_1,...,x_k\in\R^n$ is called \emph{$k$ purely high degree} if it is orthogonal to any function $q:(\R^n)^k\to\R$ that has degree at most $k$ in each sample $x_1,...,x_k$, i.e., $\E_{x_1,...,x_k\sim D_\emptyset}p(x_1,...,x_k)q(x_1,...,x_k)=0$ for all function $q:(\R^n)^k\to\R$ that has degree at most $k$ in each sample $x_1,...,x_k$. The hypothesis testing problem $D_\emptyset$ vs. $\cS=\{D_u\}_{u\sim\mu}$ is called \emph{$(\delta,k)$-nice} if no $k$-purely high-degree function of $k$ samples can distinguish between the two cases with probability $\delta$.
\end{definition}

\Cref{def:nice} might look strange at first glance. However, the intuition here is that it rules out testing problems that can be solved using very few samples as usually we choose $k=\polylog(n)$ as degrees of functions. It is shown that classical random noise can be a sufficient condition for ensuring niceness:
\begin{lemma}[Theorem 5.2 of~\cite{brennan2020statistical}]\label{lem:classical_noise_nice}
Consider a hypothesis testing problem $D_\emptyset$ vs. $\cS=\{D_u\}_{u\sim\mu}$ and suppose that $D_\emptyset=D_{\emptyset,\text{ind}}^{\otimes n}$ is a product distribution. Suppose that there is no $k$-sample algorithm that can distinguish between the two cases with probability $\delta_0$. We construct $\cS'=\{D_u'\}_{u\sim\mu}$ by sampling $x_i'\sim D'_u,x_i'\in\R^n$ using a sample $x\sim D_u,x\in\R^n$ and replacing each coordinate $x_j$ independently by a fresh sample from $D_{\emptyset,\text{ind}}$ with probability $\kappa$. Then the distinguishing task $D_\emptyset$ vs. $\cS'$ is $(\delta_0(1-\kappa)^{k^2},k)$-nice.
\end{lemma}

\noindent Given the niceness condition, Ref.~\cite{brennan2020statistical} shows that 
\begin{lemma}[Theorem 1.6 of~\cite{brennan2020statistical}]
Consider a $(m^{-k/2}/4,k)$-nice hypothesis testing problem $D_\emptyset$ vs. $\cS=\{D_u\}_{u\sim\mu}$, we have:
\begin{itemize}
    \item If there exists some $0\leq m'\leq m$ such that $\text{SDA}(\cS,m',\mu)\leq\left(\frac{2m}{m'}\right)^k$, then we have the $4mk$-sample square degree-$(\infty,k^2)$ advantage bounded by $o(1)$.
    \item If the $m$-sample square degree-$(\infty,k)$ advantage is bounded by $o(1)$, then there exists $0\leq m'\leq m$ such that $\text{SDA}(\cS,m',\mu)\geq\left(\frac{2m}{m'}\right)^{\Omega(k)}$.
\end{itemize}
\end{lemma}

Now, we go to the quantum hypothesis testing problem between $\rhomm$ vs. $\rho\sim\mathcal{E}$. As we only consider single-copy PVMs here, the null distribution $D_\emptyset$ after measuring $\rhomm$ is always the uniform distribution, which is a product distribution of $(1/2,1/2)$ on each bit. We can obtain a quantum version of \Cref{lem:classical_noise_nice} using the local depolarization noise $\Lambda_\kappa^{\otimes n}$ with an identical depolarization channel $\Lambda_\kappa$ with noise rate $\kappa$ on each qubit, which is a common noise model considered on near-term noisy quantum devices~\cite{chen2022complexity}. Here, $\Lambda_\kappa$ is given by
\begin{align}\label{eq:local_depolarize}
\Lambda_\kappa(\rho)\to(1-\kappa)\rho+\kappa\frac{I_2}{2},
\end{align}
where $\rho$ is a single-qubit state. In real experiments, we usually have to implement a PVM $\{U\ket{x}\bra{x}U^\dagger\}_{x=1}^{2^n}$ by implementing the unitary $U$ first and perform a computational basis measurement. If we assume a readout error before the computational basis measurement in the form of local depolarization noise $\Lambda_\kappa^{\otimes n}$, it will naturally replace every bit of the readout information with $(1/2,1/2)$ with probability $\kappa$, which is exactly $D_{\emptyset,\text{ind}}$. We thus have the following observations:

\begin{observation}(Hypothesis testing with local depolarization noise as readout error is nice)
Consider a quantum state hypothesis testing problem $\rhomm$ vs. $\rho\sim\mathcal{E}$ and corresponding protocols using PVMs. Suppose that we suffer from a readout error before the computational basis measurement in the form of local depolarization noise $\Lambda_\kappa^{\otimes n}$. If there is no $k$-sample algorithm that can distinguish between the two cases with probability $\delta_0$. Then the noisy distinguishing task $D_\emptyset$ vs. $\cS$ is $(\delta_0(1-\kappa)^{k^2},k)$-nice.
\end{observation}

In the low-degree framework, we usually assume that $k=\polylog(n)$ and $k=\omega(\log n)$. In this case, we can always find some \textbf{constant} threshold $\kappa_0$ such that $\delta_0(1-\kappa)^{k^2}\leq m^{-k/2}/4$ for any $m=\poly(n)$ and constant $\delta_0$. We thus have the following equivalence between the statistical query model and the copy-wise low-degree model when we have local depolarization noise as readout error:

\begin{theorem}[Equivalence between statistical query model and the copy-wise low-degree model under noise]Consider solving a quantum state hypothesis testing problem $\rhomm$ vs. $\rho\sim\mathcal{E}$ using protocols with $m=\poly(n)$ single-copy PVMs. Assume that there is no $k$-sample algorithm for $k=\omega(\log n)$ that can distinguish between the two cases with \emph{constant} probability $\delta_0$. Then there exists some constant threshold $\kappa_0$ such that if there is a readout error in the form of local depolarization noise $D_\kappa^{\otimes n}$ with $\kappa>\kappa_0$, we have:
\begin{itemize}
    \item If there exists some $0\leq m'\leq m$ such that $\text{SDA}(\mathcal{E},m',\mu)\leq\left(\frac{2m}{m'}\right)^k$, then the square degree-$(\infty,k^2)$ advantage $\chi_{\leq(\infty,k^2)}^2$ (see \Cref{sec:basic_low_degree} for the formal definition)
    is bounded by $o(1)$.
    \item If the optimal distinguishing advantage achieved by any $4mk$-sample copy-wise degree-$\leq(\infty,k)$ distinguished (i.e. the square degree-$(\infty,k)$ advantage $\chi_{\leq(\infty,k)}^2$) is bounded by $o(1)$, then there exists $0\leq m'\leq m$ such that $\text{SDA}(\mathcal{E},m',\mu)\geq\left(\frac{2m}{m'}\right)^{\Omega(k)}$.
\end{itemize}
\end{theorem}

\subsection{No-go results for other models of adaptivity}\label{sec:general_adaptivity_full}

It is natural to ask what would happen if we considered protocols with a greater level of adaptivity. The difficulty stems mainly from the fact that our model focuses on the complexity of post-processing and does not take into account the computational complexity of making adaptive choices of bases to measure in. As such, it is possible to ``hide'' unbounded computation in the adaptive choices.

\subsubsection{\texorpdfstring{$k$}{k}-gram copy-wise adaptivity}

For starers, one might consider a $k$-gram adaptivity model. In this model, we choose the measurement bases for each copy based on the output of the previous $k-1$ copies. Unfortunately, one can show that if the hypothesis testing task is \emph{information-theoretically} possible using any number of single-copy measurements, possibly arbitrarily adaptive, there is always a 2-gram adaptive sequence of measurements for which the classical post-processing is low-degree tractable, in fact, trivial.

The construction for the cheating algorithm is given as follows. Denote the adaptive sequence of PVMs that would information-theoretically suffice to solve the problem as $\{\{U_i\ket{x}\bra{x}U_i^\dagger\}_{x=1}^{2^n}\}_{i=1}^m$. Note that for each $U_i$, we can decompose it into constant terms in the form of a direct product of the first $n-1$ qubits and the last qubit, so we can simulate this protocol using a polynomial number of PVMs on $n-1$ qubits instead of $n$ qubits. We can then measure $\text{poly}(n)$ copies of the state and record the whole history in a Markovian way in the last qubit. The reason is that we can encode arbitrarily many bits of information into $SU(2)$, so we can always record the history of all results of the protocol being simulated up to a given point into the choice of the measurement for the last qubit of the next copy being measured. After the simulation is finished, we adaptively choose one last basis to measure in: we perform an inefficient computation to post-process the information encoded in the last qubit of the previous qubit and determine the answer to the hypothesis testing problem, and then we measure the last qubit of the final copy in either the $X$ or $Z$ basis depending on the answer. Having hidden all of this computation in the adaptive choice of bases, we can trivially post-process the readouts to determine the answer by looking at what basis we measured the last qubit of the final copy in.

\subsubsection{Arbitrary single-qubit measurement with low-degree adaptivity}

The above suggests that merely restricting the adaptivity to depend on the readouts for the previous copy does not suffice, because too much computation can be hidden even into such adaptive choices.

Another natural model to consider is a low-degree adaptivity model. In this model, we are restricted to using $m$ single-qubit measurements, and the choice of the basis on each qubit in the total $mn$ qubits can depend on at most the outcomes of the previous $k\geq 3$ qubits. For this model, we also show that it is possible to hide exponential computational complexity in computing the adaptivity if we assume that the task is sample efficient for local measurements.

The construction for the algorithm is given as follows. Note that the overall probability distribution based on performing a local measurement on a copy at once should be the same as performing it qubit by qubit, so we consider the latter one.

We keep a classical counter. We just measure each qubit on the computational basis first and add the counter by one after each qubit. When the counter goes to $k$, we store the history of the previous $k-1$ qubits on the choice of basis we measure on the $k$-th qubit (recall that $SU(2)$ can store infinitely many bits of information). Next, we set the counter to $0$. We again measure each qubit on the computational basis and add the counter by one after each qubit. When the counter again goes to $k$, we store the previous history of all $2k$ qubits using the $k$-th qubit basis and the results from $(k+1)$-th to $(2k-1)$-th qubit. We repeat the above procedure for $\poly(n)$ time and, as in the previous section, we can collect the information from all readouts for the protocol being simulated into the choice of a basis for one qubit in $SU(2)$.

\subsubsection{Discrete single-qubit basis measurement with low-degree adaptivity}

Both the above algorithms stem from the fact that the low-degree model is a poor proxy for the complexity of adaptivity in situations where one can use arbitrarily large bit complexity. It is natural to hope that if we work with a discrete set of measurement bases for each qubit, then these pathologies might go away.

One exemplary model is when one can perform $m$ single-qubit measurements. However, the choice of the basis on each qubit is a low-degree function of degree $k$ on all previous history, and the possible choices are within the single-qubit Pauli basis $\{\{\ketbra{0},\ketbra{1}\},\{\ketbra{+},\ketbra{-}\},\{\ketbra{+i},\ketbra{-i}\}\}$. Quantitatively, we suppose that every adaptive choice of single-qubit Pauli basis to measure is a low-degree junta $\{0,1,+,-,+i,-i\}\to\{\sigma_x,\sigma_y,\sigma_z\}$. In the following, we show that even if the adaptivity choice function is a function of $O(1)$ previous single-qubit measurement outcomes, proving a hardness result is as hard as proving a circuit lower bound if the problem can be information-theoretically solved using $m_{\text{info}}=\poly(n)$ single-qubit measurements.

Without loss of generality, we suppose the problem can be solved in information-theoretically solved using $m_{\text{info}}=\poly(n)$ computational basis measurements. To show a lower bound, we at least need to rule out the following family of strategies: for copies $i=m_{\text{info}}+1,...,m_{\text{info}}+\poly(n)$, the basis in which the first qubit of copy $i$ is measured depends on either $O(1)$ single-qubit outcomes from the first $m_{\text{info}}$ copies or $O(1)$ of the bases in which the first qubits of copies $m_{\text{info}}+1,...,m_{\text{info}}+i-1$ are measured. At the end of the protocol, one looks at the basis in which the first qubit of the last copy was measured and outputs a Boolean predicate of that.

Note that we can think of the first qubits of these copies as effectively computing intermediate gates of a Boolean circuit whose inputs are the bits from measuring the first $m_{\text{info}}$ copies. To rule out the possibility that there is some protocol of the above form, we would have to prove there is no poly-sized circuit that processes the information from the first $m_{\text{info}}$ computational basis measurements and solves the detection problem.

\section{Applications of the general framework}
\label{sec:applications}

\subsection{Low-degree hardness of error mitigation}\label{sec:qem}

In this section, we will apply the low-degree framework to show an information-computation gap for error mitigation. Previously, Ref.~\cite{quek2024exponentially} proved that quantum error mitigation is \emph{statistically} hard for 1D geometrically local noisy quantum circuits of depth $\omega(\log n)$ at \emph{constant} noise rate~\cite{quek2024exponentially}. Here we prove that even with significantly less noise, in fact at a rate inverse polynomial in the system size and with only a \emph{constant} number of layers of noise, error mitigation is low-degree hard against protocols with a polynomial number of single-qubit measurements, even though it is statistically tractable for protocols with polynomial number of noisy data samples. 

In Section~\ref{sec:mitigation_prelims} we formally define the problem and lower bound instance, and in Section~\ref{sec:mitigation_lbd} we formally state and prove our hardness result.

\subsubsection{Error mitigation preliminaries}
\label{sec:mitigation_prelims}

For any circuit consisting of $l$ layers, $\mathcal{C} = \mathcal{U}_l\circ\ldots \circ\mathcal{U}_1$ where $\mathcal{U}_i(\cdot) := U_i(\cdot)U_i^{\dagger}$, let $\tilde{\Phi}_{\mathcal{C}}$ denote the noisy version of $\mathcal{C}$, that is, a channel acting on $X\in \mathcal{H}_n$ as $X \mapsto \mathcal{D}_{\kappa}^{\otimes n} \circ \mathcal{U}_l \circ \ldots \circ\mathcal{D}_{\kappa}^{\otimes n} \circ \mathcal{U}_1(X)$ where each $\mathcal{D}_{\kappa}^{\otimes n}$ is a layer of single-qubit depolarizing noise.

Error mitigation is defined as the following problem:
\begin{definition}[Error mitigation]
Upon input of:
\begin{enumerate}
\item a classical description of a noiseless circuit $\mathcal{C}$ and $a$ finite set $\mathcal{M}=\left\{O_i\right\}$ of observables such that $\max_i\lVert O_i \rVert<1$;
\item copies of the output state $\sigma^{\prime}$ of the noisy circuit $\tilde{\Phi}_{\mathcal{C}}(\rho_{\text{in}})$, and the ability to perform collective measurements,
\end{enumerate}
 output $1/\poly(n)$-precise estimates of the expectation values $\tr\left(O_i \mathcal{C}(\rho_{\text{in}})\right)$ for each $O_i \in \mathcal{M}$.
The number $m$ of copies of $\sigma^{\prime}$ needed is the sample complexity of error mitigation.
\end{definition}

\noindent To show a hardness result for this problem, we will consider the following hypothesis testing task: 
\begin{problem}[Noisy circuit hypothesis testing\label{prob:noisy_hypo}]
    Given $m$ copies of some state $\rho$ and the description of some circuit $\mathcal{C} \sim \mathcal{E}$, decide if $\rho$ is
\begin{itemize}
\item Null hypothesis: $I/2^n = \tilde{\Phi}_{\mathcal{C}} \circ \mathcal{C}^{\dagger}(I/2^n)$
\item Alternative hypothesis: $\tilde{\Phi}_{\mathcal{C}} \circ \mathcal{C}^{\dagger}(\ket{\overline{0}}\bra{\overline{0}})$
\end{itemize}
promised that one of the two is the case. 
\end{problem}

\noindent Hardness for this task implies hardness for error mitigation as follows. By using just a single observable, $\mathcal{M} = \{\ket{\overline{0}}\bra{\overline{0}}\}$ and setting  $\rho_{\text{in}} = \mathcal{C}^{\dagger}(I/2^n)$ (in the null hypothesis) or $\rho_{\text{in}} = \mathcal{C}^{\dagger}(\ket{\overline{0}}\bra{\overline{0}})$ (in the alternative hypothesis), an algorithm that solves error mitigation can certainly solve \Cref{prob:noisy_hypo}. 

Note that we have used a different reduction from that in \cite{quek2024exponentially}. The upshot is that to prove low-degree hardness of error mitigation, it suffices to prove the low-degree hardness of the hypothesis testing problem, which we will do in the rest of this section. While the hard instance in \cite{quek2024exponentially} assumes that there is a layer of single-qubit depolarization noise after each layer of the circuit, we consider the ensemble of quantum circuits that interleaves a constant number of random circuits of depth $\omega(\log(n))$ that form approximate unitary $2$-designs and local depolarization noise of noise rate $O(1/n^{1-\delta'})$ here. Therefore, the total layer of noise in our model is only constant while the depth of the circuit can be up to $\omega(\log n)$.

\subsubsection{Lower bound}
\label{sec:mitigation_lbd}

\begin{theorem}[Extremely low-noise error mitigation is low-degree hard\label{thm:EM}]
For any $\delta,\delta'<1$, there exists an ensemble of noisy circuits on geometrically-local architectures in any dimension, with only constantly-many layers of single-qubit depolarizing noise at rate $\kappa = \Omega(1/n^{1-\delta'})$, such that error mitigation is degree-$k$ hard for algorithms that make single-qubit measurements on polynomially-many copies of the noisy state, or polynomially-many non-adaptive measurements prepared by depth $k$ geometrically local circuits. That is, the hardness of error mitigation persists even in this extremely low-noise setting.
\end{theorem}

\noindent Even at constant depolarizing noise rate (higher than what we consider in the above theorem), error mitigation is only known to become statistically hard for a geometrically local circuit after $\omega(\log(n))$ layers of noise have acted. This is because every additional layer of noise in the circuit results in a multiplicative contraction in relative entropy to the maximally mixed state\footnote{It was proven in \cite{quek2024exponentially} that the hardness of error mitigation kicks in at $\poly \log \log(n)$ depth on circuits with all-to-all connectivity. When restricted to $\mathscr{D}$-dimensional geometrically local connectivity, Ref.~\cite{quek2024exponentially} proved that quantum error mitigation becomes {\em statistically} hard only after $\min(O(\log(n)^{2/\mathscr{D}} \poly\log \log(n), \log(n))$ layers of constant rate noise -- while we only require constantly many layers of extremely low noise.}. Since we only have a {\em constant} number of layers of noise at an extremely low rate $1/n^{1-\delta'}$ in the above setting, that is far from enough to render the problems statistically hard. Thus, \Cref{thm:EM} suggests an information-computation gap for this setting. 

\begin{fact}\label{fact:EMpaper}
For a circuit $\mathcal{C} = \mathcal{U}_l\circ \ldots \circ\mathcal{U}_1$, let $\rho^{(l)}:=\tilde{\Phi}_{\mathcal{C}}(\rho_{\text{in}})$. If each $U_i$ is sampled from an exact $n$-qubit $2$-design,
\begin{equation*}
\mathbb{E}_{\mathcal{U}_1^{l}}\tr({\rho^{(l)}}^2)  \leq c^{nl}(1-2^{-n})+\frac{1}{2^n}
\end{equation*}
where $c$ is a noise-dependent parameter,
\begin{equation*}
    c = \frac{1+3(1-\kappa)^2}{4}.
\end{equation*}
\end{fact}
This fact is proven in Ref. \cite{quek2024exponentially}. When $\kappa = \frac{1}{n^{1-\delta'}}$ for $0<\delta'<1$, 
\begin{equation*}
    \lim_{n\rightarrow \infty} c^{nl} = e^{-\frac{3}{2}ln^{\delta'}}.
\end{equation*}
To see this, let $y= c^{nl},$ and consider that:
\begin{align*}
    \log(y) &= nl \log\left(\frac{1 + 3(1-\frac{1}{n^{1-\delta'}})^2}{4}\right)\\
    &= nl \log( 1- \frac{3}{2n^{1-\delta'}} + \frac{3}{4n^{2(1-\delta')}} + O(\frac{1}{n^{3(1-\delta')}}))\\
    & = nl \left[- \frac{3}{2n^{1-\delta'}} + \frac{3}{4n^{2(1-\delta')}} + \frac{9}{8n^{2(1-\delta')}} + O\left(\frac{1}{n^{3(1-\delta')}}\right)\right]\\
\end{align*}
and therefore 
\begin{align*}
    \lim_{n\rightarrow \infty } \log(y) = -\frac{3}{2}ln^{\delta'}
\end{align*}
as desired. This computation illustrates that in the expression
\[
  c^{nl}(1 - 2^{-n}) + l\epsilon
\]
which appears in the next corollary, the term $l\epsilon$ asymptotically (in $n$) dominates the term $c^{nl}$ at noise rate $\kappa = 1/{n^{1-\delta'}}$ for constant $\epsilon$. 

\begin{corollary}\label{corr:EM}
For any $a, l \in \mathbb{Z}_{+}$ and $\epsilon, \epsilon^{\ast} \in \mathbb{R}_{+}$ let 
\begin{equation*}
R(a):=2^{a-n}\left[c^{n l}\left(1-2^{-n}\right)+l \epsilon\right]+\epsilon^{\ast}.
\end{equation*}
 Define the random state  $\rho^{(l)} := \tilde{\Phi}_{\mathcal{C}}(\rho_{\text{in}})$ where $\mathcal{C} = \mathcal{U}_l\circ \ldots \mathcal{U}_1$, each $\mathcal{U}_i$ is sampled from an $\epsilon$-approximate 2-design for $i=1,\ldots l$ and $\mathcal{U}_l$ is sampled from an $\epsilon^{\ast}$-approximate $2$-design. 

For any $A \subseteq [n]$,  and input state $\rho_{\text{in}}$, we have:
\begin{equation*}
\Pr_{\mathcal{U}_1,\ldots, \mathcal{U}_l} \left[ 
  \left\| 
    \tr_{[n] \setminus A}[\rho^{(l)}] - \frac{I_A}{2^{|A|}} 
  \right\|_1 \geq 2^{|A|/2} R(|A|)^{1/4} 
\right] \leq R(|A|)^{1/2}.
\end{equation*}
\end{corollary}

\begin{proof} 
Suppose instead the last block $\mathcal{U}_l$ were sampled from an exact $2$-design $\mu(2)$, instead of an approximate one. Let us first consider the expected value (over this alternative distribution) of the quantity
\begin{align}
\mathbb{E}_{\mathcal{U}_l \sim \mu(2)}\left[ 
  \left\| 
    \tr_{[n] \setminus A} \left[ \rho^{(l)} \right] 
    - \frac{I_A}{2^{|A|}} 
  \right\|_2^2 
\right]
&= \mathbb{E}_{\mathcal{U}_l \sim \mu(2)} \Biggl[ 
  \tr \left( \left( \tr_{[n] \setminus A}[\rho^{(l)}] \right)^2 \right)
  - 2 \tr \left( 
    \frac{ \tr_{[n] \setminus A}[\rho^{(l)}] }{2^{|A|}} 
  \right)
  + \tr \left( \frac{ I_A }{ 2^{2|A|} } \right) 
\Biggr] \nonumber\\
&= \mathbb{E}_{\mathcal{U}_l \sim \mu(2)} \left[ 
  \tr \left( \left( \tr_{[n] \setminus A}[\rho^{(l)}] \right)^2 \right)
\right] - 2^{-|A|} \nonumber\\
&= 2^{|A|-n} \left[ 
  \tr \left( (\rho^{(l-1)})^2 \right) - \frac{1}{2^n} 
\right]. \label{eq:exact}
\end{align}

Here we have integrated over the Haar measure in \Cref{eq:exact}. However, since the actual distribution from which $\mathcal{U}_l$ is sampled is an $\epsilon^{\ast}$-approximate $2$-design, the expectation over the actual distribution satisfies
\begin{align*}
\mathbb{E}_{\mathcal{U}_l} \left[ 
  \left\| 
    \tr_{[n] \setminus A}[\rho^{(l)}] - \frac{I_A}{2^{|A|}} 
  \right\|_2^2 
\right] 
\leq 2^{|A|-n} \left[ 
  \tr \left( (\rho^{(l-1)})^2 \right) - \frac{1}{2^n} 
\right] + \epsilon^{\ast}.
\end{align*}

Taking expectation over the remaining $l-1$ blocks:
\begin{align*}
\mathbb{E}_{\mathcal{U}_1,\ldots, \mathcal{U}_l} \left[ 
  \left\| 
    \tr_{[n] \setminus A}[\rho^{(l)}] - \frac{I_A}{2^{|A|}} 
  \right\|_2^2 
\right] 
&= 2^{|A|-n} \left[ 
  \mathbb{E}_{\mathcal{U}_1, \ldots, \mathcal{U}_{l-1}} 
  \left[ \tr \left( (\rho^{(l-1)})^2 \right) \right] 
  - \frac{1}{2^n} 
\right] + \epsilon^{\ast} \\
&\leq 2^{|A|-n} \left[ 
  c^{nl}(1 - 2^{-n}) + \frac{1}{2^n} + l\epsilon - \frac{1}{2^n} 
\right] + \epsilon^{\ast} \\
&= 2^{|A|-n} \left[ 
  c^{nl}(1 - 2^{-n}) + l\epsilon 
\right] + \epsilon^{\ast}.
\end{align*}
where in the inequality we have used \Cref{fact:EMpaper}. 

By Markov's inequality:
\begin{align*}
\Pr_{\mathcal{U}_1,\ldots, \mathcal{U}_l} \left[ 
  \left\| 
    \tr_{[n] \setminus A}[\rho^{(l)}] - \frac{I_A}{2^{|A|}} 
  \right\|_2 \geq R^{1/4} 
\right] 
= \Pr \left[ 
  \left\| 
    \cdots 
  \right\|_2^2 \geq R^{1/2} 
\right] 
\leq R^{1/2}.
\end{align*}

Converting $\ell_2$ to $\ell_1$ via Cauchy-Schwarz:
\begin{equation*}
\Pr_{\mathcal{U}_1,\ldots, \mathcal{U}_l} \left[ 
  \left\| 
    \tr_{[n] \setminus A}[\rho^{(l)}] - \frac{I_A}{2^{|A|}} 
  \right\|_1 \geq 2^{|A|/2} R^{1/4} 
\right] \leq R^{1/2}.
\end{equation*}
\end{proof}

\noindent We now consider the full scenario on $m$ copies of the system, each copy being output by the same circuit of $n$ qubits that is sampled once at the start. Let $T \subseteq [mn]$. We may write any such $T$ as a union over the elements of $T$ in each of the copies: $T = T_1\cup\ldots T_m$, meaning that the overall product state can also be decomposed as:
\begin{equation*}
    \tr_{[mn]\setminus T}\bigl[(\rho^{(l)})^{\otimes m}\bigr] = \bigotimes_{i=1}^m  \tr_{[n]\setminus T_i}\bigl[\rho^{(l)}\bigr] 
\end{equation*}

By a union bound over all $m$ copies,  one then has that with probability at least 
\begin{equation*}
1-\sum_{i=1}^m R(|T_i|)^{1/2}\geq 1- m R(|T|)^{1/2}
\end{equation*}

\begin{equation*}
    \Bigl\|
  \tr_{[n]\setminus T_i}\bigl[\rho^{(l)}\bigr] - 
  \frac{I_{T_i}}{2^{|T_i|}}
\Bigr\|_1\leq 
2^{|T_i|/2} R(|T_i|)^{1/4}
 \qquad \forall i\in [m].
\end{equation*}

Again using the identity $\lVert A\otimes B - C\otimes D \rVert_1 \leq \lVert B - D \rVert_1 + \lVert A - C \rVert_1$  $m$ times, we obtain that
\begin{align*}
&\Pr_{\mathcal{U}_1,\dots,\mathcal{U}_\ell}
\Bigl[
  \Bigl\|
   \tr_{[mn]\setminus T}\bigl[(\rho^{(l)})^{\otimes m}\bigr] 
  -
  \frac{I_T}{2^{|T|}}
\Bigr\|_1
  \;\le\;
 \sum_{i=1}^m 2^{|T_i|/2} R(|T_i|)^{1/4}
\Bigr]\\
&\geq \Pr_{\mathcal{U}_1,\dots,\mathcal{U}_\ell}
\Bigl[
  \Bigl\|
   \tr_{[mn]\setminus T}\bigl[(\rho^{(l)})^{\otimes m}\bigr] 
  -
  \frac{I_T}{2^{|T|}}
\Bigr\|_1
  \;\le\;
m2^{|T|/2} R(|T|)^{1/4}
\Bigr]\\
&\geq 1- mR(|T|)^{1/2}.
\end{align*}

Therefore
\begin{align*}
    \mathbb{E}_{\mathcal{U}_1,\dots,\mathcal{U}_\ell} \Bigl\|
\tr_{[mn]\setminus T}\bigl[(\rho^{(l)})^{\otimes m}\bigr]
  -
  \frac{I_T}{2^{|T|}}
  \Bigr\|_1 &\leq
  m2^{|T|/2} R(|T|)^{1/4} + mR(|T|)^{1/2}\\
  &\leq m 2^{|T|/2} R(|T|)^{1/4}
\end{align*}

By the strong convexity of trace distance, 
\begin{equation*}
  \Bigl\|
   \mathbb{E}_{\mathcal{U}_1,\dots,\mathcal{U}_\ell}\left(\tr_{[mn]\setminus T}\bigl[(\rho^{(l)})^{\otimes m}\bigr] \right)
  -
  \frac{I_T}{2^{|T|}}
  \Bigr\|_1 \leq m 2^{|T|/2} R(|T|)^{1/4}
\end{equation*}

$\rho_k^{(l)}$ is a state output by $l$ $\epsilon$-approximate 2-designs alternating with a layer of depolarizing noise. So far, we have been agnostic as to how we generate the $\epsilon$-approximate $2$ designs, but as we want to obtain the shallowest depth at which error mitigation fails, now we prescribe that we use the shallowest $\epsilon$-approximate $2$-design generators we know of. We use the prescription of \cite{schuster2024random} to implement all the blocks with geometrically-local circuits on any architecture, in depth $\log(n/\epsilon)$ for the first $n-1$ blocks and depth $\log(n/\epsilon^{\ast})$ for the last.

By the calculations below \Cref{fact:EMpaper}, for the extremely low noise rate we consider,
\begin{equation*}
R(|T|):=2^{|T|-n}\left[c^{n l}\left(1-2^{-n}\right)+l \epsilon\right]+\epsilon^{\ast} \approx 2^{|T|-n}l\epsilon +\epsilon^{\ast}. 
\end{equation*}
Then, if we choose $k=\log(n)$, $\epsilon^{\ast} = 2^{-\Omega(k)\log(n)}$, it is easy to check that $\epsilon^{\ast} > 2^{|T|-n}l\epsilon$ for any $|T|<k$, and so for this choice of parameters, 
\begin{equation}\label{eq:nqubit_2}
  \Bigl\|
   \mathbb{E}_{\mathcal{U}_1,\dots,\mathcal{U}_\ell}\left(\tr_{[mn]\setminus T}\bigl[(\rho^{(l)})^{\otimes m}\bigr] \right)
  -
  \frac{I_T}{2^{|T|}}
  \Bigr\|_1 \leq m 2^{|T|/2} O(\epsilon^{\ast})^{1/4}
\end{equation}
To bound the performance of low-degree algorithms using single-qubit measurements, since we only need to consider $|T|<k$ and $m=\poly(n)$, we may check that the right hand side of \Cref{eq:nqubit_2} becomes $2^{-\Omega(k \log n)},$ which allows us to apply \Cref{thm:general_single_qubit}. A similar argument and application of \Cref{coro:general_bounded_depth} allows us to prove the statement for measurements of bounded depth $\Omega(k)$.

\subsection{Low-degree hardness of learning and testing random circuits}\label{sec:circuit}
We now turn to another application of our framework, learning and testing quantum circuits of bounded depth~\cite{huang2024learning,landau2024learning,zhao2024learning,nietner2023average,vasconcelos2024learning}. In~\cite{nietner2023average}, it is shown that learning the distribution of measuring a state prepared by a quantum circuit of depth $\omega(\log n)$ is hard for the statistical query model. For shallow circuits of constant depth, however, Ref~\cite{huang2024learning,landau2024learning} provides computational- and sample-efficient algorithms to learn unitaries and states. It is natural to wonder if this task is hard under the low-degree framework after a certain depth threshold.

\subsubsection{Random quantum circuits and circuit moment bounds}\label{sec:basic_circuit_moment}

Throughout this work, we consider geometrically-local quantum circuits, which contain only local single- and two-qubit gates on neighboring two qubits. More concretely, we focus on the brickwork architecture defined as follows:
\begin{definition}
A $n$-qubit brickwork quantum circuit of depth $L$ with a periodic boundary condition is a quantum circuit that is of the form
\begin{align*}
U=(U_{2,3}^{(L)}\otimes...\otimes U_{n,1}^{(L)})(U_{1,2}^{(L-1)}\otimes...\otimes U_{n-1,n}^{(L-1)})...(U_{2,3}^{(2)}\otimes...\otimes U_{n,1}^{(2)})(U_{1,2}^{(1)}\otimes...\otimes U_{n-1,n}^{(1)}),
\end{align*}
where $U_{i,j}^{(l)}$ is the a two-qubit gate from $U(4)$ acting on the $i$-th and the $j$-th qubit in the $l$-th layer. Here, without loss of generality, we assume $L$ and $n$ are even.
\end{definition}

We consider two kinds of random quantum circuit ensembles throughout this work. The first kind of ensemble we consider is called \emph{brickwork random quantum circuits}, which is a random brickwork quantum circuit with each two-qubit gate chosen Haar randomly.

The second kind of ensemble we consider is the \emph{coarse-grained random quantum circuit}. Consider a probability distribution over random circuits consisting of two steps: First, we apply Haar random unitaries to $n/n'$ blocks each consisting of $n'$ qubits. Second, we shift the layers such that the gates of the two layers overlap on $n'/2$ and apply Haar random unitaries to the new $n/n'$ blocks. Moreover, we denote by $\nu_{n',n,L}$ the ensemble as above where, instead of Haar random unitaries, we draw the unitary on the blocks of $n'$ qubits from a brickwork random quantum circuits of depth $L$.

For these ensembles of random quantum circuits, we summarize the results from~\cite{chen2024incompressibility,brandao2016local,haferkamp2022random,schuster2024random}:
\begin{lemma}[Combination of results in~\cite{haferkamp2022random,schuster2024random}]\label{lem:circuit_design}
We have that:
\begin{itemize}
    \item Brickwork random quantum circuits are a $\epsilon$-approximate unitary $t$-design in depth $O(nt\log^7(t)+\log(1/\epsilon))$.
    \item Coarse-grained random quantum circuits $\nu_{C\log(nt/\epsilon),n,Ct\log(nt/\epsilon)}$ is a $\epsilon$-approximate unitary $t$-design for a sufficiently large constant $C>0$.
\end{itemize}
Combining these two results together, we consider the ensemble $\mu_{n',n,L}$ that is similar to $\nu_{n',n,L}$ but replacing the $n'$-qubit Haar random unitary with a brickwork random quantum circuits with large enough depth. We have $\mu_{C\log(nt/\epsilon),n,Ct\log(nt/\epsilon)}$ is a $\epsilon$-approximate unitary $t$-design for a sufficiently large constant $C>0$.
\end{lemma}

Based on the general framework, the low-degree hardness for learning and testing quantum circuits follows directly.

Here, we consider the following hypothesis testing problem:
\begin{itemize}
    \item $H_0$ (null case): the state is the maximally mixed state $\rhomm=I/2^n$.
    \item $H_1$ (alternative case): the state $\rho$ is randomly chosen from a random geometrically local circuit of depth $L$ according to $\mu_{n',n,L}$ (see \Cref{sec:basic_circuit_moment} for the definition).
\end{itemize}

Here, we focus on states prepared by circuits of depth $L$~\cite{huang2024learning}. Suppose we have an algorithm that can learn a state prepared by a circuit of depth $L$ using single-copy PVMs, which given copies of $U\ket{0^n}$ promised that $U$ is prepared by a quantum circuit of depth $L$, outputs a classical description of the circuit that implements unitary $U'$ such that $d_\tr(U\ket{0^n},U'\ket{0^n})\leq\epsilon$. We can then solve the hypothesis testing using single-copy ancilla-free measurements via the following procedure:
\begin{enumerate}
    \item We first implement the algorithm and output a unitary $U'$.
    \item We then implement the measurement $\{U'^\dagger\ket{x}\bra{x}U'\}_{x=1}^{2^n}$, which is a unitary circuit $U'^\dagger$ followed by a computational basis measurement.
\end{enumerate}
In the null case $H_0$, we are sampling from a uniform distribution. However, in the alternative case $H_1$, we are sampling from a state that is $\epsilon$-close to $\ket{0}$.

\subsubsection{Low-degree hardness for non-adaptive single-copy measurements}\label{sec:circuit_single_copy}
In this subsection, we consider non-adaptive protocols with single-copy PVMs. By \Cref{coro:general_ancilla}, given $k\leq O(n)$ and an ensemble $\mathcal{E}$ that is $2^{-\tOmega(k\log n)}$-approximate state $2$-design, weak detection is degree-$k$ hard for any protocol between $\rhomm$ vs. $\mathcal{E}$ using polynomially many non-adaptive PVMs $\{\{U_i\ket{x}\bra{x}U_i^{\dagger}\}_{x=1}^{2^{n+n'}}\}_{i=1}^m$ with up to $n'=\Theta(n)$ ancillary qubits. By \Cref{lem:circuit_design}, $\mu_{C\log(nt/\epsilon),n,Ct\log(nt/\epsilon)}$ is a $\epsilon$-approximate unitary $t$-design for a sufficiently large constant $C>0$. To form a $2^{-\tOmega(k\log n)}$-approximate state $2$-design, we only need circuit depth of scaling
\begin{align*}
O(\log(nt/\epsilon))=\tO(k\log n),
\end{align*}
where we use $t=2$.

Therefore, combining these two results together, we obtain the following low-degree hardness results for non-adaptive protocols with single-copy measurements and bounded ancillary qubits:
\begin{corollary}\label{coro:circuit_single_copy}
For the hypothesis testing problem $\rhomm$ vs. $\rho\sim\mu_{\tO(k\log n),n,\tO(k\log n)}$ (i.e. an random circuit ensemble of depth $\tO(k\log n)$), weak detection is degree-$k$ hard for any protocol that makes polynomially many non-adaptive single-copy PVMs $\{\{U_i\ket{x}\bra{x}U_i^{\dagger}\}_{x=1}^{2^{n+n'}}\}_{i=1}^m$ with up to $n'=\Theta(n)$ ancillary qubits.
\end{corollary}

By the correspondence between the hypothesis testing and the learning problem, we have the following corollary.

\begin{corollary}\label{coro:circuit_single_copy_learning}
Given any $k=O(n)$, learning quantum circuits of depth $\tO(k\log n)$ is degree-$k$ hard for degree-$k$ protocols with a polynomial number of non-adaptive single-copy PVMs with up to $n'=\Theta(n)$ ancillary qubits.
\end{corollary}

\subsubsection{Low-degree hardness for models with restricted adaptivity}\label{sec:circuit_adaptivity}
We then consider adaptive protocols. We first consider the model of restricted adaptivity within blocks, where we have $m_1=\poly(n)$ blocks, each of $m_0=\polylog(n)$ adaptive copies. Yet the blocks are independent from one another. (see \Cref{sec:general_adaptivity_within_block}). By \Cref{coro:general_adaptivity_within_block}, if $\mathcal{E}$ forms a $2^{-\Omega(k^4m_0^2\log^2 k\log m_1\log^2 m_0)}$-approximate state $2km_0$-design, weak detection between $\mathcal{E}$ and $\rhomm$ is degree-$(\cD,k)$ hard for any adaptivity within round-based model protocol that makes polynomially many single-copy PVMs $\{\{U_i\ket{x}\bra{x}U_i^{\dagger}\}_{x=1}^{2^n}\}_{i=1}^m$. These protocols use $m_1=\poly(n)$ blocks each with $m_0=m/m_1=\polylog(n)$ samples, and adaptivity is only allowed in each block. By \Cref{lem:circuit_design}, $\mu_{C\log(nt/\epsilon),n,Ct\log(nt/\epsilon)}$ is a $\epsilon$-approximate unitary $t$-design for a sufficiently large constant $C>0$. To form a $2^{-\Omega(k^4m_0^2\log^2 k\log m_1\log^2 m_0)}$-approximate state $2km_0$-design, we only need circuit depth of scaling
\begin{align*}
O(\log(nt/\epsilon))=\tO(k^4m_0^2\log m_1),
\end{align*}
where we use $t=2km_0$. We thus have the following result:
\begin{corollary}\label{coro:circuit_adaptivity_within_block}
For the hypothesis testing problem $\rhomm$ vs. $\rho\sim\mu_{\tO(k^4m_0^2\log m_1),n,\tO(k^4m_0^2\log m_1)}$ (i.e. an random circuit ensemble of depth $\tO(k^4m_0^2\log m_1)$), weak detection is degree-$(\cD,k)$ hard for any adaptivity within round-based model protocol making polynomially many single-copy PVMs $\{\{U_i\ket{x}\bra{x}U_i^{\dagger}\}_{x=1}^{2^n}\}_{i=1}^m$. These protocols use $m_1=\poly(n)$ blocks each with $m_0=m/m_1=\polylog(n)$ samples, and adaptivity is only allowed in each block. Moreover, this indicates that learning quantum circuits of depth $\tO(k^4m_0^2\log m_1)$ is hard any adaptivity within round-based model protocol with a degree at most $(\cD,k)$ and $m=\poly(n)$ single-copy PVMs.
\end{corollary}

We then consider the model of restricted adaptivity among blocks, where we have $m_1=\text{polylog}(n)$ adaptive blocks of $m_0=\text{poly}(n)$ non-adaptive copies within each block. By \Cref{coro:general_adaptivity_among_block}, if $\mathcal{E}$ forms a $2^{-\Omega(km_1^2\log^2 m_1\log m_0\log k)}$-approximate state $2k$-design, weak detection is degree-$(\cD,k)$ hard for any adaptivity among round-based model protocol that makes polynomially many single-copy PVMs $\{\{U_i\ket{x}\bra{x}U_i^{\dagger}\}_{x=1}^{2^n}\}_{i=1}^m$. These protocols use $m_1=\polylog(n)$ blocks each with $m_0=m/m_1=\poly(n)$ samples, and adaptivity is only allowed among blocks. By \Cref{lem:circuit_design}, $\mu_{C\log(nt/\epsilon),n,Ct\log(nt/\epsilon)}$ is a $\epsilon$-approximate unitary $t$-design for a sufficiently large constant $C>0$. To form a $2^{-\Omega(km_1^2\log^2 m_1\log m_0\log k)}$-approximate state $2k$-design, we only need circuit depth of scaling
\begin{align*}
O(\log(nt/\epsilon))=\tO(km_1^2\log m_0),
\end{align*}
where we use $t=2k$. We thus have the following corollary:
\begin{corollary}\label{coro:circuit_adaptivity_among_block}
For the hypothesis testing problem $\rhomm$ vs. $\rho\sim\mu_{\tO(km_1^2\log m_0),n,\tO(km_1^2\log m_0)}$ (i.e. an random circuit ensemble of depth $\tO(km_1^2\log m_0)$), weak detection is degree-$(\cD,k)$ hard for any adaptivity among round-based model protocol that makes polynomially many single-copy PVMs $\{\{U_i\ket{x}\bra{x}U_i^{\dagger}\}_{x=1}^{2^n}\}_{i=1}^m$. These protocols use $m_1=\polylog(n)$ blocks each with $m_0=m/m_1=\poly(n)$ samples, and adaptivity is only allowed among blocks. Moreover, this indicates that learning quantum circuit of depth $\tO(km_1^2\log m_0)$ is hard any adaptivity among round-based model protocol with a degree at most $k$ and $m=\poly(n)$ single-copy PVMs.
\end{corollary}

\subsubsection{Cryptographic applications}\label{sec:circuit_crypto}
In this part, we provide some discussion on cryptographic applications of the low-degree hardness proved in our work. A recent work from Fefferman, Ghosh, Sinha, and Yuen~\cite{fefferman2024hardness} showed that \emph{assumptions} about learning the output state of random quantum circuits could be used as the foundation for secure \emph{one-way state generators} in quantum cryptography. A one-way state generator is an efficient algorithm that takes as input a classical key $\mathsf{key}$, and outputs a quantum state $\ket{\psi_{\mathsf{key}}}$ from that key. Anyone with the key $\mathsf{key}$ can efficiently verify that $\ket{\psi_{\mathsf{key}}}$ is the correct output. However,  given polynomially many copies of the $\ket{\psi_{\mathsf{key}}}$, it should be computationally hard for an adversary to produce a key $\mathsf{key}'$ with the corresponding state $\ket{\psi_{\mathsf{key}'}}$ close to $\ket{\psi_{\mathsf{key}}}$. 

In their setting, they have to conjecture that random quantum circuits of depth $L\geq\log^2 n$ are computationally hard to learn in the average case as there is no solid proof for the hardness. Here, however, we are able to provide a rigorous hardness result in the low-degree model based on \Cref{coro:circuit_single_copy}:
\begin{observation}[Computational low-degree $\delta=2^{-\tTheta(k\log n)}$-no-learning observation]\label{obs:circuit_low_degree}
Let $\mathcal{C}_{n,\tTheta(k\log n)}$ be the class of coarse-grained random quantum circuit (see \Cref{sec:basic_circuit_moment} for definition). For all low-degree algorithms with degree at most $k$ and single-copy PVMs with up to $\Theta(n)$ ancillary qubits, given polynomially-many copies of $C\ket{0^n}$ for $C\in\mathcal{C}_{n,\tTheta(k\log n)}$ and $C\sim\mu_{\tTheta(k\log n),n\tTheta(k\log n)}$, the probability of the algorithm outputting a circuit $D$ such that $\abs{\braket{0^n|C^\dagger D|0^n}}\geq\delta$ is bounded by $\delta\leq 2^{-\tTheta(k\log n)}$.
\end{observation}

Based on this observation, we can used the following protocol to create a \emph{one-way state generators against low-degree adversaries}, which is an efficient algorithm that takes as input a classical key $\mathsf{key}$ and outputs a quantum state $\ket{\psi_{\mathsf{key}}}$ from that key such that
\begin{itemize}
    \item Anyone with the key $\mathsf{key}$ can efficiently verify that $\ket{\psi_{\mathsf{key}}}$ is the correct output.
    \item Given polynomially many copies of the $\ket{\psi_{\mathsf{key}}}$, it should be computationally hard for any adversary that can perform single-copy PVMs with up to $\Theta(n)$ ancillary qubits and low-degree post-processing with degree at most $k$ on the data to produce a key $\mathsf{key}'$ with the corresponding state $\ket{\psi_{\mathsf{key}'}}$ close to $\ket{\psi_{\mathsf{key}}}$.
\end{itemize}

\begin{algorithm}[ht]
    \textbf{Generation algorithm:} Given input key $\mathsf{key}\in\{0,1\}^{r(n)}$ for some polynomial $r(n)$, we interpret it as a description of an $n$-qubit circuit $C$ from the ensemble $\mathcal{C}_{n,\tTheta(k\log n)}$. Output $C\ket{0}$.\\
    \textbf{Verification procedure:} Given input key $\mathsf{key}\in\{0,1\}^{r(n)}$ and a state $D\ket{0}$, apply $C^\dagger$ to the state and measure in the computational basis. Accept if the result is all zeroes, and reject otherwise.
    \caption{Random circuit one-way state generators against low-degree adversaries}
    \label{fig:owsg}
\end{algorithm}

Moreover, by the fact that $k=\omega(\log(n))$ in the usual low-degree model, $\delta$ is thus inversely superpolynomially small. According to~\cite{fefferman2024hardness}, the one-way state generator against low-degree adversaries can be made noise robust against any noisy model via parallel repetition as long as the fidelity of implementing the circuit in the noisy case is at least $1/(\poly(n)L)$ where $L$ is the circuit depth.

\subsection{Low-degree hardness of learning sparse Hamiltonian from Gibbs states}\label{sec:hamiltonian_gibbs}

In this section, we introduce the second application of our framework, learning Hamiltonian from Gibbs states~\cite{anshu2020sample,anshu2021efficient,haah2022optimal,bakshi2024learning,wodecki2024learning,narayanan2024improved}. In this task, we are asked to estimate a Hamiltonian $H=\sum_{P_a\in\cP_n}\lambda_aP_a$ containing $J$ terms given access to the Gibbs state of the Hamiltonian under temperature parameter $\beta=1/kT$:
\begin{align*}
\rho_{\beta,H}=\frac{e^{\beta H}}{\tr(e^{\beta H})}.
\end{align*}
The algorithm has to output each $\lambda_a$ within additive error $\epsilon$ without prior knowledge of the structure of the Hamiltonian.

All the known algorithms for this task focus on the case when the Hamiltonian is $\varsigma$-local. In~\cite{anshu2020sample}, the first polynomial sample complexity (but computationally inefficient) bound for this task is given as $O\left(\frac{2^{\text{poly}(\beta)}J^2\log J}{\beta^{O(1)}\epsilon^2}\right)$. Ref.~\cite{anshu2021efficient} then shows that in the case when all $P_a$ are commuting the algorithm can be made computationally efficient. In the high-temperature regime where $\beta\leq\beta_c=(100e^62^{O(\varsigma)})^{-1}$, Ref.~\cite{haah2022optimal} provides an algorithm that achieves $O\left(\frac{2^{O(\beta)}\log J}{\beta^2\epsilon^2}\right)$ sample complexity and $O\left(\frac{2^{O(\beta)}J\log J}{\beta^2\epsilon^2}\right)$ time complexity. A recent result~\cite{bakshi2024learning} provided an algorithm that works for all $\beta$ with sample complexity $\text{poly}(J,1/\beta,(1/\epsilon)^{\max\{1,O(\beta)\}})$ and time complexity $\text{poly}(J,1/\beta,(1/\epsilon)^{2^{O(\beta)}})$. Follow-up works improved the double exponential scaling in $\beta$ to single exponential scaling~\cite{wodecki2024learning,narayanan2024improved}.

It turns out that in the regime where $\varsigma=O(1)$, $\epsilon=\Theta(1)$ and $\beta=\Theta(1)$, statistical and computational efficient algorithms are available. From the lower bound side, the best-known sample complexity lower bound is provided as $\Omega\left(\frac{e^{\beta}}{\beta^2\epsilon^2}\cdot\log n\right)$. It is natural to ask if we can computationally efficiently learn non-local (sparse) Hamiltonian (i.e., $J=\poly(n)$ but $\varsigma$ can be as large as $n$) from Gibbs state when $\beta=\Theta(1)$ and $\epsilon=\Theta(1)$.

In this section, we prove that learning even sparse but non-local Hamiltonian from the Gibbs state at $\beta=\Theta(1)$ is low-degree hard for non-adaptive single-copy measurements.

\subsubsection{Basic results in random matrix theory}\label{sec:basic_random_matrix}

We will also need some results from the random matrix theory. In this work, we consider \emph{Wigner's Gaussian unitary ensemble (GUE)}~\cite{wigner1958distribution}. A $d\times d$ GUE is a family of complex Hermitian random matrices specified by
\begin{align*}
H_{jj}&=\frac{g_{jj}}{\sqrt{d}},\\
H_{jk}&=\frac{g_{jk}+ig_{jk}'}{\sqrt{2d}},\ \text{for }k>j,
\end{align*}
where $g_{jj},g_{jk},g_{jk}'$ are independent standard Gaussian $\cN(0,1)$. The definition we consider here follows from~\cite{chen2024sparse} and has an additional $1/\sqrt{d}$ normalization factor from the standard definition of Gaussian unitary ensemble. In the following, we denote a random Hamiltonian chosen from GUE as $\hgue$. We can also write a $\hgue$ in the Pauli basis as:
\begin{align*}
\hgue=\sum_{P\in\cP_n}\frac{g_P}{d}P
\end{align*}
with each $g_P$ an independent standard Gaussian, and $d=2^n$. The GUE is invariant under any unitary transformation. Therefore, any subspace of low-energy eigenvectors will be Haar random. We will also need the following fact

\begin{fact}\label{fact:norm_GUE}
With probability at least $1-\exp(\Theta(n))$ for random Hamiltonian $\hgue$ from GUE, we have $\norm{\hgue}\leq 3$.
\end{fact}

\begin{proof}
We define normalized $q$-norms to be $\norm{B}_q=(\overline{\tr}\abs{B}^q)^{1/q}$ where $\overline{\tr}=\frac{1}{N}\cdot\tr$. The average-case normalized $q$-norms is defined as $|||B|||_q=(\E\overline{\tr}\abs{B}^q)^{1/q}$. By the random matrix theory~\cite{tropp2015introduction,tropp2018second,bandeira2023matrix,brailovskaya2024universality,anderson2010introduction}, we have 
\begin{align*}
|||\hgue|||_p\leq2\left(1+\frac{(p/2)^{3/4}}{\sqrt{d}}\right).
\end{align*}
By Markov's inequality, we have
\begin{align*}
\Pr[\norm{\hgue}\geq t]\leq\frac{\mathbb{E}_{\hgue}[\norm{\hgue}^p]}{t^p}\leq d\frac{|||\hgue|||_p}{t^p}\leq\left(\frac{2^{\log d/p}}{t}\cdot2\left(1+\frac{(p/2)^{3/4}}{\sqrt{d}}\right)\right)^p.
\end{align*}
Recall that $\log d=n$, we can choose $p=C\log n/\epsilon$ for some large enough $C$ and $t=2(1+\epsilon)$. Then we have
\begin{align*}
\Pr[\norm{\hgue}\geq2(1+\epsilon)]\leq\left(\frac{2^{\log d/p}}{t}\cdot2\left(1+\frac{(p/2)^{3/4}}{\sqrt{d}}\right)\right)^p\leq\left(
\frac{1+\epsilon/2}{1+\epsilon}\right)^p\leq\exp(-cn),
\end{align*}
for some constant $c$. By choosing $\epsilon=1/2$ we finish the proof.
\end{proof}

\subsubsection{Low-degree hardness for learning an arbitrary Hamiltonian}
Note that if we allow $J$ to be arbitrarily large, we will probably have to output an exponential number of $\lambda_a$ when $J$ is exponentially large. To avoid this output size barrier that prevents the possible computationally efficient algorithms, we consider the following distinguishing algorithm. In particular, the Hamiltonian is chosen either of the following two ensembles:
\begin{itemize}
    \item The Hamiltonian is a random Hamiltonian of form $\hgue=\sum_{P_a\in\cP_n}\frac{g_a}{d}P_a$ chosen from the Wigner's Gaussian unitary ensemble (GUE) (see \Cref{sec:basic_random_matrix}). Here, each $g_a$ is chosen as an IID standard Gaussian. $H$ can also be a random complex matrix with $H_{ij}=H_{ji}^*=\frac{g_{ij}+ig_{ij}}{\sqrt{2d}}$ and $H_{ii}=\frac{g_{ii}}{\sqrt{d}}$ where $g_{ii},g_{ij},g_{ij}'$ are independent standard Gaussian.
    \item The Hamiltonian is $H=0$, which yields a maximally mixed state as the Gibbs state.
\end{itemize}
The goal is to decide which is the case using copies of Gibbs state of the Hamiltonian. The task is again to distinguish between the random Gibbs state $\rho_{\beta,H}=\frac{e^{-\beta H}}{\tr(e^{-\beta H})}$, where $H$ is chosen from the GUE ensemble $\hgue$, from the maximally mixed state. Therefore, we can exploit our standard framework. In the following, we stick to the case when $\beta=1$ (low-temperature) and we denote such state as $\rho_H$.

We havethe following result showing the low-degree hardness of distinguishing between $\rhomm$ vs. $\rho_{\hgue}$, indicating the low-degree hardness of learning arbitrary Hamiltonian from its Gibbs state at $\beta=1$:
\begin{theorem}\label{thm:hgue_gibbs}
For the hypothesis testing problem $\rhomm$ vs. $\rho\sim\rho_{\hgue}$ and $k\leq O(n)$, weak detection is degree-$k$ hard for any protocol that makes polynomially many non-adaptive single-copy PVMs with up to $n'=\Theta(n)$ ancillary qubits. Moreover, this indicates that learning arbitrary Hamiltonian from its Gibbs state at $\beta=1$ is hard for any non-adaptive protocol with a degree at most $k\leq O(n)$ and $m=\poly(n)$ single-copy PVMs with up to $n'=\Theta(n)$ ancillary qubits.
\end{theorem}

\begin{proof}
According to the general framework for non-adaptive ancilla-assisted single-copy PVMs in \Cref{thm:general_ancilla}, we only need to consider the term
\begin{align*}
d_{\tr}\left(\mathbb{E}_{\hgue}\tr_{[mn]\backslash T}\left[\bigotimes_{i=1}^mU_i\rho_H U_i^\dagger\right],\frac{I_T}{2^{|T|}}\right)
\end{align*}
for any choice of unitary sequence $\{U_i\}_{i=1}^m$, and any $T\in[mn]$ and $\abs{T}=k\leq O(n)$. Recall from \Cref{fact:norm_GUE}, with probability $1-\exp(\Theta(n))$ for random Hamiltonian $\hgue$ from GUE, we have $\norm{\hgue}\leq 3$. As $\hgue$ is invariant from any unitary, we thus have the eigenstates of $\hgue$ chosen Haar randomly. Given any random $\hgue=\sum_{j=1}^{2^n}\lambda_j\ket{H_j}\bra{H_j}$ where $\lambda_i$ is the $i$-th eigenvalue and $\ket{H_j}$ is the corresponding eigenstate. With probability $1-\exp(-\Theta(n))$, we have $-3\leq\lambda_j\leq3$ by \Cref{fact:norm_GUE}, and each $\ket{H_j}$ chosen from Haar random ensemble. By \Cref{coro:general_ancilla}, as $\ket{H_j}$ is chosen from a unitary $2$-design ensemble (actually exact Haar randomly), for any subset $T\subset [mn]$ with $\abs{T}\leq n$, any $m$, and any fixed sequence of unitaries $\{U_i\}_{i=1}^m$, we have
\begin{align*}
d_{\tr}\left(\tr_{[mn]\backslash T}\left[\bigotimes_{i=1}^mU_i\ketbra{H_j}U_i^\dagger\right],\frac{I_T}{2^{\abs{T}}}\right)\leq m\cdot 2^{\abs{T}/2}\cdot (2^{\abs{T}-n})^{1/4}=m\cdot 2^{3\abs{T}/4-n/4}
\end{align*}
with probability $1-m\cdot (2^{\abs{T}-n})^{1/2}=1-m\cdot 2^{\abs{T}/2-n/2}$ over the choice of $\ket{H_j}$.

We then consider the Gibbs state $\rho_{\hgue}=\frac{e^{-\hgue}}{\tr(e^{-\hgue})}=\frac{1}{\tr(e^{-\hgue})}\sum_{i=1}^{2^n}e^{-\lambda_j}\ket{H_j}\bra{H_j}$. We focus on the case when $\norm{\hgue}\leq 3$. In this case, we always have $e^{-\lambda_j}=\Theta(2^{-n})$. We thus have
\begin{align*}
d_{\tr}\left(\mathbb{E}_{\hgue}\tr_{[mn]\backslash T}\left[\bigotimes_{i=1}^mU_i\rho_H U_i^\dagger\right],\frac{I_T}{2^{|T|}}\right)\leq O\left(m\cdot 2^{3\abs{T}/4-n/4}\right).
\end{align*}

We then consider PVMs with $n'$ ancillary qubits as in \Cref{sec:general_ancilla}. By \Cref{thm:general_ancilla}, we need to compute the following term:
\begin{align*}
d_{\tr}\left(\mathbb{E}_{\rho_{\hgue}}\underset{[m(n+n')]\backslash T}{\tr}\left(\bigotimes_{i=1}^{m}U_i\left(\rho_{\hgue}\otimes\ketbra{0}^{\otimes n'}\right) U_i^\dagger\right),\underset{[m(n+n')]\backslash T}{\tr}\left(\bigotimes_{i=1}^{m}U_i\left(\rhomm\otimes\ketbra{0}^{\otimes n'}\right) U_i^\dagger\right)\right).
\end{align*}
Recall that $\rho_{\hgue}=\frac{e^{-\hgue}}{\tr(e^{-\hgue})}=\frac{1}{\tr(e^{-\hgue})}\sum_{i=1}^{2^n}e^{-\lambda_j}\ket{H_j}\bra{H_j}$ has each $\ket{H_j}$ chosen from state $2$-design (actually true Haar random) and the coefficient is of the scaling $O(2^{-n})$, we can use the triangle inequality with \eqref{eq:trace_dis_ancilla} to obtain
\begin{align*}
&d_{\tr}\left(\mathbb{E}_{\rho_{\hgue}}\underset{[m(n+n')]\backslash T}{\tr}\left(\bigotimes_{i=1}^{m}U_i\left(\rho_{\hgue}\otimes\ketbra{0}^{\otimes n'}\right) U_i^\dagger\right),\underset{[m(n+n')]\backslash T}{\tr}\left(\bigotimes_{i=1}^{m}U_i\left(\rhomm\otimes\ketbra{0}^{\otimes n'}\right) U_i^\dagger\right)\right)\\
&\leq 2m\cdot 2^{\abs{T}/2}\left(2\cdot2^{\abs{T}+11n'-n}+3\epsilon\right)^{1/4}.\qedhere
\end{align*}
\end{proof}

\subsubsection{Low-degree hardness for learning a sparse Hamiltonian}

In the above section, we consider learning an arbitrary Hamiltonian from its Gibbs state at temperature $\beta=1$. Now, we show that even for sparse Hamiltonian, which has $J=\poly(n)$ Pauli terms, is low-degree hard to learn from its Gibbs states at temperature $\beta=1$.

We consider the following distinguishing task between two ensembles:
\begin{itemize}
    \item The Hamiltonian is a random sparse Pauli string (RSPS) Hamiltonian ensemble $H_{\text{RSPS}}$ where each Hamiltonian is an independent sum of $J$ random Pauli strings with random sign coefficients~\cite{chen2024sparse}:
    \begin{align*}
    H=\sum_{a=1}^J\frac{r_a}{\sqrt{J}}P_a,
    \end{align*}
    where $P_a$ is IID chosen from $\cP_n$ and $r_a$ is chosen IID uniformly in $\{+1,-1\}$.
    \item The Hamiltonian is $H=0$, which yields a maximally mixed state as the Gibbs state.
\end{itemize}
The goal is to decide which is the case using copies of the Gibbs state of the Hamiltonian at $\beta=1$. For this distinguishing task, we have the following hardness result.
\begin{corollary}\label{coro:hgue_gibbs_sparse}
For the hypothesis testing problem $\rhomm$ vs. $\rho\sim\rho_{H_{\text{RSPS}}}$ at $J=O(n^3)$ and any $k\leq O(n)$, weak detection is degree-$k$ hard any protocol that makes polynomially many non-adaptive single-copy PVMs with up to $n'=\Theta(n)$ ancillary qubits. Moreover, this indicates that learning a sparse Hamiltonian (containing polynomial terms in the Pauli basis) from its Gibbs state at $\beta=1$ is hard for any non-adaptive protocol with a degree at most $k\leq O(n)$ and $m=\poly(n)$ single-copy PVMs with up to $n'=\Theta(n)$ ancillary qubits.
\end{corollary}

\begin{proof}
The first observation is that the RSPS ensemble matches the GUE to the third moment, i.e.,
\begin{align*}
&\E_{\hgue}[\hgue]=\E_{H_{\text{RSPS}}}[H_{\text{RSPS}}]=0,\\
&\E_{\hgue}[\hgue\otimes\hgue]=\E_{H_{\text{RSPS}}}[H_{\text{RSPS}}\otimes H_{\text{RSPS}}]=\frac{1}{d^2}\sum_{P_a\in\cP_n}P_a\otimes P_a,\\
&\E_{\hgue}[\hgue\otimes\hgue\otimes\hgue]=\E_{H_{\text{RSPS}}}[H_{\text{RSPS}}\otimes H_{\text{RSPS}}\otimes H_{\text{RSPS}}]=0.
\end{align*}
It is invariant under random Clifford unitary, which forms a unitary $3$-design~\cite{zhu2017multiqubit}. We can thus conclude that the eigenstates of $H_{\text{RSPS}}$ form a unitary $3$-design. 

Ref.~\cite{chen2024sparse} shows that when $J=\poly(n)$, the behavior of $H_{\text{RSPS}}$ is close to $\hgue$. In particular, by Theorem III.1 of~\cite{chen2024sparse}, when $J\geq O(n^3/\epsilon^4)$ ensures that 
\begin{align}\label{eq:RSPS_eig}
\Pr_{H\sim H_{\text{RSPS}}}[\norm{H}\geq 2(1+\epsilon)]\leq\exp(-\Theta(n)).
\end{align}

As the eigenstates of $H_{\text{RSPS}}$ form the state $3$-design and with $1-\exp(-\Theta(n))$ the eigenvalue of $H_{\text{RSPS}}$ is bounded within $[-3,3]$. Note that for learning arbitrary Hamiltonian in \Cref{thm:hgue_gibbs}, we only use the fact that the eigenvalue is bounded with high probability and the eigenstates are chosen from state $2$-design. We thus again have 
\begin{align*}
&d_{\tr}\left(\mathbb{E}_{\rho_{\text{RSPS}}}\underset{[m(n+n')]\backslash T}{\tr}\left(\bigotimes_{i=1}^{m}U_i\left(\rho_{\text{RSPS}}\otimes\ketbra{0}^{\otimes n'}\right) U_i^\dagger\right),\underset{[m(n+n')]\backslash T}{\tr}\left(\bigotimes_{i=1}^{m}U_i\left(\rhomm\otimes\ketbra{0}^{\otimes n'}\right) U_i^\dagger\right)\right)\\
&\leq 2m\cdot 2^{\abs{T}/2}\left(2\cdot2^{\abs{T}+11n'-n}+3\epsilon\right)^{1/4}.
\end{align*}
We can then obtain the claimed result via \Cref{thm:general_ancilla}.
\end{proof}

\section{Quantum planted biclique}\label{sec:conn_planted_clique}
Here, we consider a new quantum learning problem that can be understood as the quantum analogue of the planted biclique problem. In this problem, we will focus on $n$-\emph{qudit} quantum states with each qudit of $d$ dimensions. We will abuse the notation by using $d$ to represent the local qudit dimension instead of the dimension of a single copy of the quantum state only throughout this subsection. 

Before formulating the task, we will have to define a few related notations. Given operators $A,B\in\C^{d\times d}$ and a subset $T\subseteq[n]$, define 
\begin{align*}
A \otimes_T B = \bigotimes^n_{i=1} (A\cdot \bone[i\not\in T] + B\cdot \bone[i\in T]).
\end{align*}
We suppose $\lambda,d,m,n\in\mathbb{N}$, and without loss of generality, $d$ is prime. Given $S\subseteq[n]$, let $\Lambda_S$ denote the $n$-qudit channel that acts as identity on the qudits outside of $S$ and as fully depolarizing noise on the qudits within $S$. Let $\Lambda_\kappa$ denote the depolarizing channel on $(\mathbb{C}^d)^{\otimes n}$ with rate $\kappa$. We will take $\kappa = \lambda/n$ throughout, so we will omit the subscript. 

We are now given either $m$ copies of the following two cases:
\begin{itemize}
    \item $\sigma = I/d^n$.
    \item $\sigma = \sigma_{\rho,S}$ which is defined as follows. Nature samples Haar-random, $d$-dimensional pure state $\rho$, and a subset $S\subseteq[n]$ of expected size $\lambda$ by including each element of $[n]$ independently with probability $\lambda$ ($\rho$ and $S$ are the same across all $m$ copies of $\sigma_{\rho, S}$). The state $\sigma_{\rho, S}$ is given by
    \begin{align*}
    \sigma_{\rho, S} \triangleq (\Lambda_\kappa\circ \Lambda_S)[\rho^{\otimes n}].
    \end{align*}
\end{itemize}
The goal is to distinguish between these alternatives.

\subsection{Lower bound for local measurements} 
\label{sec:lower_pds} 
We first prove the hardness for protocols of bounded degree making local measurements, i.e., single-qudit measurements. Here, we suppose that the learner performs (non-adaptive) local measurements. In particular, suppose the $j$-th qudit of the $i$-th copy of $\sigma_{\rho, S}$ is measured using the basis vectors $\{\ket{\psi^{i,j}_1},\ldots,\ket{\psi^{i,j}_d}\}$. Given $\bx = (x_{i,j})\in [d]^{m\times n}$, define $\ket{\psi_\bx} \triangleq \otimes_{i,j} \ket{\psi^{i,j}_{x_{i,j}}}$.

\begin{theorem}[Low-degree hardness for local measurements]\label{thm:pds_lbd}
Weak detection for quantum planted biclique is degree-$k$ hard for any protocol that makes $m$ non-adaptive local measurements at $k=n^{o(1)}$ and $\lambda\leq o(nd^{1/4}m^{-1/2})$. 
\end{theorem}

\begin{proof}
Note that if $\sigma = \Id/d^n$, then the distribution over outcomes after performing these measurements is uniform over $[d]^{m\times n}$. On the other hand, if $\sigma = \sigma_{\rho, S}$, the probability of observing $\bx \in [d]^{m\times n}$ is given by $\bra{\psi_\bx}\E_{\rho,S}\sigma_{\rho, S}^{\otimes m}\ket{\psi_\bx}$. For fixed $\rho, S$, denote the likelihood ratio by 
\begin{align*}
    \overline{D}_{\rho,S}(\bx)\triangleq d^{mn}\bra{\psi_\bx}\sigma_{\rho, S}^{\otimes m}\ket{\psi_\bx},
\end{align*}
and define $\overline{D} \triangleq \E_{\rho, S}\overline{D}_{\rho, S}$.

\begin{lemma}\label{lem:expectS}
    Given $t\in\mathbb{N}$ and density matrix $\rho \in \mathbb{C}^{d\times d}$,
    \begin{align*}
        \E_S \sigma_{\rho, S}^{\otimes t} = \sum_{U\subseteq[t]} \kappa^{|U|} (I/d^n)^{\otimes U^c} \otimes \Bigl(\sum_{T_1,\ldots,T_{|U|}\neq \emptyset} \bigotimes_{\ell \in U} (I/d\otimes_{T_\ell}\Delta) \cdot \kappa^{|\cup_\ell T_\ell|}\Bigr)\,,
    \end{align*}
    where $\Delta \triangleq \rho - \Id/d$.
\end{lemma}

\begin{proof}
    Note that we can write
    \begin{align*}
        \sigma_{\rho, S} &= \kappa \Bigl(\bigotimes^n_{i=1} (\Id/d + \mathbb{1}[i\in S]\Delta)\Bigr) + (1 - \kappa) \Id/d^n \\
        &= \kappa\sum_{T\subseteq[n]} (\Id/d \otimes_T \Delta)\cdot \mathbb{1}[T\subseteq S] + (1 - \kappa)\Id/d^n \\
        &= \Id/d^n + \kappa\sum_{T\neq\emptyset} (\Id/d\otimes_T \Delta)\cdot \mathbb{1}[T\subseteq S]\,,
    \end{align*}
    so
    \begin{align*}
        \sigma^{\otimes t}_{\rho, S} = \sum_{U\subseteq[t]} \kappa^{|U|} (I/d^n)^{\otimes U^c} \otimes \Bigl(\sum_{T_1,\ldots,T_{|U|}\neq \emptyset} \bigotimes_{\ell \in U} (I/d\otimes_{T_\ell}\Delta) \cdot \mathbb{1}[T_1\cup\cdots\cup T_\ell\subseteq S]\Bigr)\,.
    \end{align*}
    The claim follows.
\end{proof}
\noindent We will consider the standard Fourier basis over $[d]^{m\times n}$. Let $\xi \triangleq e^{\frac{2\pi \i}{d}}$. Given a partial assignment $\mathbf{W} = (W,\alpha)$, where $W\subseteq [m]\times[n]$ and $\alpha: W\to\{1,\ldots,d-1\}$, define the basis function
\begin{align*}
    \chi_{\bW}(\bx) \triangleq \prod_{(i,j)\in W} \xi^{\alpha_{i,j}x_{i,j}}\,,
\end{align*}
and define the associated Fourier coefficient of a function $f: [d]^{m\times n} \to \mathbb{C}$ by $\hat{f}[\mathbf{W}] = \E[f(\bx)\chi_W(\bx)]$, where $\bx$ is uniform over $[d]^{m\times n}$. To lighten notation, we will denote the Fourier transform of $\overline{D}$ by $\hat{D}$ instead of $\widehat{\overline{D}}$.

Let $\mathbf{W} = (W,\alpha)$ where $W = \sqcup^m_{i=1} W_i$ for $W_i\subseteq \{i\}\times[n]$ and $\alpha_i: W_i \to \{1,\ldots,d-1\}$ is the restriction of $\alpha$ to $W_i$; define $\mathbf{W}_i = (W_i, \alpha_i)$. Given $\bx\in[d]^{m\times n}$, denote by $\bx|_{W_i} \in [d]^n$ the restriction of $\bx$ to the coordinates indexed by $W_i$, and define $\ket{\psi_{x|_{W_i}}} \in (\mathbb{C}^d)^{\otimes |W_i|}$ in the natural way.

Then we have
\begin{align}
\widehat{\overline{D}_{\rho, S}}[\mathbf{W}] = d^{|W|} \prod_{(i,j)\in W} \E_{x_{i,j}}\bigl[\xi^{\alpha_{i,j}x_{i,j}} \bra{\psi^{i,j}_{x_{i,j}}}\sigma_{\rho, S}\ket{\psi^{i,j}_{x_{i,j}}}\bigr]\,. \label{eq:fourier_coeff} 
\end{align}
We can now compute the expectation of this over $\rho, S$ by applying Lemma~\ref{lem:expectS}. Let $\supp(\bW) \triangleq \{i\in [m]: W_i \neq \emptyset\}$. Note that because $\E_x[\xi^{\alpha x}] = 0$ for $\alpha \in \{1,\ldots,d-1\}$, the only terms in the expression for $\E \sigma^{\otimes W}_{\rho, S}$ which contribute to Eq.~\eqref{eq:fourier_coeff} are given by $U = \supp(\bW)$ and $T_i = W_i$ for all $i\in \supp(\bW)$. We thus obtain
\begin{align*}
    \hat{D}[\mathbf{W}] = \kappa^{|\supp(\bW)| + |\cup_\ell W_\ell|}\, \E_{\Delta} \Bigl[\prod_{(i,j) \in W}\E_{x_{i,j}}\bigl[\xi^{\alpha_{i,j}x_{i,j}} \bra{\psi^{i,j}_{x_{i,j}}}d\Delta\ket{\psi^{i,j}_{x_{i,j}}}\bigr]\Bigr]\,.
\end{align*}
For any $W\subseteq[m]\times [n]$, we have that the Fourier mass
\begin{align*}
    \mu_W \triangleq \sum_{\alpha: W\to \{1,\ldots,d-1\}} |\hat{D}[(W,\alpha)]|^2
\end{align*}
can be expressed as
\begin{align}
    \mu_W = \kappa^{2|\supp(\bW)| + 2|\cup_\ell W_\ell|} \,\E_{\Delta,\Delta'}\prod_{(i,j) \in W} C_{i,j}(\Delta,\Delta')\,,\label{eq:sumalpha}
\end{align}
where
\begin{align}
    C_{i,j}(\Delta,\Delta') \triangleq \E_{x,x' \sim [d]}\Bigl[\sum_{a \in \{1,\ldots,d-1\}}\xi^{a(x-x')} \bra{\psi^{i,j}_{x}}d\Delta\ket{\psi^{i,j}_x}\bra{\psi^{i,j}_{x'}}d\Delta'\ket{\psi^{i,j}_{x'}}\Bigr]\,. \label{eq:expectDelta}
\end{align}
Note that
\begin{align*}
    \sum_{a \in \{1,\ldots,d-1\}} \xi^{a(x-x')} = d\cdot\bone[x = x'] - 1\,,
\end{align*}
so
\begin{align*}
    C_{i,j}(\Delta,\Delta') = \E_x [\bra{\psi^{i,j}_x} d\Delta\ket{\psi^{i,j}_x}\cdot \bra{\psi^{i,j}_x} d\Delta'\ket{\psi^{i,j}_x}]
\end{align*}

To compute \eqref{eq:expectDelta}, we will appeal to the following moment calculation:

\begin{lemma}\label{lem:haar}
    For $\Delta = \rho - \Id/d$ with $\rho$ Haar-random,
    \begin{align*}
        \E (d\Delta)^{\otimes t} = \sum_{\pi\in \mathcal{S}_t} \gamma_{|\mathrm{fix}(\pi)|} P_\pi\,,
    \end{align*}
    where $\mathrm{fix}(\pi) \subseteq[t]$ denotes the set of fixed points of $\pi$, and $P_\pi$ denotes the $n$-qudit permutation operator which acts via $\pi$ on the qudits indexed by $U$ and as identity elsewhere. Here, $\gamma_0,\ldots,\gamma_t$ satisfy
    \begin{align*}
        |\gamma_0| = \frac{d^t}{d(d+1)\cdots (d + t-1)} \qquad \text{and} \qquad |\gamma_f| \le O(t/\sqrt{d})^f \ \ \text{for} \ \ f > 0
    \end{align*}
    if $t = o(\sqrt{d})$.
\end{lemma}

\begin{proof}
    By a standard calculation,
    \begin{align}
        \E (d\Delta)^{\otimes t} \triangleq \sum_{U\subseteq[t]} (-1)^{t - |U|} \beta_{d,|U|} \sum_{\pi \in \mathcal{S}_U} P_\pi  \qquad \text{for}  \qquad \beta_{d,s}\triangleq \frac{d^s}{d(d+1)\cdots (d + s - 1)}\,, \label{eq:betadef}
    \end{align}
    where $\mathcal{S}_U$ denotes the symmetric group on the elements in $U$.

    For any $n$-qudit permutation operator $P_\pi$, if $V = \mathrm{fix}(\pi)$ has size $f$, then its coefficient in the above sum is given by
    \begin{align*}
        \sum_{U\supseteq V^c} (-1)^{t-|U|}\beta_{d,U} = \sum^f_{r = 0} (-1)^r \binom{f}{r} \beta_{d,t-r} \triangleq \gamma_f\,.
    \end{align*}
    The claimed expression for $\gamma_0$ is immediate. 
    
    It remains to establish the bound on $|\gamma_f|$ for $f > 0$. By Fact~\ref{fact:powerseries}, $\beta_{d,t-r}$ admits the series expansion $1 + \sum^\infty_{a=1} c_a(t-r) \cdot d^{-a}$, where each $c_a$ is a degree-$2a$ polynomial in $t - r$ (and thus in $r$) satisfying $|c_a(t-r)| \le (1 + t - r)^{2a} \le (2t)^{2a}$. Note that for each $a \le f/2$, we have
    \begin{align*}
        \sum^f_{r=0} (-1)^r\binom{f}{r} c_a(t-r) = 0\,.
    \end{align*}
    We conclude that
    \begin{align*}
        |\gamma_f| = \Bigl|\sum^f_{r=0} (-1)^r \binom{f}{r} \sum^\infty_{a = \lfloor f/2\rfloor + 1} c_a(t-r) d^{-a}\Bigr| \le O(t/\sqrt{d})^f
    \end{align*}
    as claimed.
\end{proof}
\noindent Therefore,
\begin{align}
    \E_{\Delta,\Delta'} \prod_{(i,j)\in W} C_{i,j}(\Delta,\Delta')
    &= \sum_{\pi,\pi'\in \mathcal{S}_W} \gamma_{|\mathrm{fix}(\pi)|} \gamma_{|\mathrm{fix}(\pi')|} \E_{\bx}[\bra{\psi_{\bx|_W}}P_\pi\ket{\psi_{\bx|_W}} \cdot \bra{\psi_{\bx|_W}}P_{\pi'}\ket{\psi_{\bx|_W}}] \nonumber\\
    &\le \Bigl(\sum_{\pi \in \mathcal{S}_W} \gamma_{|\mathrm{fix}(\pi)|} \E_{\bx}[|\bra{\psi_{\bx|_W}}P_\pi\ket{\psi_{\bx|_W}}|^2]^{1/2}\Bigr)^2 \nonumber\\
    &\le |W|!^2 \cdot O(1/d)^{|W|/2}\,, \label{eq:estimateEDeltas}
\end{align}
where in the third step we used Lemma~\ref{lem:haar} and Lemma~\ref{lem:derange} below to conclude that for any permutation $\pi \in \mathcal{S}_W$,
\begin{align*}
    \gamma_{\mathrm{fix}(\pi)}\E_{\bx}[|\bra{\psi_{\bx|_{W\backslash \mathrm{fix}(\pi)}}}P_\pi\ket{\psi_{\bx|_{W\backslash \mathrm{fix}(\pi)}}}|^2]^{1/2}& \le O(|W|/\sqrt{d})^{|\mathrm{fix}(\pi)|/2} \cdot d^{-|W\backslash\mathrm{fix}(\pi)|/4} \\
    &\le O(1/d)^{|W|/4}\\
    &\le O(|W|^4/d)^{|W|/2}\,.
\end{align*}

Finally, we can substitute the bound in Eq.~\eqref{eq:estimateEDeltas} back into Eq.~\eqref{eq:sumalpha} to obtain
\begin{equation*}
    \mu_W \le \kappa^{2|\supp(\bW)| + 2|\cup_\ell W_\ell|}\cdot  O(|W|^2/\sqrt{d})^{|W|}\,.
\end{equation*}

Therefore, for any fixed $Z\subseteq[n]$ and $t$, the sum over $\mu_W$ for which $|\mathrm{supp}(W)| = t$, $\cup_\ell W_\ell = Z$, and $|W| \le k$ is upper bounded by
\begin{equation*}
    \binom{m}{t} \kappa^{2t + 2|Z|} \sum_{1 \le |W_1|,\ldots,|W_t| \le k: |\cup_\ell W_\ell| = Z}O(k^2/\sqrt{d})^{\sum_\ell |W_\ell|}\,.
\end{equation*}
Note that the sum in the above equation can be upper bounded by
\begin{align*}
    \sum_{1 \le |W_1|,\ldots,|W_t| \le k: W_\ell\subseteq Z \forall \ell} O(k^2/\sqrt{d})^{\sum_\ell |W_\ell|} &\le \bigl((1+O(k^2/\sqrt{d}))^{|Z|}-1\bigr)^t \\
    &\le O(k^3/\sqrt{d})^t\,.
\end{align*}

Summing over all nonempty $Z\subseteq[n]$ of size at most $k$, we conclude that
\begin{align*}
    \sum_{W: 1\le |W|\le k} \mu_W &\le \sum^k_{t=1} \sum^{k}_{s=1} \binom{m}{t}\binom{n}{s} \kappa^{2s + 2t} O(k^3/\sqrt{d})^t \\
    &\le \sum^\Deg_{t=1} O\Bigl(\frac{k^3 m\lambda^2}{n^2\sqrt{d}}\Bigr)^t \cdot \sum^{k}_{s=1}(\lambda^2/n)^s\,.
\end{align*}
In particular, if $k = n^{o(1)}$, this is bounded by $o(1)$ provided
\begin{align*}
    \lambda&\le \tilde{o}(nd^{1/4}m^{-1/2})\,.\qedhere
\end{align*}
\end{proof}

\subsection{Upper bound for single-copy measurements}
Here, we demonstrate that for the detection problem $\{\sigma_{\rho,S}\}_{\rho\sim\muh,S}$ vs. $I/d^n$, there is a low-degree algorithm using (non-local) single-copy measurements that succeeds provided $\lambda\geq \omega(n^{2/3})$. 

\begin{theorem}[A low-degree efficient protocol with single-copy PVMs] \label{thm:single_copy_pds}
There is a constant-degree protocol for weak detection for quantum planted biclique that uses polynomially many non-adaptive single-copy PVMs, provided $\lambda\geq \omega(n^{2/3})$.
\end{theorem}

\begin{proof}
Let $\Pi \triangleq \frac{\Id + \SWAP}{2} \in \mathbb{C}^{d^2 \times d^2}$. Suppose that $n$ is even, and for every copy of the unknown $n$-qudit state, consider measuring the $2j$-th and $(2j + 1)$-th qudits with the PVM $\{\Pi, \Id - \Pi\}$. Given $i\in [n], j\in[n/2]$, let $Z_{ij}(\sigma)$ denote the indicator random variable for whether the outcome of that measurement is $\Pi$ when the $i$-th copy of the underlying state being measured is $\sigma$. For any $i,j$, we have
\begin{align*}
    \E[Z_{ij}(\Id/d^n)] = \frac{d+1}{2d}\,,
\end{align*}
whereas for fixed $S$,
\begin{align*}
    \E[Z_{ij}(\sigma_{\rho, S})] &= \kappa\cdot \Bigl\{\frac{d + 1}{2d}\cdot (1 - \bone_j) + \bone_j\Bigr\} + (1 - \kappa)\cdot \frac{d + 1}{2d} \\
    &= \frac{d+1}{2d} + \frac{\kappa(d-1)}{2d}\cdot \bone_j \triangleq \alpha_{S,j} \,,
\end{align*}
where $\bone_j$ denotes the indicator for the event that $(2j,2j+1)\in S\times S$.
The expected value of the latter over the randomness of $S$ is
\begin{align*}
    \E_S[\alpha_{S,j}] = \frac{d+1}{2d} + \kappa^3\cdot \frac{d-1}{2d} \triangleq \overline{\alpha}\,.
\end{align*}

Let us now consider the test statistic $Z = \frac{2}{mn}\sum_{ij} Z_{ij}$.

As each $Z_{ij}(\Id/d^n)$ is an independent sample from $\mathrm{Ber}(\frac{d+1}{2d})$, we have
\begin{align*}
    \Var[Z(\Id/d^n)] = \frac{2}{mn}\cdot \frac{d^2-1}{4d^2} \asymp \frac{1}{mn}\,.
\end{align*}
We can also compute
\begin{align*}
    \E_S \E[Z(\sigma_{\rho, S})^2] &= \E_S \Bigl[\frac{4}{m^2n^2}\Bigl(\sum_{i,j} \alpha_{S,j} + \sum_{i\neq i'; j} \alpha^2_{S,j}  + \sum_{i,i'; j\neq j'} \alpha_{S,j} \alpha_{S,j'}\Bigr)\Bigr] \\
    &= \frac{4}{m^2n^2} \Bigl(\frac{mn}{2}\overline{\alpha} + \frac{m(m-1)n}{2} \Bigl(\kappa^2\overline{\alpha}^2 + (1-\kappa^2)\Bigl(\frac{d+1}{2d}\Bigr)^2\Bigr)  + \frac{m^2n(n-2)}{4}\overline{\alpha}^2\Bigr) \\
    &= \overline{\alpha}^2 + \frac{2\overline{\alpha}}{mn} + \frac{4}{mn} \Bigl(\frac{m-1}{2} \Bigl(\kappa^2\overline{\alpha}^2 + (1-\kappa^2)\Bigl(\frac{d+1}{2d}\Bigr)^2\Bigr)  - \frac{m}{2}\overline{\alpha}^2\Bigr) \\
    &= \overline{\alpha}^2 + \frac{2}{mn}(\overline{\alpha} - \Bigl(\frac{d+1}{2d}\Bigr)^2) + \frac{4}{mn}\Bigl(\frac{m-1}{2}\kappa^2 - \frac{m}{2}\Bigr)(\overline{\alpha}^2 - \Bigl(\frac{d+1}{2d}\Bigr)^2) \\
    &= \overline{\alpha}^2 + \Theta(1/mn + \kappa^3/n)\,,
\end{align*}
so the variance of the test statistic $Z(\sigma_{\rho,S})$ is bounded by $\Theta(1/mn + \kappa^3/n)$. In particular, detection is possible provided that 
\begin{align*}
    \kappa^6 \gg 1/mn + \kappa^3/n\,,
\end{align*}
which holds provided $\lambda \gg n^{2/3}$, as claimed.
\end{proof}

\subsection{Upper bounds for local measurements}

Finally, we show that with computational basis measurements, there is a computationally inefficient algorithm that succeeds provided $\lambda \gg \min(d\log n, n^{1/2}d^{1/4})$ and a low-degree algorithm that succeeds provided $\lambda \gg n^{1/2}d^{1/4}$. In particular, the latter threshold matches the threshold below which we showed low-degree hardness in Section~\ref{sec:lower_pds}.

\begin{proposition}[Information-theoretic and computational tractability with local measurements]\label{lem:local_upper_pds}
    Weak detection for quantum planted biclique is information-theoretically possible using measurements in the computational basis provided $\lambda \gg \min(d\log n, n^{1/2} d^{1/4})$. There is a constant-degree protocol for weak detection using measurements in the computational basis provided $\lambda \gg n^{1/2} d^{1/4}$.
\end{proposition}

\begin{proof}
    Consider the observable value $\Tr(Z\rho)$ for $Z = \mathrm{diag}(1,\ldots,-1,\ldots)$. By standard anti-concentration of the Haar measure in $d$ dimensions, with high probability $|\Tr(Z\rho)| = \Omega(1/\sqrt{d})$. In other words, if for every qudit of every copy, we measure in the $d$-dimensional computational basis and look at whether the measurement outcome is among the first $d / 2$ computational basis vectors, this is distributed as $\mathrm{Ber}(1/2 + \Theta(1/\sqrt{d}))$ if the qudit being measured is a copy of $\rho$, and as $\mathrm{Ber}(1/2 + \Theta(1/\sqrt{d}))$ if the qudit being measured is a copy of $\Id/d$. In particular, with large constant probability, post-processing the readouts to solve quantum planted biclique now amounts to solving an instance of classical \emph{bipartite planted dense subgraph}~\cite{feldman2017statistical}. In particular, this entails distinguishing between a completely random $n\times n$ bipartite graph (with edge probabilities equal to $1/2$) versus one in which a random subgraph of expected size $\lambda\times \lambda$ has been planted whose edge connection probabilities are elevated to $1/2 + p$ for $p = \Theta(1/\sqrt{d})$.

    First suppose that $\lambda \gg d\log n$. Note that for a random $t\times t'$ subgraph with edge probabilities equal to $1/2$, the number of edges is distributed as $\mathrm{Bin}(tt', 1/2)$ and thus lies within $tt'/2 \pm O(\sqrt{tt'}\sqrt{\log 1/\delta})$ with probability $1 - \delta$. By a union bound over all $t\times t'$ subgraphs for $t, t' \le O(\lambda)$ of a completely random $n\times n$ bipartite graph with edge probabilities equal to $1/2$, we conclude that under the null hypothesis, with high probability no such $t\times t'$ subgraph has more than $tt'/2 + O(t^{3/4}t'^{3/4}\sqrt{\log n})$ edges. On the other hand, in the planted scenario, the number of edges in the planted $t\times t'$ subgraph, where $\E[t]  = \E[t'] = \lambda$, is distributed as $\mathrm{Bin}(tt', 1/2 + p)$. Provided that $tt' p \gg t^{3/2}t'^{3/4}\sqrt{\log n}$, which happens with high probability over $t,t'$ because $\lambda \gg d\log n$, one can thus perform weak detection for quantum planted biclique by brute force enumerating over all subsets of $t$ copies of the state, and all subsets of $t$ qudits out of $n$ qudits.

    Next suppose that $\lambda \gg n^{1/2}d^{1/4}$. The total number of edges present in the full graph under the null hypothesis is distributed as $\mathrm{Bin}(n^2, 1/2)$. Under the alternative hypothesis where there is a planted $t\times t'$ subgraph with elevated edge probabilities $1/2 + p$ for $p = \Theta(1/\sqrt{d})$, the total number of edges is instead distributed as $\mathrm{Bin}(n^2 - tt', 1/2) + \mathrm{Bin}(tt', 1/2 + p)$. The expected difference, over the randomness of $t,t'$, between the means of these random variables is $\lambda^2 p$, whereas both random variables have variance $O(n^2)$. So provided that $\lambda \gg \sqrt{n/p} \asymp n^{1/2}d^{1/4}$, there is a constant-degree algorithm for performing weak detection by simply counting the number of edges present in the bipartite graph given by the readouts.
\end{proof}

\section{Low-degree hardness for hybrid and classical problems.}\label{sec:hybrid_classical}

In this section, we provide some extended applications of the quantum low-degree model in different problems with classical or hybrid input by encoding classical input into quantum data. In particular, we embed a classical problem that is hard for classical low-degree algorithms into quantum data in a low-degree way to prove the hardness of quantum tasks under our quantum low-degree model. We first use a reduction from classical low-degree hardness of Learning Parity with Noise to show the low-degree hardness for Learning Stabilizer with Noise~\cite{poremba2024learning} in \Cref{sec:other_lsn_avg}. We also show the low-degree hardness of agnostic learning product states by embedding tensor PCA, which is shown to be classically low-degree hard, into quantum data in \Cref{sec:conn_agnostic_product}. 

\subsection{Low-degree hardness of Learning Stabilizers with Noise}\label{sec:other_lsn_avg}
We now consider the task of learning stabilizer with noise, which is another candidate for a quantum analog of learning parities with noise~\cite{poremba2024learning} other than learning stabilizer states in \Cref{sec:stabilizer_state}. The \emph{Learning Parity with Noise ($\mathsf{LPN}$)} problem and its assumed hardness, are cornerstones on which much of modern classical and quantum cryptography is built~\cite{blum1994cryptographic}. This is the problem of decoding a random linear code
under Bernoulli noise. Formally, the input to the problem is
\begin{align*}
(A\sim\{0,1\}^{n\times l},Ax+e\ (\text{mod }2)),
\end{align*}
where $A\sim\{0,1\}^{n\times l}$ and $x\in\{0,1\}^l$ are chosen uniformly at random, and $e=\text{Ber}_\kappa^{\otimes l}$ is a random Bernoulli error for some appropriate noise rate $\kappa$. The goal is to recover $x$. For $n=\poly(l)$, $\mathsf{LPN}$ is believed to be hard for both classical and quantum algorithms running in time $\poly(l)$ in the constant noise regime.

Note that the $\mathsf{LPN}$ decoding problem is an \emph{average-case} problem, where the success probability of an algorithm is measured on average over the random choices of $A$, $x$, and $e$. A prominent piece of evidence for the hardness of $\mathsf{LPN}$ is a worst-to-average-case reduction ~\cite{brakerski2019worst,yu2021smoothing}. Such a reduction means that an algorithm that succeeds with high probability over a random instance would even suffice to solve the worst-case instance of $\mathsf{LPN}$. That is to say, most instances within that class are hard. Refs. ~\cite{brakerski2019worst,yu2021smoothing} are able to provide such a reduction from the nearest codeword problem~\cite{brakerski2019worst,yu2021smoothing}, a worst-case promise version of $\mathsf{LPN}$.

Regarding the low-degree hardness of $\mathsf{LPN}$-related problems, an exemplary result is on a special case of learning parity (without noise), known as $k$-XORSAT. Here, the variables are $l$ $\{\pm 1\}$-valued $x_1,...,x_l$ and we are told $n$ random constraints of the form $x_{i_{1,s}}x_{i_{2,s}}...x_{i_{k,s}}=b_s$ with $b_s=\pm 1$, for $s\in[n]$. The goal is to determine whether there is an assignment $x_1,...,x_l\in\{\pm 1\}$ that satisfies all the constraints. This problem at $l\ll n\ll l^{3/2}$ is computationally hard for low-degree methods, i.e., it is hard to distinguish between a random $b$ (which is unsatisfiable with high probability) and a formula with a planted assignment~\cite{kunisky2019notes}. Although this problem can be solved in polynomial time using Gaussian elimination over the finite field, these intrinsic high-degree algorithms tend not to be robust to even a tiny amount of noise.

A recent work~\cite{poremba2024learning} provides a potential natural quantum analog of $\mathsf{LPN}$, \emph{Learning Stabilizer with Noise ($\mathsf{LSN}$)}, and gave evidence for its hardness by showing that $\mathsf{LSN}$ too has a worst-to-average-case reduction. It is natural to ask if $\mathsf{LSN}$ replicates the above feature of $\mathsf{LPN}$, i.e. that even its noiseless version is low-degree hard. Here, we answer in the affirmative, where the low-degree hardness we prove is qubit-wise low-degree hardness. 

The $\mathsf{LSN}$ problem is to decode a random quantum stabilizer code under local depolarizing noise. Let $S\leq\cP_n$ be an abelian subgroup of the Pauli group which has $n-l$ generators
\begin{align*}
S=\expval{g_1,...,g_{n-l}}, \qquad g_i\in \pm \{I,X,Y,Z\}^{n}.
\end{align*}
This is already sufficient to specify an error correcting code, whose codespace $C(S)\in(\C^2)^{\otimes n}$ is the simultaneous $+1$ eigenspace of all elements of $S$, i.e.
\begin{align*}
C(S)=\{\ket{\psi}:\ g\ket{\psi}=\ket{\psi},\ \forall g\in C(S)\}.
\end{align*}
If $S$ has $n-l$ generators, then $C(S)$ has $2^l$ many codewords. We use $\text{Stab}(n,l)$ to denote the set of all stabilizer subgroups of $\cP_n$ with $n-l$ generators.
Formally, the goal for $\mathsf{LSN}$ is to find the underlying random string $x\in\{0,1\}^l$ given as input a sample
\begin{align*}
(S\sim\text{Stab}(n,l),\ E\ket{\psi_{S,x}}),\quad\text{with}\quad\ket{\psi_{S,x}}=U_{\text{Enc},S}(\ket{0^{n-l}}\otimes\ket{x}),
\end{align*}
where $S\sim\text{Stab}(n,l)$ is a a uniformly random stabilizer subgroup, $E=\Lambda_\kappa^{\otimes n}$ 
is a Pauli error from a local depolarizing channel with noise rate $\kappa$, and $U_{\text{Enc},S}$ is some canonical encoding circuit for the stabilizer code associated with $S$. 

Ref. \cite{poremba2024learning} implicitly show that for every instance of $\mathsf{LPN}$, of the form $(A \in \Z_2^{n\times l},A \cdot s + e \Mod{2})$, there exists a corresponding instance of $\mathsf{LSN}$:

\begin{equation}\label{eq:mapLSNLPN_1}
(U_{A}, \ket{A\cdot  s + e \Mod{2}})
\end{equation}
where $U_{A}$ is the Clifford encoding circuit\footnote{Note that $U_{A}$ can be described in terms of $\mathrm{CNOT}$ gates~\cite{10.5555/2011763.2011767}, which are themselves Clifford gates.} encoding operation
\begin{equation}\label{eq:mapLSNLPN_2}
U_{A}: \ket{0^{n-l}} \otimes \ket{x} \rightarrow \ket{A \cdot x \Mod{2}}
\end{equation}
which is an injective matrix multiplication for any vector $x \in \Z_2^l$. These instances correspond to each other in the sense that a solution for the $\mathsf{LSN}$ problem (which takes the form of an output bit string $s\in \{0,1\}^l$) is correct if and only if $s$ is correct for the corresponding $\mathsf{LPN}$ problem.

In this work, we prove the average-case hardness of $\mathsf{LSN}$ in the low-degree framework. Before, we have to clarify what `low-degree' means in the context of a problem with classical-quantum hybrid input such as $\mathsf{LSN}$. This input consists of a classical matrix $S\sim\text{Stab}(n,l)$ containing $n-l$ stabilizers and a (noisy) quantum encoded state $E\ket{\psi_{S,x}}$. In the noiseless case, a general decoding algorithm could apply the inverse $U^\dagger_{\text{Enc},S}$ of the encoding unitary $U_{\text{Enc},S}$ and measure in the computational basis, which solves the problem exactly. While the measurement itself is a low-degree algorithm, the decision of which unitary to apply is high-degree as it depends on all $O(n^2)$ bits of $S$. This indicates that the definition of the algorithm's degree must take into account the degree of the function that chooses the unitary to apply. This would rule out the above `cheating' algorithm, as it should -- because the above algorithm could also solve the classical parities problem by performing Gaussian elimination to choose which unitary to apply. Yet it is known that parities is low-degree hard.

To address this issue, we will define a low-degree algorithm in two ways, depending on how the hybrid quantum-classical input is supplied to the algorithm:
\begin{itemize}
    \item \textbf{$\mathsf{LSN}$ with purely quantum input. }In this model, we are given $m$-copies of purely quantum input consisting of two parts. For each copy, the first $n\times l$ qubits are just a computational basis state indicating which Stabilizer code was applied (at most $O(n^2)$ bits are required to specify its stabilizer group). The last $l$ qubits are the (noisy) encoded state. In this case, what we can do is to perform a joint PVM (a global fixed unitary acting on all qubits and a computational basis measurement) on the joint state. Suppose the post-processing after the measurement is a function of degree $k$, then the overall protocol is equivalent to a classical algorithm of degree $2k$ as the output distribution is at most quadratic in each input qubit.
    \item \textbf{$\mathsf{LSN}$ with hybrid quantum-classical input with low-degree quantum encoding of classical input. }In this case, we are given the classical matrix $S$ and $m$ copies of $l$-qubit (noisy) encoded state $\ket{y}$. We can then perform a measurement on each copy of code state by implementing a quantum unitary $U(S)$ encoding all classical information on the unitary and then implementing a computational basis measurement. The output distribution is given by $p(z) = \{\abs{\braket{z|U(S)|y}}^2\}_z$, where $z$ is a $l$ bit string. In this case, the definition of a low-degree algorithm is one where $p(z)$ is a low-degree function of degree at most $k$ on each bit of $S$. This eliminates the above-mentioned `cheating' algorithm.
\end{itemize}
Given these two types of $\mathsf{LSN}$, we have the following low-degree hardness result.

\begin{theorem}[Average-case $\mathsf{LSN}$ is low-degree hard, even without noise]\label{thm:other_lsn_avg}
There are no non-adaptive low-degree algorithms of degree $\polylog(l)$ solving the average-case problem $\mathsf{LSN}$, even without noise, with probability at least $2/3$ at $n=\poly(l)$.
\end{theorem} 

\begin{proof}
We first recall the classical proof of worst-case low-degree hardness for $r$-XORSAT, a special case of learning parity (without noise). Let $\mathcal{H}$ be a $r$-uniform hypergraph with $l$ vertices and $n$ hyperedges $e_1,...,e_n$. We consider the task of distinguishing between the following two distributions over $b\in\{\pm 1\}^n$:
\begin{itemize}
    \item $b\sim\text{Unif}(\{\pm 1\}^n)$ (null).
    \item Sample a random $x\in\{\pm1\}^l$ and for $i\in[n]$ set $b_i=\prod_{j\in C_i}x_j$ (alternative planted).
\end{itemize}
We would like to show that these two distributions are low-degree indistinguishable (for some fixed $\mathcal{H}$). 

An even cover is a set of hyperedges such that each vertex appears an even number of times. Equivalently, we can write $\mathcal{H}\in\Z_2^{n\times l}$ as a matrix (equivalent to $A$ in the definition of $\mathsf{LPN}$) where the $i$-th row $C_i$ is the indicator vector of the hyperedge $e_i$. Then an even cover is a set of rows that sum up to $0$.

We will show that if $\mathcal{H}$ has no even cover of size $\leq k$, then no degree-$k$ polynomials can distinguish between the two cases. Since the null case is just the uniform distribution over $\{\pm1\}^n$, the natural basis is the monomial basis $\{b_S=\prod_{i\in S}b_i\}_{S\subseteq[n]}$. The degree-$k$ advantage is thus
\begin{align*}
\Adv{k}^2=\sum_{S\subseteq[n],\ 0<|S|\leq k}\E_b[b_S]^2.
\end{align*}
Note that the expectation is taken over the unknown random solution $x$, for any $S\subseteq[n]$ of size at most $k$, we have $\bigoplus_{i=1}^n C_i\neq\emptyset$ as $\mathcal{H}$ has no even cover of size $\leq k$. Therefore, we have
\begin{align*}
\E_b[b_S]=\E_{x\sim\{\pm 1\}^l}\prod_{i\in S}x_{C_i}=\E_{x\sim\{\pm 1\}^l}x_{\bigoplus_{i=1}^n C_i\neq\emptyset}=0.
\end{align*}

Via the above argument, we have shown that any low-degree protocol of degree $\polylog(l)$ cannot solve worst-case Learning Parities (without noise) with even negligible probability. By the mapping given by \eqref{eq:mapLSNLPN_1} and \eqref{eq:mapLSNLPN_2}, there is a worst-case Learning Stabilizers (without noise) problem corresponding to every Learning Parities (without noise) problem. Furthermore, a non-adaptive degree-$k$ quantum algorithm for the latter problem implies a degree-$2k$ classical algorithm for the former problem. This is because, although a quantum algorithm is allowed to perform quantum measurements (assumed to be single-copy PVMs) in an arbitrary basis $\{U\ketbra{x}U^\dagger\}_{x=1}^{2^n}$, $U$ is a linear transformation which doesn't increase degree and the probability distribution is at most quadratic in each input qubits. Therefore, the nonexistence of a low-degree classical algorithm for Learning Parities in the worst-case implies the nonexistence of a low-degree quantum algorithm for Learning Stabilizers in the worst-case. 

We further argue that there is also no low-degree quantum algorithm for Learning Stabilizers (without noise) problem in the average-case, by performing a worst-to-average-case reduction that preserves low-degreeness. (Note that we perform this reduction {\em exactly}, while \cite{poremba2024learning} were only able to perform a worst-to-average-case reduction to a nonstandard definition of average case where the distribution depends on the worst-case instance.) Given any worst-case instance of Learning Stabilizers (without noise), we can create an average-case instance by applying the following `re-randomizing unitary' to the input quantum state $U_{\text{Enc},S}(\ket{0^{n-l}}\otimes \ket{x})$:
\begin{itemize}
    \item a logical Pauli $X^{u}$ for the known code $S$, where $u \sim \{0,1\}^l$ is a random string
    \item a random $n$-qubit Clifford $C$ to the received state $\ket{\psi_{S,x}}$.
\end{itemize}  
This maps it to the exact average-case distribution over noiseless instances, as the logical Pauli maps the message $x$ to a uniformly random one:
   \begin{align*}
        \overline{X^{u}} E\ket{\overline{\psi_x}}^S = E\ket{\overline{\psi_{x \oplus u}}}^{S} \qquad \text{where } x\oplus u \sim \{0,1\}^l
    \end{align*}
while the random Clifford $C$ maps the code to a uniformly random code:
\begin{align*}
    C\circ U_{\text{Enc},S} = U_{\text{Enc},S'} \qquad \text{where } S' \sim \text{Stab}(n,l).
\end{align*}

We can then feed the re-randomized quantum state to the low-degree $\mathsf{LSN}$ solver, which outputs the correct $x'$ with high probability. We then output $x'\oplus u$. This amounts to a low-degree quantum algorithm for the worst-case $\mathsf{LSN}$ instance that succeeds with high probability as, again, applying the re-randomizing unitary described above does not increase the degree of the solution. 
\end{proof}

\subsection{Low-degree hardness of agnostic tomography of product states}\label{sec:conn_agnostic_product}

In this section, we observe that if one starts with a classical hypothesis testing problem which is classically low-degree hard, then if one encodes it into a quantum hypothesis testing problem using any low-degree mapping, the resulting quantum problem is quantum low-degree hard. The reasoning is along similar lines to the previous section: given a choice of measurements of the encoded state, the distribution over readouts is itself low-degree in the original classical input.

We carry out this observation formally for the special case of tensor PCA and use it to derive a new average-case hardness result for quantum learning. 

Specifically, we consider the task of agnostic tomography~\cite{grewal2024improved,grewal2024agnostic,chen2024stabilizer,bakshi2024learninga}. We are given copies of an unknown quantum state $\rho$ that has fidelity at least $\tau$ with some state in a given class $\Xi$. Our goal is to find a state $\sigma\in\Xi$ such that the fidelity between $\sigma$ and $\rho$ is at least $\tau-\epsilon$. When $\Xi$ is the set of product states, Ref.~\cite{bakshi2024learninga} proposed a computational and sample efficient algorithm in the regime $\tau\geq O(1)$ and $\epsilon=1/\poly(n)$. It is further shown that when $\tau=1/\poly(n)$, this task is NP-hard in the worst-case.

In this section, we show that agnostic tomography of product state when $\tau=1/\poly(n)$ is \emph{average-case} low-degree hard for protocols with a polynomial number of single-copy PVMs. Quantitatively, we have:
\begin{theorem}[\emph{Average-case} hardness of agnostic tomography of product state at $\tau=1/\poly(n)$]\label{thm:agnostic_product}
There exists a polynomial $p(n)$ and a non-trivial probability distribution $\mu$ over all pure states such that given $\ket{\psi}\sim\mu$, it has fidelity at least $\tau\geq 1/p(n)$ with the set $\Xi$ of product state. However, finding a product state $\sigma\in\Xi$ such that the fidelity $\braket{\psi|\sigma|\psi}$ between $\ket{\psi}$ and $\sigma$ is at least $\tau-\epsilon$ for any $0<\epsilon\leq\tau/2$ is average-case hard for quantum algorithms with degree $k\leq o(n)$ and polynomially many single-copy measurements.
\end{theorem}

\begin{remark}
The proof of \Cref{thm:agnostic_product} follows from a low-degree quantum encoding of the classical problem of tensor principal components analysis (PCA) into quantum data. A similar encoding can be extended to transfer any low-degree-hard additive Gaussian problem~\cite{kunisky2019notes} into a quantum hypothesis testing problem, for which quantum low-degree hardness will follow from classical low-degree hardness.
\end{remark}

\noindent To prove \Cref{thm:agnostic_product}, we first consider the following classical hypothesis testing task, which is known as the order-$4$ tensor principal components analysis (PCA) problem. Given signal parameter $\lambda$, we are given a classical tensor $T\in\C^{n\times n\times n\times n}$ which is sampled from the following two cases:
\begin{itemize}
    \item $T=G$ where $G\in\C^{n\times n\times n\times n}$ is a Gaussian random tensor with each entry an IID standard Gaussian variable $\sim\mathcal{N}(0,1)$ (null).
    \item $T=\lambda x^{\otimes 4}+G$ where $x\in\{\pm 1\}^{n}$ is chosen uniformly at random and $G\in\C^{n\times n\times n\times n}$ is a Gaussian random tensor (alternative).
\end{itemize}
Regarding this hypothesis testing problem, Ref.~\cite{kunisky2019notes} proved the following results on the classical low-degree likelihood ratio in the query model where one can query the entries $T$ via a low-degree function of degree at most $k$.

\begin{lemma}[Theorem 3.3 of~\cite{kunisky2019notes}]\label{lem:tensor_pca}
Consider the above order-$4$ tensor PCA distinguishing problem with $\lambda\leq O(n^{-1}k^{-1/2})$, the square degree-$k$ advantage for algorithms of degree at most $k$ that query the tensor $T$ is bounded by $o(1)$.
\end{lemma}

In the following, we consider a modified version of this tensor PCA problem where each tensor is normalized. Given signal parameter $\lambda$, we are given a classical tensor $T\in\C^{n\times n\times n\times n}$ which is sampled from the following two cases:
\begin{itemize}
    \item $T=G/\norm{G}_F$ where $G\in\C^{n\times n\times n\times n}$ is a Gaussian random tensor with each entry an IID standard Gaussian variable $\sim\mathcal{N}(0,1)$ (null).
    \item $T=(\lambda x^{\otimes 4}+G)/\norm{\lambda x^{\otimes 4}+G}_F$ where $x\in\{\pm 1\}^{n}$ is chosen uniformly at random and $G\in\C^{n\times n\times n\times n}$ is a Gaussian random tensor (alternative).
\end{itemize}
Here, the Frobenius norm $\norm{T}_F$ of a tensor $T$ is defined as
\begin{align*}
\norm{T}_F\coloneqq\sqrt{\sum_{i,j,k.l=1}^nT_{ijkl}^2}.
\end{align*}
We can show that this normalized order-$4$ tensor PCA distinguishing problem is also low-degree hard for algorithms with degree at most $k$ under the same parameter regime with \Cref{lem:tensor_pca}.
\begin{lemma}\label{lem:normal_tensor_pca}
Consider the above \emph{normalized} order-$4$ tensor PCA distinguishing problem with $\lambda\leq O(n^{-1}k^{-1/2})$, the square degree-$k$ advantage for algorithms of degree at most $k$ that query the tensor $T$ is bounded by $o(1)$.
\end{lemma}

\begin{proof}
We first consider the null case when $T=G/\norm{G}_F$ for a Gaussian random tensor $G\sim\C^{n\times n\times n\times n}$. Given a Gaussian random tensor $G$, the average value and variance of square Frobenius norm of $G$ is given by
\begin{align*}
\E[\norm{G}_F^2]&=\E\left[\sum_{i,j,k.l=1}^nT_{ijkl}^2\right]=n^4,\\
\Var[\norm{G}_F^2]&=\Var\left[\sum_{i,j,k.l=1}^nT_{ijkl}^2\right]=\sum_{i,j,k.l=1}^n\Var\left[T_{ijkl}^2\right]=2n^4,
\end{align*}
where we use the fact that each entry of $G$ is chosen from IID standard Gaussian distribution $\mathcal{N}(0,1)$. Therefore, by the Chernoff bound, we have
\begin{align*}
(1-3\epsilon)n^4\leq\norm{G}_F^2\leq(1+3\epsilon)n^4
\end{align*}
with probability at least $1-2\exp(-\Theta(n^4\epsilon^2))$ for $0\leq\epsilon\leq1/2$. Note that $(1-3\epsilon)\leq (1-\epsilon)^2$ and $(1+\epsilon)^2\leq 1+3\epsilon$ for $0\leq\epsilon\leq1/2$, we thus have
\begin{align*}
(1-\epsilon)n^2\leq\norm{G}_F\leq(1+\epsilon)n^2
\end{align*}
with probability at least $1-2\exp(-\Theta(n^4\epsilon^2))$ for $0\leq\epsilon\leq1/2$.

We then consider the alternative case when $T=(\lambda x^{\otimes 4}+G)/\norm{\lambda x^{\otimes 4}+G}$. For a random choice of $\lambda x^{\otimes 4}+G$, the square Frobenius norm is computed as
\begin{align*}
\norm{\lambda x^{\otimes 4}+G}_F^2=\lambda^2\norm{x^{\otimes 4}}_F^2+\lambda\expval{x^{\otimes 4},G}+\norm{G}_F^2.
\end{align*}
The first term will always be $\lambda^2n^4$ by definition, and the last term has mean value $n^4$ and variance $2n^4$. We thus consider the cross term $\lambda\expval{x^{\otimes 4},G}$. As $x$ is chosen uniformly at random from $\{\pm 1\}^{n}$, the mean value for this cross term is $0$. The variance is given by
\begin{align*}
\Var\left[\lambda\expval{x^{\otimes 4},G}\right]=&\lambda\E\left[\expval{x^{\otimes 4},G}^2\right]\\
&=\lambda\E\left[\left(\sum_{i,j,k,l=1}^nx_ix_jx_kx_lG_{ijkl}\right)^2\right]\\
&=\lambda\E\left[\sum_{i,j,k,l=1}^n\left(x_ix_jx_kx_l\right)^2G_{ijkl}^2\right]\\
&=\lambda n^4
\end{align*}
where we use the fact that $G_{ijkl}$ are IID standard Gaussian variables and thus $\E[G_{ijkl}G_{i'j'k'l'}]=\delta_{ii'}\delta_{jj'}\delta_{kk'}\delta_{ll'}$. Summing all these together, we have
\begin{align*}
\E[\norm{\lambda x^{\otimes 4}+G}_F^2]&=\E\left[\sum_{i,j,k.l=1}^nT_{ijkl}^2\right]=(1+\lambda^2)n^4,\\
\Var[\norm{\lambda x^{\otimes 4}+G}_F^2]&=\Var\left[\sum_{i,j,k.l=1}^nT_{ijkl}^2\right]=\sum_{i,j,k.l=1}^n\Var\left[T_{ijkl}^2\right]=(2+\lambda)n^4.
\end{align*}
Given that $\lambda\leq O(n^{-1}k^{-1/2})=o(1)$, we conclude that 
\begin{align*}
(1-\epsilon)\sqrt{1+\lambda^2}n^2\leq\norm{\lambda x^{\otimes 4}+G}_F\leq(1+\epsilon)\sqrt{1+\lambda^2}n^2
\end{align*}
with probability at least $1-2\exp(-\Theta(n^4\epsilon^2))$ for $0\leq\epsilon\leq1/2$.

If we add a fixed scaling factor to the original version of tensor PCA, the low-degree hardness result in \Cref{lem:tensor_pca} will not change. In the normalized tensor PCA, with probability $1-2\exp(-\Theta(n^4\epsilon^2))$ we will add a scaling factor within the range of $[(1+\epsilon)^{-1}n^{-2},(1-\epsilon)^{-1}n^{-2}]$ for the null case compared to the standard case, and a scaling factor within the range of $[(1+\epsilon)^{-1}(\sqrt{1+\lambda^2})^{-1}n^{-2},(1-\epsilon)^{-1}(\sqrt{1+\lambda^2})^{-1}n^{-2}]$ for the alternative case compared to the standard case. The deviation on the low-degree likelihood ratio is within the range $[(\sqrt{1+\lambda^2})^{-1}\cdot\frac{1-\epsilon}{1+\epsilon},(\sqrt{1+\lambda^2})^{-1}\cdot\frac{1+\epsilon}{1-\epsilon}]$, the two sides of which are both of constant scaling. Thus, we can still obtain the same low-degree hardness result for the normalized tensor PCA.
\end{proof}

We then consider the spectral norm of the two cases for the normalized tensor PCA. The spectral norm of an $n\times n\times n\times n$ tensor $T\in\mathbb{C}^{n\times n\times n\times n}$ is defined to be
\begin{align*}
\norm{T}_{\text{spec}}=\max_{x,y,z,u\in\mathbb{C}^n}\frac{\abs{\expval{T,x\otimes y\otimes z\otimes u}}}{\norm{x}_2\norm{y}_2\norm{z}_2\norm{u}_2}.
\end{align*}
It is shown in~\cite{bakshi2024learninga} that the existence of an algorithm for agnostic tomography of product states at parameter $\tau$ and $\epsilon\lesssim\tau$ indicates the existence of an algorithm that estimates the spectral norm of a tensor $T\in\C^{n'\times n'\times n'\times n'}$. Here, we further show that such reduction is low-degree. Formally, we have the following lemma:
\begin{lemma}\label{lem:agnostic_reduction}
Suppose approximating the spectral norm of random tensors $T\in\C^{n'\times n'\times n'\times n'}$ with respect to some distribution corresponding to $\mu$ to additive error $\epsilon'\norm{T}_F$ is hard for any classical algorithm of degree $2k$. then agnostic tomography of product states with the unknown quantum state encoded from the random classical data drawn from $\mu$ in a particular way at some parameter $\tau=\poly(\epsilon')$ and $\epsilon\lesssim\tau$ is hard for any quantum algorithm of degree $k$ using polynomially many single-copy measurements. 
\end{lemma}

We will delay the proof of this lemma. Based on \Cref{lem:agnostic_reduction}, we now provide the proof for the main result claimed in \Cref{thm:agnostic_product}.

\begin{proof}[Proof of \Cref{thm:agnostic_product}]
We first consider the normalized tensor PCA. In the null case, the tensor is $T=G/\norm{G}_F$ for a Gaussian random tensor $G\in\C^{n\times n\times n\times n}$. By the concentration inequality for Gaussian random tensor~\cite{tomioka2014spectral}, we have
\begin{align*}
\norm{G}_{\text{spec}}\leq\Theta(\sqrt{n})
\end{align*}
with probability at least $1-\exp(-\Theta(n))$. Therefore, we have
\begin{align*}
\norm{\frac{G}{\norm{G}_F}}_{\text{spec}}\leq\Theta(n^{-3/2})
\end{align*}
with probability at least $1-\exp(-\Theta(n))$. 

In the alternative case, we observe that
\begin{align*}
\norm{\lambda x^{\otimes 4}+G}_{\text{spec}}\geq \lambda n^2.
\end{align*}
Therefore, we have 
\begin{align*}
\norm{\frac{\lambda x^{\otimes 4}+G}{\norm{\lambda x^{\otimes 4}+G}_F}}_{\text{spec}}\geq\Theta(\lambda)
\end{align*}
with probability at least $1-\exp(-\Theta(n^2))$ when $\lambda=o(1)$. 

Given $\lambda=\Theta(n^{-1}k^{-1/2})$ and $k=o(n)$, we have
\begin{align*}
\norm{\frac{\lambda x^{\otimes 4}+G}{\norm{\lambda x^{\otimes 4}+G}_F}}_{\text{spec}}\geq\omega(n^{-3/2}).
\end{align*}
Note that the tensor in the normalized tensor PCA always has $\norm{T}_F=1$. If there is an algorithm that can estimate the spectral norm to $n^{-3/2}\norm{T}_F$, we can use this algorithm to solve the normalized tensor PCA problem at any $k=o(n)$. 

Since the normalized tensor PCA is low-degree hard by as shown by \Cref{lem:normal_tensor_pca} for any classical algorithms of degree $k=o(n)$ at $\lambda\leq O(n^{-1}k^{-1/2})$, then estimating the spectral norm of these random tensors $T\in\C^{n\times n\times n\times n}$ to $n^{-3/2}\norm{T}_F$ additive error is hard any classical algorithms of degree $k=o(n)$. By \Cref{lem:agnostic_reduction}, when we encode these random tensors into random quantum states in \eqref{eq:quantum_encode}, agnostic tomography of product states with these random quantum state is hard any quantum algorithms of degree $k/2=o(n)$ making polynomially many single-copy measurements. 
\end{proof}

Finally, we provide the proof of \Cref{lem:agnostic_reduction}.

\begin{proof}[Proof of \Cref{lem:agnostic_reduction}]
We will follow the reduction inspired by~\cite{bakshi2024learninga}. Assume we are given a tensor $T\in\mathbb{C}^{n\times n\times n\times n}$ with $\norm{T}_F=1$, we construct
a quantum state on $4n$ qubits as
\begin{align}\label{eq:quantum_encode}
\ket{\psi_T}=\sum_{i,j,k,l}T_{ijkl}'\ket{e_ie_je_ke_l},
\end{align}
where $\ket{e_ie_je_ke_l}=\ket{e_i}\otimes\ket{e_j}\otimes\ket{e_k}\otimes\ket{e_l}$, where $\ket{e_i}$ is the product state that is $\ket{1}$ on the $i$-th qubit and $\ket{0}$ on the others. We note that this is a qubit encoding that preserves the Boolean properties. Moreover, this qubit encoding makes every qubit of the quantum state a degree-$1$ function on any classical input bit.

For $T\in\mathbb{C}^{n'\times n'\times n'\times n'}$ and a Haar random isometry $U\in\mathbb{C}^{n\times n'}$, as long as $n'>10\log n$, then with $0.99$ probability over the randomness of $U$, we have $\max_{ijkl}\abs{U^{\otimes 4}T}\leq\left(\frac{10n'}{n}\right)^2$ with $U^{\otimes 4}T\in\C^{n\times n\times n\times n}$ by Claim 7.5 in~\cite{bakshi2024learninga}. Note that $U$ is linear, $U^{\otimes 4}T$ is a degree-$1$ function on any classical input bit.

Given a tensor $T\in\mathbb{C}^{n\times n\times n\times n}$ with $\norm{T}_F=1$ and denote $M=\max_{ijkl}n^2\abs{T_{ijkl}}$, let $\text{OPT}_T$ be the spectral norm of $T$ and $\text{OPT}_{\ket{\psi_T}}$ to be defined as
\begin{align*}
\text{OPT}_{\ket{\psi_T}}=\max_{\ket{\sigma}=\ket{\sigma_1}\otimes...\otimes\ket{\sigma_n}}\abs{\braket{\sigma|\psi_T}}.
\end{align*}
Lemma 7.6 of~\cite{bakshi2024learninga} then shows that for large enough $n$, we have
\begin{align*}
e^{-2}\text{OPT}_T-\frac{10M}{n^{0.2}}\leq\text{OPT}_{\ket{\psi_T}}\leq e^{-2}\text{OPT}_T+\frac{10}{n^{0.1}}+\frac{10M}{n^{0.2}}.
\end{align*}

In the following, we reduce approximating the norm of a tensor to agnostic learning product states to inverse polynomial accuracy. Given $\norm{T}_F=1$, we consider $U^{\otimes 4}T$ for $U$ an $\epsilon$-approximate Haar random isometry, and the corresponding quantum state $\ket{\psi_{U^{\otimes 4}T}}$. We thus have
\begin{align*}
e^{-2}\text{OPT}_{U^{\otimes 4}T}-\frac{10M}{n^{0.2}}\leq\text{OPT}_{\ket{\psi_{U^{\otimes 4}T}}}\leq e^{-2}\text{OPT}_{U^{\otimes 4}T}+\frac{10}{n^{0.1}}+\frac{10M}{n^{0.2}}.
\end{align*}
By the claim, we can choose $M=(10m)^2$. As $U$ is a isometry, we have $\text{OPT}_{U^{\otimes 4}T}=\text{OPT}_{T}$. We thus have
\begin{align*}
e^{-2}\text{OPT}_{T}-\frac{1000m^2}{n^{0.2}}\leq\text{OPT}_{\ket{\psi_{U^{\otimes 4}T}}}\leq e^{-2}\text{OPT}_{T}+\frac{10}{n^{0.1}}+\frac{1000m^2}{n^{0.2}}.
\end{align*}
Choosing $\epsilon<0.01$ and $n=\Omega((m/\epsilon)^{20})$, we have
\begin{align*}
e^{-2}(\text{OPT}_{T}-\epsilon)\leq\text{OPT}_{\ket{\psi_{U^{\otimes 4}T}}}\leq e^{-2}(\text{OPT}_{T}+\epsilon).
\end{align*}
Therefore, finding the optimal product state fidelity to accuracy $0.01\epsilon^2$ gives the spectral norm of the tensor to $\epsilon$ additive error. 

Recall that both the encoding from classical tensor to the quantum state is degree $1$, and the random rotation $U^{\otimes 4}T$ is also degree $1$. When we measure each copy of the quantum state $\ket{\psi_{U^{\otimes 4}T}}$ using a single-copy PVM $\{V\ket{x}\bra{x}V^\dagger\}_{x}$ for some fixed unitary $V$, the output distribution over $\{\abs{\braket{x|V|\psi_T}}^2\}_x$ is a degree -$2$ function on each qubit as $V$ is also linear. Combining all the above results together, we conclude that, if there exists a low-degree quantum algorithm of degree $k$ for agnostic tomography of product states such that for any $n$ qubit state at parameter $\tau$ and $\epsilon\lesssim\tau$ with sample complexity and running time $f(n,\tau)$ for some function $f$ and high success probability. There is a classical algorithm of degree $2k$ that can approximate the spectral norm of any tensor $T\in\C^{n'\times n'\times n'\times n'}$ to additive error $\epsilon'\norm{T}_F$ with running time $f(\poly(n'/\epsilon'),\poly(\epsilon'))$ with high probability. This is the claim in \Cref{lem:agnostic_reduction}.
\end{proof}

\section*{Acknowledgments}

We are grateful to Tim Hsieh for generously sharing his expertise on low-degree hardness during this project. We also thank Alex Wein for his thoughtful and thorough feedback on an early version of this manuscript. We also thank Scott Aaronson, Anurag Anshu, John Bostanci, Adam Bouland, Hsin-Yuan Huang, Jerry Li, Yunchao Liu, Zidu Liu, Qi Ye, and Henry Yuen for illuminating discussions on learning and cryptography, learning shallow circuits, matrix product states, and agnostic learning of product states. SC acknowledges support from NSF Award 2430375.

%%%%%%%%%%%%%%%%%%%%%%%%%%%%%%%%%%%%%%%%%%%%%%%%%%%%%%%%%%%%%%%%%%%%%%%%%%%%%%%%%%%%%%%%%%%%%%%%%%%%%%%%%%%%%%

\bibliographystyle{MyRefFont}
\bibliography{quantum-low-degree}

%%%%%%%%%%%%%%%%%%%%%%%%%%%%%%%%%%%%%%%%%%%%%%%%%%%%%%%%%%%%%%%%%%%%%%%%%%%%%%%%%%%%%%%%%%%%%%%%%%%%%%%%%%%%%%

\clearpage
\newpage
\appendix

\section{The corner case: Low-degree hardness of learning stabilizer states}\label{sec:stabilizer_state}

Here, we also provide a proof of low-degree hardness for learning stabilizer states, which is an exact quantum analogue of parity. We call this a corner case because, similar to parity, we can show low-degree hardness for learning stabilizer states even though there exist time efficient high-degree algorithms. 

Let us first define stabilizer states. Let $\mathcal{P}_n$ be the \emph{$n$-qubit Pauli group}, i.e.\ all tensor products of 
single-qubit Pauli matrices $\{I, X, Y, Z\}$ (including the global phases 
$\{\pm 1, \pm i\}$), where 
\begin{equation*}
    I=\begin{pmatrix}1 & 0\\0& 1\end{pmatrix}\,, \qquad \sigma_x=\begin{pmatrix}0 & 1\\1& 0\end{pmatrix}\,, \qquad \sigma_y=\begin{pmatrix}0 & -i\\i& 0\end{pmatrix}\,, \qquad \sigma_z=\begin{pmatrix}1 & 0\\0& -1\end{pmatrix}.
\end{equation*}
A subgroup $G \subseteq \mathcal{P}_n$ is called a 
\emph{stabilizer group} if:
\begin{enumerate}
  \item $G$ is abelian (i.e.\ all elements commute),
  \item $-I \notin G$, where $I$ is the $n$-qubit identity operator,
  \item $\lvert G \rvert = 2^r$ for some $r \leq n$.
\end{enumerate}

A \emph{stabilizer state} is any state that is the simultaneous $+1$ eigenstate of a maximal stabilizer group $G$, i.e. one with $n$ generators. The task of learning stabilizer states is the following:

\begin{definition}[Learning stabilizer states]\ \\ \vspace{-.5cm}
\begin{description}
\item[Input:] $m$ copies of $\ket{\psi}$, where $\ket{\psi}$ is an $n$-qubit stabilizer state.
\item[Output:] Classical description of $\ket{\psi}$.
\end{description}
\end{definition}
We note that $\ket{\psi}$ has an efficient classical description -- for example, its set of stabilizers, which take $O(n^2)$ bits to specify. 

On the algorithmic front, Aaronson and Gottesman~\cite{aaronson2008identifying} provided polynomial-time learning algorithms using either $m=O(n^2)$ single-copy measurements or $m=O(n)$ joint measurements on multiple copies. In particular, the single-copy measurement approach relies on a procedure called computational difference sampling. Montanaro~\cite{montanaro2017learning} then provided an improved algorithm using $O(n)$ two-copy measurements via a procedure called Bell difference sampling. Both computational difference sampling based on single-copy measurements and Bell difference sampling based on two-copy measurements rely on high-degree classical post-processing using the measurement data collected. The key subroutine of these sampling algorithms is to sample the Pauli observables (in a high-degree fashion) that stabilize the unknown stabilizer states from the sampling. However, there are up to $2^n$ possible stabilizer states that share the same group of Pauli observables. It is even unknown if there are any low-degree algorithms that can find the underlying stabilizer state given a group of Pauli observables that stabilize it. 

In general, it is natural to ask if there is a low-degree algorithm that learns stabilizer states also from a learning theory perspective. This is because learning stabilizer states can be one possible natural quantum analog of learning parity, which is shown to be low-degree hard in the classical framework~\cite{kunisky2019notes}. In the following, we show a negative answer for single-copy PVMs with up to $\Theta(n)$ ancillary qubits.

We simply consider distinguishing between random stabilizer states (alternative ensemble) and the maximally mixed state (null). As any low-degree algorithm that learns stabilizer states with high probability would imply a low-degree algorithm for this distinguishing problem, the following implies low-degree hardness for learning stabilizer states:

\begin{corollary}\label{coro:stabilizer_ancilla}
For the hypothesis testing problem $\rhomm$ vs. random stabilizer states, weak detection is degree-$k$ hard for any protocol that makes polynomially many non-adaptive single-copy PVMs $\{\{U_i\ket{x}\bra{x}U_i^{\dagger}\}_{x=1}^{2^{n+n'}}\}_{i=1}^m$ with up to $n'=\Theta(n)$ ancillary qubits.
Moreover, the task of learning stabilizer states is hard for degree-$k$ protocols with a polynomial number of non-adaptive single-copy PVMs with up to $n'=\Theta(n)$ ancillary qubits.
\end{corollary}

\begin{proof}
    It is known that random stabilizer states form a state $3$-design~\cite{zhu2017multiqubit}. The claim thus follows immediately by~\Cref{coro:general_ancilla}.
\end{proof}

\section{Deferred proofs in the general framework for non-adaptive measurements}

\subsection{Proof of \texorpdfstring{\Cref{coro:general_single_qubit}}{Corollary 6.2}}\label{sec:coro_kLLR}

For any subset $T\subseteq[mn]$ of size $\leq k$, 
\begin{align*}
\norm{\mathbb{E}_{\rho\sim\mathcal{E}}\overline{D}_{\tr_{[mn]\backslash T}(\rho^{\otimes m})}-\textbf{1}}^2& = \mathbb{E}_{\mathbf{s}}\mathbb{E}_{\rho, \rho'\sim\mathcal{E}} \overline{D}_{\tr_{[mn]\backslash T}(\rho^{\otimes m})}(\mathbf{s}) \overline{D}_{\tr_{[mn]\backslash T}(\rho'^{\otimes m})}(\mathbf{s}) - 1\\
&= \mathbb{E}_{\mathbf{s}}2^{2{|T|}} \left(2^{-\abs{T}}+\zeta_{\bm{s}}\right)\left(2^{-\abs{T}}+\zeta_{\bm{s}}\right)-1,\quad \text{for some }\zeta_{\bm{s}}\in[-\epsilon,\epsilon]\\
& \leq \epsilon^2 2^{2|T|}.
\end{align*}
where the second equality follows from the $k$-local distinguishability assumption.

\subsection{Low-degree hardness for measurements prepared by bounded-depth circuits}\label{sec:general_bounded_depth}

We also consider other single-copy PVM class. Given any PVM $\{\ket{\phi_{s_i}}\bra{\phi_{s_i}}\}_{s_i}$, we can always write it $\{U_i\ket{s_i}\bra{s_i}U_i^\dagger\}_{s_{i}}$ for $s_i\in\{0,1\}^n$ (see \Cref{sec:basic_quantum}). Recall that the square degree-$k$ advantage can be computed as 
\begin{align*}
\norm{\mathbb{E}_{\rho\sim\mathcal{E}}\overline{D}_{\rho,m}^{\leq k}-\textbf{1}}^2&=\sum_{\abs{T}\leq k,T\neq\emptyset}\mathbb{E}_{\rho,\rho'\sim\mathcal{E}}\mathbb{E}_{\bm{s},\bm{s}'}\overline{D}_{\rho,m}(\bm{s})\overline{D}_{\rho',m}(\bm{s}')\chi_T(\bm{s})\chi_T(\bm{s}'),
\end{align*}
where
\begin{align*}
\overline{D}_{\rho,m}(\bm{s})=2^{mn}\prod_{i=1}^m\tr\left(\rho U_i\ket{s_i}\bra{s_i}U_i^\dagger\right)=2^{mn}\tr\left(\rho^{\otimes m}\bigotimes_{i=1}^{m}U_i\ket{s_i}\bra{s_i}U_i^\dagger\right).
\end{align*}

As a particular example, we consider the case when $U_i$'s can be efficiently prepared by ($\mathscr{D}$-dimensional geometrically-local) circuits of depth at most $L$. In this case, we again fix a particular $T$ with $\abs{T}\leq k$. We have:
\begin{align*}
&\mathbb{E}_{\rho,\rho'\sim\mathcal{E}}\mathbb{E}_{\bm{s},\bm{s}'}\overline{D}_{\rho,m}(\bm{s})\overline{D}_{\rho',m}(\bm{s}')\chi_T(\bm{s})\chi_T(\bm{s}')\\
=&\mathbb{E}_{\rho,\rho'\sim\mathcal{E}}\mathbb{E}_{\bm{s},\bm{s}'}2^{2mn}\tr\left(\rho^{\otimes m}\bigotimes_{i=1}^{m}U_i\ket{s_i}\bra{s_i}U_i^\dagger\right)\tr\left(\rho'^{\otimes m}\bigotimes_{i=1}^{m}U_i\ket{s'_i}\bra{s'_i}U_i^\dagger\right)\chi_T(\bm{s})\chi_T(\bm{s}')\\
=&\left(\underset{\rho\sim\mathcal{E}}{\mathbb{E}}\underset{\bm{s}[T]}{\mathbb{E}}2^{\abs{T}}\tr\left(\rho^{\otimes m}\bigotimes_{i=1}^{m}U_i\cdot\ket{\bm{s}[T]}\bra{\bm{s}[T]}\cdot \bigotimes_{i=1}^{m}U_i^\dagger\right)\chi_T(\bm{s})\right)^2.
\end{align*}
When $U_i$'s can be efficiently prepared by ($\mathscr{D}$-dimensional geometrically-local) circuits of depth at most $L$, $\bigotimes_{i=1}^{m}U_i\cdot\ket{\bm{s}[T]}\bra{\bm{s}[T]}\cdot \bigotimes_{i=1}^{m}U_i^\dagger$ is an observable of size at most $2^Lk$ ($(2L+k)^{\mathscr{D}}$ qubits. In this case, similar to the case for single-qubit measurements, we need local indistinguishability for subsystems of size $2^Lk$ ($(2L+k)^{\mathscr{D}}$.

\noindent We note that we prove \Cref{coro:general_bounded_depth} directly using the stronger assumption similar to \Cref{coro:general_single_qubit} for convenience. It is possible to prove the same result based on a weaker assumption in \Cref{thm:general_single_qubit}.

\subsection{Proof of \texorpdfstring{\Cref{thm:general_ancilla}}{Theorem 6.5}}\label{sec:pf_general_ancilla}

We now compute the square degree-$k$ advantage, which is 
\begin{align*}
\norm{\mathbb{E}_{\rho\sim\mathcal{E}}\overline{D}_{\rho,m}^{\leq k}-\textbf{1}}^2=\mathbb{E}_{\bm{s}}\mathbb{E}_{\rho,\rho'\sim\mathcal{E}}\overline{D}_{\rho,m}^{\leq k}(\bm{s})\overline{D}_{\rho',m}^{\leq k}(\bm{s})-1.
\end{align*}
We can also decompose the likelihood ratio for the whole history as 
\begin{align*}
\overline{D}_{\rho,m}(\bm{s})=\sum_{T\subseteq[mn]}\alpha_T^\rho\chi_T(\bm{s}),\ \alpha_T^{\rho}=\frac{1}{2^{nm}}\sum_{\bm{s}}\overline{D}_{\rho,m}(\bm{s})\chi_T(\bm{s})=\mathbb{E}_{\bm{s}}\overline{D}_{\rho,m}(\bm{s}).
\end{align*}
The set $\chi_T(\bm{s})$ is again an orthonormal basis of functions. We note that here we cannot choose the Fourier basis as it is no longer orthogonal under $D_\emptyset$. But there exists a set of orthogonal bases $\chi_T(\bm{s})$.

Given that the square degree-$k$ advantage satisfies:
\begin{align*}
\overline{D}_{\rho,m}^{\leq k}(\bm{s})=\sum_{\abs{T}\leq k}\alpha_T^\rho\chi_T(\bm{s}).
\end{align*}
We can thus compute the square degree-$k$ advantage as
\begin{align*}
\norm{\mathbb{E}_{\rho\sim\mathcal{E}}\overline{D}_{\rho,m}^{\leq k}-\textbf{1}}^2=\sum_{\abs{T}\leq k,T\neq\emptyset}\mathbb{E}_{\rho,\rho'\sim\mathcal{E}}\mathbb{E}_{\bm{s},\bm{s}'}\overline{D}_{\rho,m}(\bm{s})\overline{D}_{\rho',m}(\bm{s}')\chi_T(\bm{s})\chi_T(\bm{s}').
\end{align*}

In the following, we consider the contribution for a fixed set $T\in[m(n+n')]$. We denote $\bm{s}[T]$ to be the set of $s_{i,r}$ with $(i,r)\in T$ for $i=1,...,m$ and $r=1,...,n+n'$. Recall that each PVM with $n'$ ancilla qubits can be represented by $\{\{U_i\ket{s_i}\bra{s_i}U_i^\dagger\}_{s_i=1}^{2^{n+n'}}\}_{i=1}^m$ We have:
\begin{align*}
&\mathbb{E}_{\rho,\rho'\sim\mathcal{E}}\mathbb{E}_{\bm{s},\bm{s}'}\overline{D}_{\rho,m}(\bm{s})\overline{D}_{\rho',m}(\bm{s}')\chi_T(\bm{s})\chi_T(\bm{s}')\\
=&\mathbb{E}_{\rho,\rho'\sim\mathcal{E}}\mathbb{E}_{\bm{s},\bm{s}'}\prod_{i=1}^m\frac{\tr(\rho\otimes\ketbra{0}^{\otimes n'}\cdot\ketbra{\phi_{s_i}})}{\tr(\rhomm\otimes\ketbra{0}^{\otimes n'}\cdot\ketbra{\phi_{s_i}})}\prod_{i=1}^m\frac{\tr\left(\rho'\otimes\ketbra{0}^{\otimes n'}\cdot\ketbra{\phi_{s_i'}}\right)}{\tr\left(\rhomm\otimes\ketbra{0}^{\otimes n'}\cdot\ketbra{\phi_{s_i'}}\right)}\chi_T(\bm{s})\chi_T(\bm{s}')\\
=&\underset{\rho,\rho'\sim\mathcal{E}}{\mathbb{E}}\sum_{\bm{s},\bm{s'}}\prod_{i=1}^m\tr(\rho\otimes\ketbra{0}^{\otimes n'}\cdot\ketbra{\phi_{s_i}})\tr\left(\rho'\otimes\ketbra{0}^{\otimes n'}\cdot\ketbra{\phi_{s_i'}}\right)\chi_T(\bm{s})\chi_T(\bm{s}')\\
=&\underset{\rho,\rho'\sim\mathcal{E}}{\mathbb{E}}\sum_{\bm{s},\bm{s'}}\tr\left(\bigotimes_{i=1}^mU_i\left(\rho\otimes\ketbra{0}^{\otimes n'}\right)U_i^\dagger\cdot\ketbra{s_i}\right)\tr\left(\bigotimes_{i=1}^mU_i\left(\rho'\otimes\ketbra{0}^{\otimes n'}\right)U_i^\dagger\cdot\ketbra{s_i'}\right)\chi_T(\bm{s})\chi_T(\bm{s}')\\
=&\underset{\rho,\rho'\sim\mathcal{E}}{\mathbb{E}}\sum_{\bm{s}[T],\bm{s'}[T]}\tr\left(\underset{[m(n+n')]\backslash T}{\tr}\left(\bigotimes_{i=1}^{m}U_i\left(\rho\otimes\ketbra{0}^{\otimes n'}\right) U_i^\dagger\right)\bigotimes_{(i,r)\in T}\ketbra{s_{i,r}}\right)\cdot\\
&\qquad \tr\left(\underset{[m(n+n')]\backslash T}{\tr}\left(\bigotimes_{i=1}^{m}U_i\left(\rho'\otimes\ketbra{0}^{\otimes n'}\right) U_i^\dagger\right)\bigotimes_{(i,r)\in T}\ketbra{s'_{i,r}}\right)\chi_T(\bm{s})\chi_T(\bm{s}')
\end{align*}
where the last line follows from the fact that $\chi_T(\bm{s}),\chi_T(\bm{s}')$ only depend on $\bm{s}[T],\bm{s}'[T]$. 

Note that $\underset{[m(n+n')]\backslash T}{\tr}\left(\bigotimes_{i=1}^{m}U_i\left(\rho'\otimes\ketbra{0}^{\otimes n'}\right) U_i^\dagger\right)$ is the RDM of $\bigotimes_{i=1}^{m}U_i\left(\rho'\otimes\ketbra{0}^{\otimes n'}\right) U_i^\dagger$ on the qubits in $T$. The above equation can be bounded if 
\begin{align*}
d_{\tr}\left(\mathbb{E}_{\rho\sim\mathcal{E}}\underset{[m(n+n')]\backslash T}{\tr}\left(\bigotimes_{i=1}^{m}U_i\left(\rho\otimes\ketbra{0}^{\otimes n'}\right) U_i^\dagger\right),\underset{[m(n+n')]\backslash T}{\tr}\left(\bigotimes_{i=1}^{m}U_i\left(\rhomm\otimes\ketbra{0}^{\otimes n'}\right) U_i^\dagger\right)\right)\leq \epsilon
\end{align*}
for $\epsilon\leq 2^{-\tilde{\Omega}(k\log n)}$. This is because suppose the above equation holds, according to the fact that $\chi_T$ are orthogonal under the null distribution, we have
\begin{align*}
\mathbb{E}_{\rho,\rho'\sim\mathcal{E}}\mathbb{E}_{\bm{s},\bm{s}'}\overline{D}_{\rho,m}(\bm{s})\overline{D}_{\rho',m}(\bm{s}')\chi_T(\bm{s})\chi_T(\bm{s}')\leq\sum_{\bm{s}[T],\bm{s'}[T]}\epsilon^2\chi_T(\bm{s})\chi_T(\bm{s}')&\leq 2^{2|T|}\epsilon^2
\end{align*}
and thus
\begin{align*}
\norm{\mathbb{E}_{\rho\sim\mathcal{E}}\overline{D}_{\rho,m}^{\leq k}-\textbf{1}}^2&=\sum_{\abs{T}\leq k,T\neq\emptyset}\mathbb{E}_{\rho,\rho'\sim\mathcal{E}}\mathbb{E}_{\bm{s},\bm{s}'}\overline{D}_{\rho,m}(\bm{s})\overline{D}_{\rho',m}(\bm{s}')\chi_T(\bm{s})\chi_T(\bm{s}')\leq k(m(n+n'))^k2^k\epsilon^2\leq o(1)
\end{align*}
as $n'\leq n$, which finishes the proof for \Cref{thm:general_ancilla}.

\section{Deferred proofs in frameworks for adaptive measurements}

\subsection{Proof of \texorpdfstring{\Cref{coro:general_adaptivity_within_block}}{Corollary 7.8}}\label{sec:pf_general_adaptivity_within_block}
By \Cref{lem:k_design_copy}, given an $\epsilon'$-approximate state $2km_0$ design with $0\leq\epsilon'\leq 1/2$, we have
\begin{align*}
\mathbb{E}_{\rho,\rho'\sim\mathcal{E}}\left(\expval{\overline{D}_{\rho}^{\leq \infty},\overline{D}_{\rho'}^{\leq \infty}}-1\right)^{2km_0}\leq O((2km_0)^22^{2km_0}(\epsilon'+2^{-n})).
\end{align*}
Given any $\bm{s}_{1:km_0-1}$, we have
\begin{align*}
\mathbb{E}_{\rho\sim\mathcal{E}}\left[\left(\overline{D}_{\rho}^{\leq\infty}(s)\right)^{2(km_0-1)}\right]&\leq(1+\epsilon')\frac{2^{2(km_0-1)n}(2(km_0-1))!}{2^n(2^n+1)...(2^n+2(km_0-1)-1)}\\
&\leq(1+\epsilon')(2(km_0-1))!.
\end{align*}
Therefore, we can choose
\begin{align*}
M=2(2(km_0-1))!=2^{O(km_0\log k\log m_0)}.
\end{align*}
By \Cref{thm:general_adaptivity_within_block}, we can thus choose
\begin{align*}
\epsilon&=2^{-\Omega(k^3m_0\log k\log m_1\log m_0\log M)}\\
&=2^{-\Omega(k^4m_0^2\log^2 k\log m_1\log^2 m_0)},
\end{align*}
which indicates that
\begin{align*}
\epsilon'=2^{-\Omega(k^4m_0^2\log^2 k\log m_1\log^2 m_0)}.
\end{align*}

\subsection{Proof of \texorpdfstring{\Cref{coro:general_adaptivity_among_block}}{Corollary 7.11}}\label{sec:pf_general_adaptivity_among_block}
By \Cref{lem:k_design_copy}, given an $\epsilon'$-approximate state $2k$ design with $0\leq\epsilon'\leq 1/2$, we have
\begin{align*}
\mathbb{E}_{\rho,\rho'\sim\mathcal{E}}\left(\expval{\overline{D}_{\rho}^{\leq \infty},\overline{D}_{\rho'}^{\leq \infty}}-1\right)^{2k}\leq O((2k)^22^{2k}(\epsilon'+2^{-n})).
\end{align*}
We also have
\begin{align*}
\mathbb{E}_{\rho\sim\mathcal{E}}\left[\left(\overline{D}_{\rho}^{\leq\infty}(s)-1\right)^{2(m_1-1)}\right]&\leq(1+\epsilon')\frac{2^{2(m_1-1)n}(2m_1)!}{2^n(2^n+1)...(2^n+2(m_1-1)-1)}\\
&\leq(1+\epsilon')(2m_1)!.
\end{align*}
Therefore, we can choose
\begin{align*}
M'=2(2m_1)!=2^{O(m_1\log m_1)}.
\end{align*}
By \Cref{thm:general_adaptivity_among_block}, we can thus choose
\begin{align*}
\epsilon&=2^{-\Omega(km_1\log k\log m_0\log M')}\\
&=2^{-\Omega(km_1^2\log^2 m_1\log m_0\log k)},
\end{align*}
which indicates that
\begin{align*}
\epsilon'=2^{-\Omega(km_1^2\log^2 m_1\log m_0\log k)}.
\end{align*}

\section{Auxiliary lemmas for quantum planted biclique}
\begin{fact}\label{fact:powerseries}
    Let $\beta_{d,s}$ be as defined in \eqref{eq:betadef}. Then
    \begin{align*}
        \beta_{d,s} = 1 + \sum^\infty_{a = 1} c_a(s) \cdot d^{-a}\,,
    \end{align*}
    for coefficients $c_a(s)$ which are degree-$2a$ polynomials in $s$ satisfying
    \begin{align*}
        |c_a(s)| \le (1+s)^{2a}
    \end{align*}
    for all $s \ge 0$.
\end{fact}

\begin{proof}
    Taylor expanding $\log(\beta_{d,s})$ and exponentiating, we find that
    \begin{align*}
        c_a(s) = (-1)^a\sum^a_{\ell = 1} \frac{1}{\ell!}\sum_{\substack{i_1 + \cdots + i_\ell = a + \ell \\ i_j \ge 2 \ \forall j}} \sigma_{i_1}(s) \cdots \sigma_{i_\ell}(s) \cdot \frac{1}{(i_1 - 1)\cdots (i_\ell - 1)}\,,
    \end{align*}
    where $\sigma_i(s) \triangleq 1^{i-1} + \cdots + (s-1)^{i-1}$. Note that $\sigma_i(s) \le s^i$, so the inner sum above can be bounded by $s^{a + \ell}\cdot \binom{a-1}{\ell-1}$. Finally, $\sum^a_{\ell = 1}\frac{1}{\ell!}\binom{a-1}{\ell-1} s^{a + \ell} \le (1 + s)^{2a}$, as claimed.
\end{proof}

\begin{lemma}\label{lem:derange}
    Let $\pi \in \mathcal{S}_W$ be a permutation with no fixed points. Then 
    \begin{align*}
        \E_{\bx}[|\bra{\psi_{\bx|_W}}P_\pi \ket{\psi_{\bx|_W}}|^2] \le d^{-|W|/2}\,.
    \end{align*}
\end{lemma}

\begin{proof}
    Let $C = (\bi,\pi(\bi),\pi^2(\bi),\ldots,\pi^{p+1}(\bi) = \bi)$ be any cycle in $\pi$. By assumption, $p > 1$. Then if the cycle length $p$ is even,
    \begin{align*}
        \E\prod^{p-1}_{\ell = 0} \braket{\psi_{x_{\pi^\ell(\bi)}}|\psi_{x_{\pi^{\ell+1}(\bi)}}}^2 \le \E\prod^{p/2-1}_{\ell = 0} \braket{\psi_{x_{\pi^{2\ell}(\bi)}} | \psi_{x_{\pi^{2\ell+1}(\bi)}}}^2 = (1/d)^{p/2}\,,
    \end{align*}
    where we used that for all $\bi\in [m]\times [n]$, $\E[\ket{\psi_{x_{\bi}}}\bra{\psi_{x_{\bi}}}] = \Id/d$.
    Similarly, if the cycle length $p$ is odd, then because $p > 1$,
    \begin{align}
        \E\prod^{p-1}_{\ell = 0} \braket{\psi_{x_{\pi^\ell(\bi)}}|\psi_{x_{\pi^{\ell+1}(\bi)}}}^2 &\le \E\prod^{p/2-1/2}_{\ell = 1} \braket{\psi_{x_{\pi^{2\ell}(\bi)}} | \psi_{x_{\pi^{2\ell+1}(\bi)}}}^2 \cdot \braket{\psi_{x_{\pi^p(\bi)}} | \psi_{x_i}}^2 \\
        &= (1/d)^{(p+1)/2} \cdot \E\norm{\psi_{x_i}}^4 = (1/d)^{(p+1)/2} \le (1/d)^{p/2}\,.
    \end{align}
    The claim follows as the sum of the lengths of the cycles in $\pi$ is equal to $|W|$.
\end{proof}

\end{document}